%% file: main.tex
\documentclass[11pt,letter]{article}
\usepackage{lineno}
\usepackage{amssymb}

\usepackage{fullpage}
\usepackage[margin = 2.54cm]{geometry}

\usepackage{microtype}
\usepackage{epsfig}
\usepackage{amssymb}
\usepackage{microtype}
\usepackage{wrapfig}
\usepackage{graphicx}
\usepackage{epsfig,amsmath,amsfonts,amssymb,amsthm,mathrsfs,ifpdf}
\usepackage{indentfirst,relsize}
\usepackage[numbers]{natbib}
\usepackage{setspace}
\usepackage{enumerate}
\usepackage{latexsym}
\usepackage{stackrel}
\usepackage[all]{xy}
\usepackage[usenames,dvipsnames]{pstricks}
\usepackage{pst-grad} 
\usepackage{pst-plot} 
\usepackage{xspace}
\usepackage{bbm}

\usepackage[super]{nth}
\usepackage[pagebackref=true, backref=true]{hyperref}
\usepackage[T1]{fontenc}

\usepackage[ruled,linesnumbered,vlined,noend]{algorithm2e}

\newcommand{\size}[1]{\left| #1 \right|}

\newcommand{\E}{\mathbb{E}}
\newcommand{\remove}[1]{}

\newcommand{\R}{\mathbb{R}}
\newcommand{\N}{\mathbb{N}}

\newcommand{\cM}{\mathcal{M}}
\newcommand{\cS}{\mathcal{S}}
\newcommand{\cT}{\mathcal{T}}
\newcommand{\cL}{\mathcal{L}}
\newcommand{\cA}{\mathcal{A}}

\newcommand{\cC}{\mathcal{C}}
\newcommand{\cD}{\mathcal{D}}
\newcommand{\cE}{\mathcal{E}}

\newcommand{\cP}{\mathcal{P}}

\newcommand{\cW}{\mathcal{W}}
\newcommand{\Oh}{\mathcal{O}}
\newcommand{\cQ}{\mathcal{Q}}
\newcommand{\tOh}{\widetilde{{\mathcal O}}}

\newcommand{\cX}{\mathcal{X}}
\newcommand{\cY}{\mathcal{Y}}

\newcommand{\eps}{\varepsilon}
\newcommand{\pr}{\mathrm{Pr}}

\newcommand{\complain}[1]{\textcolor{black}{#1}}

\remove{\usepackage[symbol]{footmisc}

}

\theoremstyle{plain}
\newtheorem{theo}{Theorem}[section]

\newtheorem{lem}[theo]{Lemma}
\newtheorem{pre}[theo]{Proposition}
\newtheorem{coro}[theo]{Corollary}

\newtheorem{cl}[theo]{Claim}
\theoremstyle{definition}
\newtheorem{defi}[theo]{Definition}
\newtheorem{rem}{Remark}
\newtheorem{obs}[theo]{Observation}

\input{macro.tex}

\title{{\bf Testing of Index-Invariant Properties in the\\
Huge Object Model}}

\author{}

\author{Sourav Chakraborty\footnote{Indian Statistical Institute, Kolkata, India. Email: chakraborty.sourav@gmail.com.}
\and
Eldar Fischer\footnote{Technion - Israel Institute of Technology, Israel. Email: eldar@cs.technion.ac.il. Research supported in part by an Israel Science
Foundation grant number 879/22. }
\and
Arijit Ghosh\footnote{Indian Statistical Institute, Kolkata, India. Email: arijitiitkgpster@gmail.com.}
\and
Gopinath Mishra\footnote{University of Warwick, UK. Email: gopianjan117@gmail.com. Research supported in part by the Centre for Discrete Mathematics and its Applications (DIMAP) and by EPSRC award EP/V01305X/1.}
\and
Sayantan Sen\footnote{Indian Statistical Institute, Kolkata, India. Email: sayantan789@gmail.com.}
}

\date{}

\begin{document}
  \maketitle
\newcommand{\mycite}[1]{{\cite{#1}}}

\pagenumbering{arabic} 



\input{abstract.tex}

\newpage

\tableofcontents

\newpage


\input{intro.tex}

\input{overview1}

\input{prelim_new}

\input{VC-dim-new}

\input{vc_cluster}

\input{vc-lb.tex}

\input{Lbound.tex}

\input{non-index-invariant}

\input{adaptive_ub}

\input{quadratic_lb_new.tex}

\input{quadratic_lb_adaptive_algo.tex}

\input{quadratic_nonadaptive_lb_new_3.tex}

\bibliographystyle{alpha}
\bibliography{reference}
\appendix

\input{appendix.tex}

\end{document}

%% file: macro.tex
\newcommand{\dhash}{D^{\#}}

\newcommand{\learn}{{\sc Test-and-Learn}\xspace}

\newcommand{\aprox}{{\sc Approx-Centers}\xspace}

\newcommand{\findpi}{{\sc Find-Permutation}\xspace}

\newcommand{\algadpt}{{\sc Alg-Adaptive}\xspace}

\newcommand{\vc}{\mathsf{vc}}

\newcommand{\fen}{\mathrm{FE}}
\newcommand{\sen}{\mathrm{SE}}
\newcommand{\gen}{\mathrm{GE}}

%% file: abstract.tex
\begin{abstract}
The study of distribution testing has become ubiquitous in the area of
property testing, both for its theoretical appeal, as well as for its
applications in other fields of Computer Science, and in various real-life statistical tasks.

The original distribution testing model relies on samples drawn
independently from the distribution to be tested. However, when testing
distributions over the $n$-dimensional Hamming cube $\left\{0,1\right\}^{n}$ for a large $n$, even reading a few samples is
infeasible. To address this, Goldreich and Ron [ITCS 2022] have defined a
model called the \emph{huge object model}, in which the samples may only be queried in a few places.

In this work, we initiate a study of a general class of properties in the huge object model,
those that are invariant under a permutation of the indices of the
vectors in $\left\{0,1\right\}^{n}$, while still not being necessarily fully symmetric as
per the definition used in traditional distribution testing.

We  prove that every index-invariant property satisfying a
bounded VC-dimension restriction admits a property tester with a number of
queries independent of $n$. To complement this result, we argue that satisfying only index-invariance or only a VC-dimension bound is insufficient to guarantee a tester whose query complexity is independent of $n$. Moreover, we prove that the dependency of sample and query complexities of our tester on the VC-dimension is essentially tight. As a second part of this work, we address the question of the
number of queries required for non-adaptive testing. We show that it can be at most quadratic in the number
of queries required for an adaptive tester in the case of index-invariant properties. This is in contrast with the tight (easily provable) exponential gap between adaptive and non-adaptive testers for general non-index-invariant properties. Finally, we provide an index-invariant property for which the quadratic gap between adaptive and non-adaptive query complexities for testing is almost tight.
\end{abstract}

%% file: intro.tex
\section{Introduction}

The field of distribution testing is currently ubiquitous in property testing, see the books and surveys of \mycite{DBLP:books/cu/Goldreich17, bhattacharyya2022property, DBLP:journals/eatcs/Fischer01,DBLP:journals/ftml/Ron08, DBLP:journals/fttcs/Ron09, DBLP:conf/propertytesting/CzumajS10,DBLP:journals/siamdm/RubinfeldS11,Canonne:Survey:ToC,CanonneTopicsDT2022} for reference.
Distribution testing has also found numerous applications in other areas of research, including topics that have real life applications~\cite{CM19,MPC20, DBLP:conf/nips/Canonne0S20,PM21}.

In the original model of distribution testing, a distribution $D$ defined over some set $\Omega$ can be accessed by obtaining independent samples from $D$, and the goal is to approximate various interesting properties of $D$. This model has been studied extensively over the last two decades, and many interesting results and techniques have emerged.

The majority of distribution testing research centers on the goal of minimizing the number of samples required to test for various properties of the underlying distribution.
If the domain of the distribution is structured (for example, if the domain is the $n$-dimensional Hamming cube $\{0,1\}^n$), then designing efficient testers brings its own challenges. A number of papers have studied the problem of testing properties of distributions defined over the $n$-dimensional Hamming cube (see \cite{DBLP:conf/colt/AliakbarpourBR16,DBLP:conf/colt/CanonneDKS17,DBLP:journals/toct/BhattacharyyaC18,BGMV20,DBLP:conf/soda/CanonneCKLW21,chen2021learning,DBLP:journals/corr/abs-2204-08690}). With the rise of big data (translating to $n$ being very large), even reading all the bits in the representation of the samples might be very expensive. To address this issue, recently Goldreich and Ron~\cite{GoldreichR21a} studied distribution testing in a different setting. 

In their model, called the \emph{huge object model}, the distribution $D$ is supported over the $n$-dimensional Hamming cube $\{0,1\}^n$, and the tester will obtain $n$-length Boolean strings as samples. However, as reading the sampled strings in their entirety might be infeasible when $n$ is large, the authors in \cite{GoldreichR21a} considered query access to the samples along with standard sampling access. {Note that without loss of generality, the number of samples will be upper-bounded by the number of queries.} Thus, a desirable goal in this model is to optimize the number of queries for testing a given property, with respect to the Earth Mover Distance notion that befits this model. \cite{GoldreichR21a} studied various natural properties like support size estimation, uniformity, identity, equality, and ``grainedness''~\footnote{A distribution $D$ over $\{0,1\}^n$ is said to be $m$-grained if the probability mass of any element in its support is a multiple of $1/m$, where $m \in \N$.} in this model, providing upper and lower bounds on the sample and query complexities for these properties.


In this paper, we study the sample and query complexities of a very natural class of properties, which we call the \emph{index-invariant properties}, in the huge object model of distribution testing.

\paragraph{Index-Invariant Distribution Properties:}

In general, a distribution property is a collection of distributions over a fixed domain $\Omega$~\footnote{We use the phrases ``a distribution is in the property'' and ``a distribution has the property'' interchangeably to mean the same thing.}. Often the property in question has some other ``symmetry''. For example, a property is called \emph{label-invariant} if any changes in the labels of the domain do not affect whether the distribution is in the property or not. 
Many of the well studied properties, such as uniformity, entropy estimation, support size estimation, and grainedness,  are label-invariant properties.  Label-invariant properties have been studied extensively in literature~\cite{batu2005complexity, paninski2008coincidence,goldreich2011testing, valiant2011testing, diakonikolas2014testing,chan2014optimal,acharya2015optimal, valiant2017estimating,batu2017generalized, diakonikolas2017sharp}.

In some cases, the distribution property is not fully label-invariant, but still has a certain amount of symmetry. For illustration, consider the following  examples: 

\begin{enumerate}
    \item {\bf Property {{\sc Monotone}}:} Any distribution $D$ over $\{0,1\}^n$ satisfies the  $\mbox{{{\sc Monotone}}} $ property if
$$\mbox{$\mathbf{X} \preceq \mathbf{Y}$ implies $D(\mathbf{X}) \leq D(\mathbf{Y})$, for any $\mathbf{X},\mathbf{Y} \in \{0,1\}^n$},$$
where for two vectors $\mathbf{X},\mathbf{Y} \in \{0,1\}^n$,  $\mathbf{X} \preceq \mathbf{Y}$ if $x_i \leq y_i$ holds for every $i \in [n]$.

\item {\bf Property {\sc Log-Super-Modularity}:} Any distribution $D$ over $\{0,1\}^n$ satisfies the property $\mbox{{{\sc Log-Super-}}}$ $\mbox{{{\sc Modularity}}}$ if
$$\mbox{$D(\mathbf{U})D(\mathbf{V}) \leq D(\mathbf{U} \wedge \mathbf{V})D(\mathbf{U} \vee \mathbf{V})$, for any $\mathbf{U},\mathbf{V} \in \{0,1\}^n$},$$ where the Boolean $\wedge$ and $\vee$ operations over the vectors are performed coordinate-wise.

\color{black}

\item {\bf Property {\sc Low-affine-dimension}:} A distribution $D$ over $\{0,1\}^n$ is said to satisfy the $\mbox{{{\sc Low-affine-dimension}}}$ property, with parameter $d \in \N$, if the {\em affine dimension}\footnote{A set $S \subseteq \R^{n}$ has {\em affine dimension} $k$ if the dimension of the smallest {\em affine set} in $\R^{n}$ that contains $S$ is $k$.} 
of the support of $D$ is at most $d$.
\end{enumerate}

Note that for the properties described above, a distribution satisfies the above properties even after the indices $\{1,\ldots,n\}$ of the vectors in $\{0,1\}^n$ are permuted by a permutation $\sigma$ defined over $[n]$. To capture this structure in the properties, we introduce the notion of \emph{index-invariant} properties. 

\color{black}
\begin{defi}[{\bf Index-invariant property}]
Let us assume that $D:\{0,1\}^n \rightarrow [0,1]$ is a distribution over the $n$-dimensional Hamming cube $\{0,1\}^n$. For any permutation $\sigma:[n]\rightarrow [n]$, let $D_\sigma$ be the distribution such that $D(w_1,\ldots,w_n)=D_\sigma(w_{\sigma(1)},\ldots,w_{\sigma(n)})$ for all $(w_1,\ldots,w_n)\in \{0,1\}^n$.  A distribution property ${\cal P}$ is said to be \emph{index-invariant} when $D$ is in ${\cal P}$ if and only if $D_\sigma$ is in ${\cal P}$, for any distribution $D$ and any permutation $\sigma$. 
\end{defi}

Informally speaking, index-invariant properties refer to those properties that are invariant under the permutations of the indices $\{1,\ldots,n\}$. 
Note that this set of properties differs from the more common notion of label-invariant properties, since the total number of possible labels, for distributions over all $n$-length Boolean vectors, is $2^{n}$. However, we are considering only permutations over $[n]$, thus in total only $n!$ permutations instead of $2^{n}!$ permutations.

\subsection{Our Results}

In this paper, as already mentioned, we study the sample and query complexities (in the huge object model) of index-invariant properties. We primarily focus on two problems. First, we study the connection between the query complexity for testing an index-invariant property and the VC-dimension of the non-trivial support of the distributions in the property. Secondly, we study the relationship between the query complexities of the adaptive and non-adaptive testers for index-invariant properties, along with their non-index-invariant counterparts.

One important and technical difference between the huge object model and the standard distribution property testing model is the use of \emph{Earth Mover Distance} ($\mbox{EMD}$) for the notion of ``closeness'' and ``farness'', instead of the more prevalent $\ell_1$ or variation distance. Thus, in the rest of the paper, by an $\eps$-tester for any property $\cP$ of distributions over $\{0,1\}^n$, we mean an algorithm that given sample and query access (to the bits of the sampled vectors) to a distribution distinguishes (with probability at least $2/3$) the case where the distribution $D$ is in the property $\cP$ from the case where the EMD of $D$ from any distribution in $\cP$ is at least $\eps$, where $\eps >0$ is a proximity parameter.

\subsection*{Testing by learning of bounded VC-dimension properties (constant query testable properties):}


We prove that a large class of distribution properties are all testable with a number of queries independent of $n$, using the \emph{testing by learning paradigm}~\cite{diakonikolas2007testing, DBLP:conf/icalp/GopalanOSSW09,servedio2010testing}, where the distributions are supported over the $n$-dimensional Hamming cube $\{0,1\}^n$. More specifically, we prove that every distribution whose support has a bounded VC-dimension can be \emph{efficiently} learnt up to a permutation, leading to efficient testers for index-invariant distribution properties that admit a global VC-dimension bound. Our main result regarding the learning of distributions in the huge object model is the following theorem.

\begin{theo}[{\bf Main learning result (informal)}]\label{theo:main}
For any fixed constant $d \in \N$, given sample and query access to an unknown distribution $D$ over $\{0,1\}^n$  and a proximity parameter $ \eps > 0$, there exists an algorithm that makes $\mathrm{poly}(\frac{1}{\eps})$ queries~\footnote{The degree of the polynomial in $\frac{1}{\eps}$ depends on the parameter $d$.}, and either outputs the full description of a distribution or {\sc fail} satisfying the following conditions:

\begin{enumerate}
    \item[(i)] If the support of $D$ is of VC-dimension at most $d$, then with probability at least $2/3$, the algorithm outputs a full description of a distribution $D'$ such that $D$ is $\eps$-close to $D'_\sigma$ for some permutation $\sigma:[n]\rightarrow [n]$.
    
    \item[(ii)]  For any $D$, the algorithm will not output a distribution $D'$ such that $D'_\sigma$ is $\eps$-far from $D$  for all permutations $\sigma:[n]\rightarrow [n]$, with probability more than $1/3$. However, if the VC-dimension of the support of $D$ is more than $d$, the algorithm may output {\sc Fail} with any probability.
\end{enumerate}
\end{theo}


In fact, our result holds for a general class of \emph{clusterable} properties (stated in Theorem~\ref{theo:mainthm1} and Corollary~\ref{coro:mainthm1}) that also covers the VC-dimension case as stated in the above theorem. Note that the above theorem corresponds to the learnability of any distribution when the VC-dimension of its support is bounded. 
As a corollary, it implies that any index-invariant distribution property admitting a global VC-dimension bound is testable with a constant number of queries, depending only on the proximity parameter $\eps$ and the VC-dimension $d$.
The corollary is stated as follows: 

\begin{coro}[{\bf Testing (informal)}]\label{coro:testvc}
Let $\cP$  be an index-invariant property such that any distribution $D\in \cP$ has VC-dimension at most $d$, where $d$ is some constant. There exists an algorithm, that has sample and query access to an unknown distribution $D$ over $\{0,1\}^n$, takes a proximity parameter $\eps > 0$, and distinguishes whether $D \in \cP$ or $D$ is $\eps$-far from $\cP$ with probability at least $2/3$, by making only $\mathrm{poly}(\frac{1}{\eps})$  queries.
\end{coro}


It turns out that our tester for testing VC-dimension property takes $\exp(d)$ samples, and performs $\exp(\exp(d))$ queries for VC-dimension $d$. We show that this bound is tight, in the sense that there exists an index-invariant property with VC-dimension $d$ such that any tester for the property requires an exponential number of samples and a doubly-exponential number of queries on $d$.

\begin{theo}\label{theo:lb-vc_intro}
Let $d,n \in \N$. There exists an index-invariant property $\cP_{\vc}$ with VC-dimension at most $d$ such that any (non-adaptive) tester for $\cP_{\vc}$ requires $2^{\Omega(d)}$ samples and $2^{2^{d - \Oh(1)}}$ queries.
\end{theo}

\color{black}
Note that from a result in \cite{GoldreichR21a}, it follows that there exists an index-invariant property $\cP$ such that any distribution $D \in \cP$ has VC-dimension $d$ and any algorithm that has sample access to a distribution $D$ over $\{0,1\}^n$ requires $\Omega(2^d/ d)$ samples~\footnote{Let $\cP$ be the distribution property of having support size at most $2^d$. Note that the VC-dimension of any member of $\cP$ is at most $d$. By \cite{GoldreichR21a}, for any small enough $\eps$, an $\eps$-test for this property requires at least $\Omega\left(2^d/d\right)$ samples.}, but Theorem~\ref{theo:lb-vc_intro} proves the lower bound on both sample and query complexities for the same property.

Theorem~\ref{theo:main} assumes that the properties are index-invariant and have bounded VC-dimension. A natural question in this regard is whether the bounded VC-dimension and index-invariance assumptions are necessary for a property to be constantly testable. We answer this question in the negative. Theorem~\ref{theo:lb-vc_intro} implies that bounded VC-dimension is necessary for a property to be constantly testable even if the property is index-invariant. The following proposition rules out the possibility that only the {bounded} VC-dimension assumption is good enough for a property to be testable by making a constant number of queries. 

\begin{pre}[{\bf Necessity of index-invariance (informal)}]\label{theo:non-index}
There exists a non-index-invariant property $\cP$ such that any distribution $D \in \cP$ has VC-dimension $O(1)$ and the following holds. There exists a fixed $\eps>0$, such that distinguishing whether $D \in \cP$ or $D$ is $\eps$-far from $\cP$ requires $\Omega(n)$  queries, where the distributions in the property $\cP$ are defined over the $n$-dimensional Hamming cube $\{0,1\}^n$.
\end{pre}
The above proposition is formally stated and proved at the end of Subsection~\ref{sec:lbgen}.

\remove{However, even for simple tasks, like distinguishing between two given distributions, we show that the number of queries required (in the huge object model) is a function of the VC-dimension of the distributions.

\begin{pre}[{\bf Necessity of bounded VC-dimension to test with bounded number of queries (informal)}]\label{cor:lb} 
There exist two distributions  $D^{yes}$ and $D^{no}$ such that the following conditions hold:
\begin{description}
\item[(i)] The VC-dimension of both distributions $D^{yes}$ and $D^{no}$ is at most $d$;

\item[(ii)] For any permutation $\sigma:[n]\rightarrow [n]$, the EMD between $D^{yes}$ and $D^{no}_\sigma$  is at least $\frac{1}{10}$;

\item[(iii)] Any algorithm, that has sample and query access to either $D^{yes}$ or $D^{no}$, and distinguishes whether it is $D^{yes}$ or $D^{no}$, must perform $\Omega(d)$ queries.
\end{description}

\end{pre}
}

\color{black}

\subsection*{Separation between adaptive and non-adaptive testers:}


Until now, all the upper bounds that we have discussed are designed for non-adaptive algorithms.
The question how adaptivity helps in designing efficient testers is interesting in its own right.
In the standard model of distribution testing, since the model is inherently non-adaptive, there is essentially no gap between adaptive and non-adaptive testers. However, in the related model of conditional sampling of distributions~\cite{DBLP:journals/siamcomp/ChakrabortyFGM16, DBLP:journals/siamcomp/CanonneRS15}, there is a super-exponential separation (constant vs. $\mbox{poly}(\log n)$) between the complexities of these two types of testers~\cite{acharya2018chasm}.

In the context of graph testing in the dense graph model, it is known that the gap between the query complexities of adaptive and non-adaptive algorithms is at most quadratic~\cite{goldreich2003three}, which has recently been proved to be tight~\cite{goldreich2022non}. {However, for bounded-degree graphs, the gap between the query complexities for some properties is constant vs. $\Omega(\sqrt{n})$, where $n$ denotes the number of vertices of the graph~\cite{goldreich1997property}}. For testing of functions, there is an exponential separation between the complexity of these two types of testers~\cite{ron2015exponentially}.

Thus, a natural question to study in this huge object model is about the gap between the query complexities of non-adaptive and adaptive algorithms. When considering general properties, there can be an exponential gap in the query complexities between non-adaptive and adaptive testers (see Theorem~\ref{lem:nonindexlb}).

However, for index-invariant properties, this gap can be at most quadratic, as stated in the following theorem.

\color{black}

\begin{theo}[{\bf Connection between adaptive and non-adaptive testers for index-invariant properties}]\label{theo:lb-main}
For any index-invariant property $\cP$, there is at most a quadratic gap between the query complexities of adaptive and non-adaptive testers. 
\end{theo}


We also prove that the above gap is almost tight, in the sense that there exists an index-invariant property which can be $\eps$-tested using $\widetilde{\Oh}(n)$ adaptive queries, while $\widetilde{\Omega}(n^2)$ non-adaptive queries are required to $\eps$-test it. 

\begin{theo}[{\bf Near-tightness of the connection between adaptive and non-adaptive testers for index-invariant properties}]\label{theo:lb-main_new_intro}
There exists an index-invariant property $\cP_{\mathrm{Gap}}$ that can be $\eps$-tested adaptively using $\tOh(n)$ queries for any $\eps \in (0,1)$, while there exists an $\eps \in (0,1)$ for which $\widetilde{\Omega}(n^{2})$ queries are necessary for any non-adaptive $\eps$-tester.
\end{theo}


\subsection*{Using EMD as the distance metric in conjunction with the notion of index-invariance:}

Recall that here we will use the Earth Mover Distance (EMD) as the distance metric defining $\eps$-testing, in contrast to the stronger variation distance, the commonly studied distance measure in distribution testing literature. As discussed in \cite{GoldreichR21a}, this is essential when we restrict ourselves to querying the samples obtained from the distribution. To illustrate this, consider two (say very sparse) distributions $D_1$ and $D_2$ whose supports are disjoint, yet admit a bijection such that every string from $\mbox{Supp}(D_1)$ is mapped to a string  from $\mbox{Supp}(D_2)$ that is very close to it in terms of the Hamming distance. The variation distance between $D_1$ and $D_2$ would be large, and yet we would not be able to distinguish the two distributions without querying some samples in their entirety, that is, without using $\Theta(n)$ queries per sample. The EMD metric is the one incorporating the Hamming distance between strings (which comes to play when we are not performing many queries to the samples) into the notion of variation distance.

Another question involves what general statements can be said about testers in this model. If we do not restrict ourselves to properties satisfying any sort of invariance, then very little can be proved on testers in general, just as is the case with general string property testing under the Hamming distance (in fact, string testing can be reduced to testing in the huge object model~\footnote{We will use this reduction for proving exponential separation between adaptive and non-adaptive testers for non-index-invariant properties (see Subsection~\ref{sec:lbgen}).}). On the other hand, if we were to restrict ourselves to label-invariant properties only, it would appear that we lose much of the rich structure offered by the ability to define distributions over strings. We believe that index-invariance is a natural middle-of-the-road restriction for the formulation of general statements about testing in the huge object model.

\color{black}

\subsection{Organization of the paper}

In Section~\ref{sec:overview}, we present a brief overview of our results. We present the related definitions in the preliminaries section (Section~\ref{sec:prelim}). We present the results about learning and testing clusterable distributions in Section~\ref{sec:learncluster}. After that, in Section~\ref{sec:vc}, we move on to present algorithms for testing properties with bounded VC-dimension. We present lower bound results for bounded VC-dimension testing in Section~\ref{sec:bounded-vc}. Later, we show the tight exponential separation between the query complexities of adaptive and non-adaptive algorithms for non-index-invariant properties in Section~\ref{sec:exponentialgap}. In Subsection~\ref{sec:quadratic_gap2}, we prove that for index-invariant properties, there is at most a quadratic gap between the query complexities of adaptive and non-adaptive testers. Finally, in Subsection~\ref{sec:quadratic_gap}, Subsection~\ref{sec:quadratic_determine_permutation}, Subsection~\ref{sec:quadratic_adaptive_ub}, and in Subsection~\ref{sec:quadratic_nonadapt_lb}, we prove that the quadratic gap between adaptive and non-adaptive testers for index-invariant properties is almost tight, ignoring poly-logarithmic factors. In Appendix~\ref{sec:prelim_prob}, we state some useful concentration inequalities that are used in our proofs.

%% file: overview1.tex
\section{Technical overview of our results}\label{sec:overview}

In this section, we provide a brief overview of our results. We start by explaining our upper bounds.

\subsection{Overview of our upper bound results for index-invariant bounded VC-dimension property}

In our main upper bound result, we prove a learning result for a general class of distributions that covers the case of learning distributions with bounded VC-dimension. We say that a distribution $D$ is $(\zeta,\delta,r)$-clusterable if we can partition the $n$-dimensional Hamming cube $\{0,1\}^n$ into $r+1$ parts $\cC_0,\ldots,\cC_r$, such that $D(\cC_0) \leq \zeta$ and the diameter of $\cC_i$ is at most $\delta$ for every $i \in [r]$ (see Definition~\ref{defi:cluster}). The main upper bound result (Theorem~\ref{theo:mainthm1}), that leads to Theorem~\ref{theo:main}, is the design of an algorithm for learning a $(\zeta,\delta,r)$-clusterable distribution up to permutations. That is, given sample and query access to a $(\zeta,\delta,r)$-clusterable distribution, we want to output a distribution $D'$ such that the Earth Mover Distance between $D$ and $D_\sigma'$ is small for some permutation $\sigma:[n]\rightarrow [n]$, by performing number of queries independent of $n$.

\paragraph*{The idea of learning $(\zeta,\delta,r)$-clusterable distributions:}
The formal algorithm is presented in Algorithm~\ref{alg:testcluster--} in Section~\ref{sec:learncluster} as \learn. The algorithm starts by taking $t_1=\Oh(\frac{r}{\zeta} \log \frac{r}{\zeta})$ samples from the input distribution $D$. Let us denote them as $\cS=\{\mathbf{X}_1,\ldots,\mathbf{X}_{t_1}\}$. If $D$ is $(\zeta,\delta,r)$-clusterable, consider its clusters $\cC_0,\ldots,\cC_r$ as described above. We say that a cluster $C_i$ is \emph{large} if the probability mass of $\cC_i$ is more than $\frac{\zeta}{10 r}$, that is, $D(\cC_i) \geq \frac{\zeta}{10 r}$. As the size of $\cS$ is sufficiently large, we know that $\cS$ intersects every large cluster with probability at least $99/100$ (see Lemma~\ref{lem:hit}). In order to estimate the masses of $\cC_i$, for each $i \in [t_1]$, we take another set of random samples $\cT=\{\mathbf{Y}_1,\ldots,\mathbf{Y}_{t_2}\}$ from $D$ where $t_2=\Oh(\frac{t_1^2}{\zeta^2} \log t_1)$, and assign each of the vectors in $\cT$ to some vector in $\cS$ depending on their Hamming distance. However, since computing the exact distances between the vectors in $\cS$ and $\cT$ requires $\Omega(n)$ queries, we use random sampling.

We take a random set of indices $R \subset [n]$ of suitable size, and project the vectors in $\cS$ and $\cT$ on $R$ to estimate their pairwise distances up to an additive factor of $\delta$. $R$ not only preserves the distances between all pairs of vectors between $\cS$ and $\cT$, but also the distances of a large fraction of the vectors in $\{0,1\}^n$ from all the vectors in $\cS$ (see Lemma~\ref{lem:dist}). Based on the estimated distances, we assign each 
vector of $\mathbf{T} \in \cT$ to a vector in $\mathbf{S} \in \cS$ such that the projected 
distance between them is at most $2\delta$. If there exists no such vector in $\cS$ corresponding to a vector $\mathbf{T}\in \cT$, then the 
vector $\mathbf{T}$ remains unassigned. Let us denote the fraction of vectors in $\cT$ that are assigned to $\mathbf{X}_i$ as $w_i$, for every $i \in [t_1]$. Let $w_0$ be the fraction of vectors in $\cT$ that are not assigned to any vector in $\cS$. If $D$ is $(\zeta, \delta,r)$-clusterable, then $w_0\leq 3\zeta$ holds 
with high probability. These $w_i$'s preserve the weights of some approximate clustering 
(which may not be the original one from	which we started, but is close to it in some sense), see Lemma~\ref{lem:mass} for the details.

Consider a distribution $D^*$ supported over $\cS$ such that 
$D(\mathbf{X}_i)\geq w_i$ for every $i \in [t_1]$. Using a number of technical 
lemmas, we prove that the EMD between $D$ and $D^*$ is \emph{small}. Note that we still can not 
report $D^*$ as the output distribution, since to do so, we need to know the exact vectors in $\cS$, which requires $\Omega(n)$ queries. To bypass this barrier, we use the provision that we are allowed to output any permutation of the distribution.  More specifically, we construct vectors $\mathbf{S}_1,\ldots,  \mathbf{S}_{t_1} \in \{0,1\}^n$ such that $d_H(\mathbf{X}_i,\sigma(\mathbf{S}_i))$ is small for every $i \in [t_1]$ and some permutation $\sigma:[n]\rightarrow [n]$. This is possible using the projections of the vectors in $\cS$ to the random set of indices $R$ for estimating the number of indices of each ``type'' with respect to $\mathcal{S}$ (see Lemma~\ref{lem:struct}). {Finally, we output the distribution $D'$ supported over the newly constructed vectors $\mathbf{S}_1,\ldots,\mathbf{S}_{t_1}$ such that $D'(\mathbf{S}_i)= D^{*}(\mathbf{X}_i)$ for every $i \in [t_1]$. The guarantee on the Hamming distance between $\mathbf{X}_i$ and $\sigma(\mathbf{S}_i)$ provides a bound on the EMD between $D'_{\sigma}$ and $D^*$, and with the above mentioned EMD bound between $D^*$ and $D$, we are done. To keep the discussion simple, we will not explain here the idea of the proof of Theorem~\ref{theo:main}(ii), which relies on a sort of converse to the above method of approximating cluster weights.}


\paragraph*{How learning $(\zeta,\delta,r)$-clusterable distribution implies Theorem~\ref{theo:main}:}

Let us define a distribution to be $(\alpha,r)$-clusterable if it is $(0,\alpha,r)$-clusterable. The learning of $(\zeta,\delta,r)$-clusterable distribution implies a learning result for any distribution that is close to being $(\alpha,r)$-clusterable (see Corollary~\ref{coro:mainthm1}) due to a technical lemma (see Lemma~\ref{lem:cluster12}). If the support of a distribution has bounded VC-dimension, using standard results in VC theory, we can show that it is also $(\alpha,r)$-clusterable, where $r$ is a function of $\alpha$ and $d$. Thus the learning result of $(\alpha,r)$-clusterable distributions implies a result allowing the learning of distributions with bounded VC-dimension.

\subsection{Overview of the lower bound result for index-invariant bounded VC-dimension properties}
To prove Theorem~\ref{theo:lb-vc_intro}, let us define the property $\cP_{\vc}$. Let $k=2^d$, $\ell=2^{2^{d-10}}$ and $\ell'=2^{2^{d-20}}$. Consider a matrix $A$ of dimension $k \times \ell$ whose column vectors are $1/3$-far from each other. Let $\mathbf{V}_1,\ldots,\mathbf{V}_k \in \{0,1\}^n$ be $k$ vectors that are formed by blowing up the row vectors of $A$ in $\{0,1\}^{\ell}$ to $\{0,1\}^n$ by repeating each bit of the vectors $n/\ell$ times, and $D_A$ be the uniform distribution over the support $\{\mathbf{V}_1,\ldots,\mathbf{V}_k\}$. Our property $\cP_{\vc}$ is the collection of all distribution that can be obtained from $D_A$ by permuting the indices. Let $D_{yes}$ be the distribution obtained from $D_A$ by randomly permuting the indices. Note that $D_{yes} \in \cP_{\vc}$. As the support size of any distribution in  $\cP_{\vc}$ is at most $2^d$, the VC-dimension of $\cP_{\vc}$ is at most $d$.

To prove the lower bound on the query complexity, let us define the distribution $D_{no}$. Let us take $\ell'$  columns of $A$ uniformly at random to form a matrix $B$ of dimension $k \times \ell'$, and $\mathbf{W}_1,\ldots,\mathbf{W}_k \in \{0,1\}^n$ be $k$ vectors that are formed by blowing up the row vectors of $B$ in $\{0,1\}^{\ell'}$ to $\{0,1\}^n$ by repeating each bit of the vectors $n/\ell'$ times. Let $D_B$ be the uniform distribution over the support $\{\mathbf{W}'_1,\ldots,\mathbf{W}'_k\}$. $D_{no}$ is the distribution obtained from $D_{B}$ by randomly permuting the indices. We show that the Earth Mover Distance between $D_{no}$ and any distribution in $\cP_{\vc}$ is at least $1/8$ (see Lemma~\ref{cl:emd-vc-low}). Observe that $D_{yes}$ divides the index set $[n]$ into $\ell$ equivalence classes and $D_{no}$ divides the index set into $\ell'$ equivalence classes. The query complexity lower bound follows from the fact that, unless we query $2^{2^{d-\Oh(1)}}$ indices, we do not hit two indices from the same equivalence class, irrespective of whether the distribution is $D_{yes}$ or $D_{no}$ (see Lemma~\ref{cl:main-vc-low}).

To prove the lower bound on the sample complexity, let us define another distribution $D'_{no}$. Let us take $k'=2^{d-20}$  rows of $A$ uniformly at random to form a matrix $B'$ of dimension $k' \times \ell$. Let $\mathbf{W}'_1,\ldots,\mathbf{W}'_{k'} \in \{0,1\}^n$ be $k'$ vectors that are formed by blowing up the row vectors of $B'$ in $\{0,1\}^{\ell}$ to $\{0,1\}^n$ by repeating each bit of the vectors $n/\ell$ times. Let $D_{B'}$ be the uniform distribution with support $\{\mathbf{W}'_1,\ldots,\mathbf{W}'_{k'}\}$. $D'_{no}$ is the distribution obtained from $D_{B'}$ by randomly permuting the indices. We show that the Earth Mover Distance between $D'_{no}$ and any distribution in $\cP_{\vc}$ is at least $1/8$ (see Lemma~\ref{cl:emd-vc-low_sample}). The sample complexity lower bound follows from the fact that, unless we take $2^{\Omega(d)}$ samples, all the samples are distinct with probability $1-o(1)$, irrespective of whether the distribution is $D_{yes}$ or $D_{no}$ (see Lemma~\ref{lem:sample}).

\color{black}

\subsection{Overview of adaptive vs non-adaptive query complexity results}

Finally, we explore the relationship between adaptive and non-adaptive testers in the huge object model. It turns out that there is a tight (easy to prove) exponential separation between the query complexities of adaptive and non-adaptive testers for non-index-invariant properties. Roughly, the simulation of an adaptive algorithm by a non-adaptive one follows from unrolling the decision tree of the adaptive algorithm. This is formally proved in Lemma~\ref{pre:adaptiveexp}. Moreover, we show that this separation is tight. For this purpose, we consider a property of strings $\cP_{Pal}$, which exhibits an exponential gap between adaptive and non-adaptive testing in the string testing model. We show how to transform a string property $\cP$ to a distribution property $1_{\cP}$ such that the query bounds on adaptive and non-adaptive testing carry over. Thus, the separation result between adaptive and non-adaptive algorithms for $\cP_{Pal}$ carries over to $1_{\cP_{Pal}}$ (see Theorem~\ref{lem:nonindexlb}). This technique, employed for a maximally hard to test string property, is also used for proving Proposition~\ref{theo:non-index}.

In contrast to the non-index-invariant properties, we prove that there can be at most a quadratic gap between the query complexities of adaptive and non-adaptive algorithms for testing index-invariant properties. The proof is very close in spirit to the proof of the quadratic relation between adaptive and non-adaptive testing of graphs in the dense model~\cite{goldreich2003three}. Given an adaptive algorithm $\cA$ with sample complexity $s$ and query complexity $q$, the main idea is to first simulate a \emph{semi-adaptive} algorithm $\cA'$ that queries $q$ indices from each of the $s$ samples and decides accordingly. Note that the sample complexity of $\cA'$ remains $s$, whereas the query complexity becomes $qs$. Once we have the semi-adaptive algorithm $\cA'$, we now simulate a non-adaptive algorithm $\cA''$. As the property we are testing is index-invariant, we can first apply a uniformly random permutation $\sigma$ over $[n]$, and then run the semi-adaptive algorithm $\cA'$ over $D_{\sigma}$ instead of $D$, where $D$ is the input distribution to be tested. This makes the tester completely non-adaptive. Its correctness follows from the index-invariance of the property we are testing.



\paragraph*{Quadratic separation between adaptive and non-adaptive testers:}

Before proceeding to present an overview of our quadratic separation result, let us first recall the support estimation result of Valiant and Valiant~\cite{valiant2010clt}, which will be crucially used in our proof. Roughly speaking, the result states that in the standard sampling model, given a distribution $D$ over $[2n]$, in order to distinguish whether $D$ has support size at most $n$, or $D$ is far from all distributions with support size at least $n$, $\Theta(n/\log n)$ samples are required.

\begin{theo}[Support Estimation bound, restatement of Corollary $9$ of Valiant-Valiant~\cite{valiant2010clt}]\label{theo:support_lb_overview}
Given a distribution $D$ over $[2n]$, that can be accessed via independent samples and a proximity parameter $\eps \in (0,1/8)$, in order to distinguish, with probability at least $\frac{3}{4}$, whether $D$ has support size at most $n$ or $D$ has at least $(1+\eps)n$ elements with non-negligible weights in its support, $\Theta(\frac{n}{\log n})$ samples from $D$ are necessary and sufficient.
\end{theo}


To construct the index-invariant property that provides a quadratic separation between the query complexities of adaptive and non-adaptive testers, we will use the above result. Let $D^{\mathrm{Supp}}_{yes}$ and $D^{\mathrm{Supp}}_{no}$ be the pair of hard distributions corresponding to the support estimation lower bound, from which we define our pair $D_{yes}$ and $D_{no}$ of hard distributions for our property. We will construct a huge object distribution property over a slightly larger domain $[N]$ with $N= \Oh(n\log n)$, where we will encode the elements of the support of the distributions $D_{yes}^{\mathrm{Supp}}$ and $D_{no}^{\mathrm{Supp}}$. Additionally, we will include a set of ``ordering'' vectors to both $D_{yes}$ and $D_{no}$ that encode a permutation $\sigma:[n] \rightarrow [n]$. Our property will be defined as a permutation of a non-index-invariant property along with an encoding of the permutation itself.


For the non-index-invariant property, we use an encoding of the elements of $\{0,1\}^n$ that can be decoded only if a sample from a family of special small sets is read in its entirety. For constructing hard distributions, we consider (encodings of) $2n$ special elements of $\{0,1\}^n$, and use over them the hard distributions corresponding to Theorem~\ref{theo:support_lb_overview}.

The encoding vectors of $D_{yes}$ and $D_{no}$ are designed in such a fashion, that if we can know the index ordering (and thus the identity of the above mentioned small sets), the support size estimation problem becomes relatively easy. However, without knowing the ordering vectors, estimating the size of the support becomes harder. More specifically, {if we already know the index ordering}, then support size estimation can be done using $\mbox{poly}(\log n)$ queries from each sample, over the $\widetilde{\Oh}(n)$ samples that are {sufficient} for solving the support estimation problem.


On the other hand, an important feature of our property ensures that unless some of the special sets are successfully hit while querying a sampled vector, which is a low probability event without prior knowledge of the encoded index ordering unless we perform $\widetilde{\Omega}(n)$ queries to that vector, then the queries do not provide any useful information about the sampled vector to the tester. This is achieved by the encoding procedure of the vectors, which is motivated from~\cite{DBLP:conf/innovations/Ben-EliezerFLR20}. 
However, it is not deployed here the same way as~\cite{DBLP:conf/innovations/Ben-EliezerFLR20}, since the surrounding proofs here are quite different (as well as the end-goal).

Since an adaptive algorithm can first learn the ordering vectors by performing $\widetilde{\Oh}(n)$ queries (as it takes $\mbox{poly}(\log n)$ samples to hit all the order encoding vectors), the adaptive tester requires $\widetilde{\Oh}(n)$ queries in total. However, for non-adaptive testers, since we have to perform all queries simultaneously, the tester would have to make $\widetilde{\Omega}(n)$ queries to each sampled vector to be able to utilize the support estimation procedure (since as explained above, fewer queries would give no useful information about the sample to the tester). As a result, $\widetilde{\Omega}(n^2)$ non-adaptive queries are required {following the lower bound result in Theorem~\ref{theo:support_lb_overview}}.

Another technical challenge is to construct the property in such a fashion that allows the crafting of ``wrong distributions'' which remain far from the property, even if we permute the support vectors. This is due to the fact that just replacing the vectors defining the index ordering does not require a change of large Earth Mover Distance. Thus we need the distributions to remain far from the property even if we reorder them. We ensure this by designing the hard distributions such that the support vectors of the distributions are far from each other. This in turn allows us to prove that the distribution $D_{no}$ will remain far from the property, as the size of its support is too large. 
The arguments involving only the mutual Hamming distance between the vectors in the support and the size of the support are invariant with respect to the index ordering, and are thus not affected by the possibility of ``cheaply'' changing the index ordering vectors.


\color{black}

%% file: prelim_new.tex
\section{Preliminaries}\label{sec:prelim}
For an integer $n$, we will denote the set $\{1, \ldots, n\}$ as $[n]$.
Given two vectors $\mathbf{X}$ and $\mathbf{Y}$ in $\{0,1\}^{n}$, we 
will denote by $d_{H}(\mathbf{X}, \mathbf{Y})$ the normalized Hamming distance between $\mathbf{X}$ and $\mathbf{Y}$, that is, 
$$
d_{H}(\mathbf{X}, \mathbf{Y}) := \frac{\size{\left\{i \in [n]\,:\, \mathbf{X}_{i} \neq \mathbf{Y}_{i}\right\}}}{n}.
$$

Unless stated otherwise, all the distance measures that we will be considering in this paper will be the normalized distances. For two vectors $\mathbf{X}, \mathbf{Y} \in \{0,1\}^n$, $\delta_H(\mathbf{X}, \mathbf{Y}) = n \cdot d_H(\mathbf{X},\mathbf{Y})$ will be used to denote the absolute Hamming distance between $\mathbf{X}$ and $\mathbf{Y}$ in the few places where we will need to refer to it. {When we write $\tOh(\cdot)$, it suppresses a poly-logarithmic term in $n$ and the inverse of the proximity parameter.}

{We will also need the following observation from \cite{alon2017testing} which roughly states that given a sequence of non-negative real numbers that sum up to an integer $n$, there is a procedure that by choosing the floor or ceiling of these real numbers, one can obtain another sequence of integers that sum up to $n$. This observation will be used in our proof.}

\begin{obs}[{\bf Restatement of \cite[Lemma~4.8]{alon2017testing}}]\label{obs:integral}
{Let $T, n \in \N$. Given $T$ non-negative real numbers $\alpha_1, \ldots, \alpha_T$ such that $\sum\limits_{i=1}^T \alpha_i =n$, there exists a procedure of choosing $T$ integers $\beta_1, \ldots, \beta_T$ such that $\beta_i \in \{\lfloor \alpha_i \rfloor, \lceil \alpha_i \rceil\}$ for every $i \in [T]$ and $\sum\limits_{i=1}^T \beta_i =n$.}
\end{obs}


\subsection{Definitions and relations of various distance measures of distributions}

We will first define $\ell_{1}$ distance between two distributions. 

\begin{defi}[{\bf $\ell_1$ distance and variation distance between two distributions}]
    Let $D_1$ and $D_2$ be two probability distributions over  a set $S$. The $\ell_1$ distance between $D_1$ and $D_2$ is defined as
    $$||D_1-D_2||_1 = \sum_{a \in S} |D_1(a) - D_2(a)|.$$

 
The variation distance between $D_1$ and $D_2$ is defined as:
$$d_{TV}(D_1,D_2) = \frac{1}{2} \cdot ||D_1-D_2||_1.$$
\end{defi}

\color{black}

Throughout this paper, the Earth Mover Distance ({\sc EMD}) is the central metric for testing ``closeness'' and ``farness'' of a distribution from a given property. It is formally defined below.  

\begin{defi}[{\bf Earth Mover Distance ({\sc EMD})}]
Let $D_1$ and $D_2$ be two probability distributions over $\{0,1\}^n$. The EMD between $D_1$ and $D_2$ is denoted by $d_{EM}(D_1, D_2)$, and defined as the solution to the following linear program:
\begin{align*}
    \mbox{Minimize} &\sum_{\mathbf{X},\mathbf{Y} \in \{0,1\}^n} f_{\mathbf{X}\mathbf{Y}} d_H(\mathbf{X},\mathbf{Y})&\\
    \mbox{Subject to} &\sum\limits_{\mathbf{Y} \in \{0,1\}^n} f_{\mathbf{X}\mathbf{Y}} = D_1(\mathbf{X}), &\forall \ \mathbf{X} \in \{0,1\}^n \\
    &\sum\limits_{\mathbf{X} \in \{0,1\}^n} f_{\mathbf{X}\mathbf{Y}} = D_2(\mathbf{Y}), &\forall \ \mathbf{Y} \in \{0,1\}^n\\
    &~~~~~~~~~~0 \leq f_{\mathbf{X}\mathbf{Y}} \leq 1, &\forall \, \mathbf{X}, \mathbf{Y} \in \{0,1\}^n
\end{align*}
\end{defi}

\noindent
Intuitively, the variable $f_{\mathbf{X}\mathbf{Y}}$ stands for the amount of probability mass transferred from $\mathbf{X}$ to $\mathbf{Y}$.

Directly from the definitions of $d_{EM}(D_{1}, D_{2})$ and $d_{H}(\mathbf{X}, \mathbf{Y})$, we get the following simple yet useful observation connecting $\ell_{1}$ distance and EMD between two distributions. 

\begin{obs}[{\bf Relation between EMD and $\ell_1$ distance}]
Let $D_1$ and $D_2$ be two distributions over the $n$-dimensional Hamming cube $\{0,1\}^n$. Then we have the following relation between the Earth Mover Distance and $\ell_1$ distance between $D_1$ and $D_2$:
$$d_{EM}(D_1, D_2) \leq \frac{||D_1-D_2||_1}{2}.$$
\end{obs}

Now we formally define the notions of ``closeness'' and ``farness'' of two distributions with respect to the Earth Mover Distance.

\begin{defi}[{\bf Closeness and farness with respect to {\sc EMD}}]\label{defi:closefardist}
Given two proximity parameters $\eps_1$ and $\eps_2$ with $0 \leq \eps_1 < \eps_2 \leq 1$, two distributions $D_1$ and $D_2$ over the $n$-dimensional Hamming cube $\{0,1\}^n$ are said to be \emph{$\eps_1$-close} if $d_{EM}(D_1, D_2) \leq \eps_1$, and \emph{$\eps_2$-far} if $d_{EM}(D_1, D_2) \geq \eps_2$.
\end{defi}

Now we proceed to define the notion of distribution properties over the Hamming cube below.

\begin{defi}[{\bf Distribution property over the Hamming cube}]\label{defi:distprophamming}
Let $\cD$ denote the set of all distributions over the $n$-dimensional Hamming cube $\{0,1\}^n$. A \emph{distribution property} $\cP$ is a topologically closed subset of $\cD$.~\footnote{We put this restriction to avoid formalism issues. In particular, the investigated distribution properties that we know of (such as monotonicity and being a k-histogram) are topologically closed.} A distribution $D \in \cP$ is said to be \emph{in the property} or to {\em satisfy} the property. Any other distribution is said to be \emph{not in the property} or \emph{to not satisfy} the property.


\end{defi}

Now we are now ready to define the notion of distance of a distribution from a property.


\begin{defi}[{\bf Distance of a distribution from a property}]\label{defi:distprop}
The distance of a distribution $D$ from a property $\cP$ is the minimum Earth Mover Distance between $D$ and any distribution in $\cP$.~\footnote{The assumption that $\cP$ is closed indeed makes it a minimum rather than an infimum.} For $\eps \in [0,1]$, a distribution $D$ is said to be \emph{$\eps$-close} to $\cP$ if the distance of  $D$ from $\cP$ is at most $\eps$. Analogously, for $\eps \in [0,1]$, a distribution $D$ is said to be \emph{$\eps$-far} from $\cP$ if the distance of $D$ from $\cP$ is more than $\eps$.

\end{defi}

\subsection{Formal definitions of various kinds of property testers}

First we define our query model below.

\begin{defi}[{\bf Query to sampled vectors}]
Let $\cA$ be a tester with a set of sampled vectors $\mathbf{V}_1,\ldots, \mathbf{V}_s$,  drawn independently from an input  distribution $D$ over $\{0,1\}^n$, where $\mathbf{V}_i = (v_{i,1}, \ldots, v_{i,n})$ for every $i \in [s]$. In order to perform a query, the tester will provide $i$ and $j$, and will receive $v_{i,j}$ as the answer to the query.
\end{defi}


In the following, we formally describe the notion of a tester.

\begin{defi}[{\bf $\eps$-test}]
Let $\eps \in (0,1)$ be a proximity parameter, and $\delta \in (0,1)$. A probabilistic algorithm $\cA$ is said to \emph{$\eps$-test} a property $\cP$ with probability at least $1-\delta$, if any input in $\cP$ is accepted by $\cA$ with probability at least $1-\delta$, and any input that is $\eps$-far from $\cP$ is rejected by $\cA$ with probability at least $1 - \delta$. Unless explicitly stated, we assume that $\delta={1}/{3}$.
\end{defi}

\color{black}


Now we define two different types of testers, {\em adaptive} testers and {\em non-adaptive} testers, which will be used throughout the paper. We begin by describing the adaptive testers. Informally, adaptive testers correspond to algorithms that perform queries depending on the answers to previous queries. Formally:

\begin{defi}[{\bf Adaptive tester}]
Let $\cP$ be a property over $\{0,1\}^++$.
An adaptive tester for $\cP$ with query complexity $q$ and sample complexity $s$ is a randomized algorithm $\cA$ that $\eps$-tests $\cP$ by performing the following:

\begin{itemize}
    \item $\cA$ first draws some random coins and samples $s$ vectors from the unknown distribution $D$, denoted by $S=\{ \mathbf{V}_1,\ldots, \mathbf{V}_s\}$.
    
    \item $\cA$ then queries the $j_1$-th index of $\mathbf{V}_{i_1}$, for some $j_1 \in [n]$ and $i_1 \in [s]$ depending on the random coins.
    
    \item Suppose that $\cA$ has executed $k$ steps and has queried the $j_{\ell}$-th index of $ \mathbf{V}_{j_{\ell}}$, where $1 \leq \ell \leq k$. At the $(k+1)$-th step, depending upon the random coins and the answers to the queries till the $k$-th step, $\cA$ will perform a query for the $j_{k+1}$-th bit of $\mathbf{V}_{i_{k+1}}$, where $j_{k+1} \in [n]$ and  $i_{k+1} \in S$.
    
    
    \item After $q$ steps, $\cA$ reports {\sc Accept} or {\sc Reject} depending on the random coins and the answers to all $q$ queries.
\end{itemize}

\end{defi}


Now we define the more restricted \emph{non-adaptive testers}. Informally, non-adaptive testers decide the set of queries to be performed on the input even before performing the first query. Formally:

\begin{defi}[{\bf Non-adaptive tester}]\label{defi:non-adaptive-tester}
Let $\cP$ be a property over $\{0,1\}^n$.
A non-adaptive tester for $\cP$ with query complexity $q$ and sample complexity $s$ is a randomized algorithm $\cA$ that $\eps$-tests $\cP$ by performing the following:
\begin{itemize}
    \item $\cA$ tosses some random coins, and depending on the answers constructs a sequence of subsets of indices $J_1, \ldots, J_s \subset [n]$ such that $\sum\limits_{i=1}^s J_i\leq q$.

    \item $\cA$ takes $s$ samples  $ \mathbf{V}_1,\ldots, \mathbf{V}_s$ from the unknown distribution $D$.

    \item $\cA$ queries for the coordinates of $\mathbf{V}_i$ that are in $J_i$, for each $i \in [s]$.
    
    \item $\cA$ reports either {\sc Accept} or {\sc Reject} based on the answers from the queries to the vectors, that is, $\mathbf{V}_1\mid_{J_1}, \mathbf{V}_2\mid_{J_2}, \ldots, \mathbf{V}_s\mid_{J_s}$, and the random coins.
    
    
\end{itemize}

\end{defi}


\subsection{Distributions and properties with bounded VC-dimension}

Now we move on to define a class of properties using the notion of the VC-dimension of the support of a distribution. Before proceeding to define the class of properties, let us recall the notions of \emph{shattering} and \emph{VC-dimension}.

Let $V$ be a collection of vectors from $\{0, 1\}^n$. For a
sequence of indices $I = (i_1, \ldots, i_k)$, with $1 \leq i_j \leq n$, let $V \mid_I$ denote the set of {\em projections} of $V$ onto $I$, that is,
$$V\mid_I = \{(v_{i_1}, \ldots, v_{i_k}): (v_1, \ldots, v_n) \in V \}.$$
If $V \mid_I = \{0,1\}^k$, then it is said that $V$ \emph{shatters} the index sequence $I$. The \emph{VC-dimension} of $V$ is the size of the largest index sequence $I$ that is shattered by $V$. VC-dimension was introduced by Vapnik and Chervonenkis \cite{vapnik2015uniform} in the context of learning theory, and has found numerous applications in other areas like 
approximation algorithms, discrete and computational geometry, discrepancy theory, see~\cite{M99,PA95,DBLP:books/daglib/0018467,C00}.



We now give a natural extension of VC-dimension to distributions. 

\begin{defi}[{\bf Distribution with  VC-dimension $d$}]
Let $d,\, n \in \N$ and $D$ be a distribution over $\{0,1\}^n$. We say that $D$ has VC-dimension at most $d$ if the support of $D$ has {VC-dimension} at most $d$. A distribution $D$ is said to be \emph{$\beta$-close to VC-dimension $d$} if there exists a distribution $D_{0}$ with VC-dimension $d$ such that $d_{EM}(D,D_0)\leq \beta$, where $\beta \in (0,1)$.
\end{defi}


%
%
%

Analogously, we can also define the notion of a $(\beta, d)$-VC-dimension property.

\begin{defi}[{\bf $(\beta, d)$-VC-dimension property}]\label{defi:closevcdimprop}
Let $d, n \in \N$ and $\beta \in (0,1)$.
A property $\cP$ over $\{0,1\}^n$ is said to be a \emph{$(\beta, d)$-VC-dimension property}
if for any distribution $D\in \cP$, $D$ is $\beta$-close to VC-dimension $d$. When $\beta=0$, we say that the VC-dimension of $\cP$ is $d$. We also say that a $(0,d)$-VC-dimension property is a \emph{bounded VC-dimension property}.
\end{defi}

We now give examples of bounded VC-dimension properties.
\begin{description}
    \item[Property {\sc Chain}:]
        For any distribution $D \in \mbox{{\sc Chain}}$, the support of $D$ can be written as a sequence $\mathbf{X}_1,\ldots,\mathbf{X}_t \in \{0,1\}^n$ such that any two
    vectors with non-zero probability are comparable, that is, $$D(\mathbf{X}_i) >0 \ \mbox{and} \ D(\mathbf{X}_j) > 0 \ \mbox{implies either} \ \mathbf{X}_i \preceq \mathbf{X}_j \ \mbox{or} \ \mathbf{X}_j \preceq \mathbf{X}_i \mbox{, for every} \ i,j \in [t]. $$

    \item[Property {\sc Low-affine-dimension}:] 
        A distribution $D$ over $\{0,1\}^n$ is said to satisfy the {\sc Low- affine-dimension} property, with parameter $d \in \N$, if the {\em affine dimension}\footnote{A set $S \subseteq \R^{n}$ has {\em affine dimension} $k$ if the dimension of the smallest {\em affine set} in $\R^{n}$ that contains $S$ is $k$.} 
of the support of $D$ is at most $d$.
\end{description}    
Observe that the VC-dimension of {\sc Chain} is $1$, and the VC-dimension of {\sc Low-affine-dimension} is $d$. ~\footnote{In fact, the property {\sc Low-affine-dimension} is a sub-property of ``support size is at most $2^d$'', which has VC-dimension $d$.} Moreover, note that both {\sc Chain} and {\sc Low-affine-dimension} are examples of index-invariant properties.



\subsection{Yao's lemma for the huge object model}

Our lower bound proofs crucially use Yao's lemma~\cite{yao1977probabilistic}. Informally, it states that for any two distributions $D_1$ and $D_2$ such that $D_1$ satisfies some property, and $D_2$ is far from the property, if the variation distance between $D_1$ and $D_2$ with respect to $q$ queries is small, then $D_1$ and $D_2$ remain indistinguishable with respect to $q$ queries. In order to formally state the lemma, we need the following definitions.

\begin{defi}[{\bf Restriction}]
Let $D$ be a distribution over a collection of functions $f:\cD \rightarrow \{0, 1\}$, and $Q$ be a subset of the domain $\cD$ of $D$. The restriction $D \mid_{Q}$ of $D$ to $Q$ is the distribution over functions of the form $g: Q \rightarrow \{0, 1\}$, which is obtained from choosing a
random function $f: \cD \rightarrow \{0, 1\}$ according to the distribution $D$, and then setting $g = f \mid_Q$, where $f \mid_Q$ denotes the
restriction of $f$ to $Q$.

\end{defi}

The following is the version of Yao's Lemma which is used for non-adaptive testers in the classical setting. The crucial observation that makes this lemma work is the observation that the deterministic version of a non-adaptive tester in the classical setting is characterized by a set of possible responses to a fixed query set $Q \subset \cD$.


\begin{lem}[{\bf Yao's lemma for non-adaptive testers, see~\cite{fischer2004art}}]
Let $\eps \in (0,1)$ be a parameter and $q \in \N$ be an integer.
Suppose there exists a distribution $D_{yes}$ on inputs over $\cD$ that satisfy a given property $\cP$, and a distribution $D_{no}$ on inputs that are $\eps$-far from satisfying the property. Moreover, assume that for any set of queries $Q \subset \cD$ of size $q$, the variation distance between $D_{yes}\mid_Q$ and $D_{no}\mid_Q$ is less than $\frac{1}{3}$. Then it is not possible for a non-adaptive tester performing $q$ (or less) queries to $\eps$-test $\cP$.
\end{lem}

In this paper, we will prove lower bounds against non-adaptive distribution testers in the huge object model. Hence, $D_{yes}$ and $D_{no}$, rather than being distributions over functions from $\cD$ to $\{0,1\}$, are distributions over distributions over $\{0,1\}^n$ (since the basic input object is a distribution over $\{0,1\}^n$).

The deterministic version of a non-adaptive tester in this setting is characterized by a set of possible responses to a sequence of queries $\mathcal{J}=(J_1, \ldots, J_s)$ to the samples. We call $s$ the \emph{length} of $\mathcal{J}$, and call $q=\sum_{i=1}^s J_i$, the \emph{size} of $\mathcal{J}$.

Given a distribution $D$ over distributions over $\{0,1\}^n$, we denote by $D\mid_{\mathcal{J}}$ the distribution over $\{0,1\}^q$ that results from first picking a distribution $\widehat{D}$ over $\{0,1\}^n$ according to $D$, then taking $s$ independent samples $\mathbf{X}_1, \ldots, \mathbf{X}_s$ according to $\widehat{D}$, and finally constructing the sequence $\mathbf{X}_1 \mid_{J_1}, \ldots, \mathbf{X}_s \mid_{J_s}$. The huge object model version of Yao's lemma for non-adaptive testers is the following one.

\begin{lem}[{\bf Yao's lemma for non-adaptive testers in the huge object model}]
Let $\eps \in (0, 1)$ be a parameter and $q,s \in \N$ be two integers. Suppose there exists a distribution $D_{yes}$ over distributions over $\{0,1\}^n$ that satisfy a given property $\cP$, and a distribution $D_{no}$ over distributions over $\{0,1\}^n$ that are $\eps$-far from satisfying the property $\cP$. Moreover, assume that for any query sequence $\mathcal{J}$ of length $s$ and size $q$, the variation distance between $D_{yes} \mid_{\mathcal{J}}$ and $D_{no} \mid_{\mathcal{J}}$ is less than $1/3$. Then it is not possible for a non-adaptive tester that takes at most $s$ samples and performs at most $q$ queries to $\eps$-test $\cP$.
\end{lem}

\color{black}

\color{black}

\color{black}

%% file: VC-dim-new.tex
\section{Learning clusterable distributions}\label{sec:learncluster}
\remove{\begin{defi}[{\bf $(\alpha, r)$-clusterable distribution}]\label{defi:cluster1}

Let $\alpha \in (0,1)$ and $r,n \in \N$. A distribution $D$ over $\{0,1\}^n$ is said to be $(\alpha,r)$-clusterable if there exists a partition $\cC_1, \ldots, \cC_s$ of the support of $D$ such that
\begin{description}
    \item[(i)]   $s \leq r$;
    \item[(ii)] For any $1 \leq i \leq s$,  $d_H(\mathbf{U},\mathbf{V})\leq \alpha$  for any $\mathbf{U},\mathbf{V} \in \cC_i$;
\end{description}

\end{defi}

\begin{defi}[{\bf Being $\beta$-EMD-close to $(\alpha,r)$-clusterable}]\label{defi:cluster2}
A distribution $D$ is said to $\beta$-EMD-close to be $(\alpha, r)$-clusterable if there exists a distribution $D_0$ such that 
\begin{description}
    \item[(i)] $D_0$ is $(\alpha, r)$-clusterable.
    \item[(ii)] $d_{EM}(D_0, D)\leq \beta$.
\end{description}
\end{defi}
\color{black}
}

{In this section, we define the notion of a $(\zeta,\delta,r)$-clusterable distribution formally (see Definition~\ref{defi:cluster}), and prove that such distributions can be learnt (up to permutation) efficiently in Theorem~\ref{theo:mainthm1}. Intuitively, a distribution $D$ defined over $\{0,1\}^n$ is called $(\zeta,\delta,r)$-clusterable if we can remove a subset of the support vectors of $D$ whose probability mass is at most $\zeta$, and we can partition the remaining vectors in the support of $D$ into at most $r$ parts, each with diameter at most $\delta$. Theorem~\ref{theo:mainthm1} states that, given a distribution $D$ over $\{0,1\}^n$, we can learn it (up to permutation) if it is $(\zeta,\delta,r)$-clusterable, and otherwise, we either report {\sc Fail} or learn the input distribution (up to permutation). Note that learning the distribution up to permutation is sufficient  to provide testing algorithms for index-invariant properties with bounded VC-dimension, which will be discussed in Section~\ref{sec:vc}.}

\begin{defi}[{\bf $(\zeta, \delta, r)$-clusterable and $(\alpha,r)$-clusterable distribution}]\label{defi:cluster}

\vspace{4 pt}

\begin{enumerate}
    \item[(i)] Let $\zeta,\delta \in (0,1)$ and $r,n \in \N$. A distribution $D$ over $\{0,1\}^n$ is called  \emph{$(\zeta, \delta,r)$-clusterable} if there exists a partition $\cC_0, \ldots, \cC_s$ of $\{0,1\}^n$ such that  
 $D(\cC_0) \leq \zeta$,   $s \leq r$, and 
    for every $1 \leq i \leq s$,  $d_H(\mathbf{U},\mathbf{V})\leq \delta$  for any $\mathbf{U},\mathbf{V} \in \cC_i$.
    
    \item[(ii)]  For $\alpha \in (0,1)$ and $r \in \N$, a distribution $D$ over $\{0,1\}^n$ is called \emph{$(\alpha,r)$-clusterable} if it is \emph{$(0,\alpha,r)$-clusterable}. For $\beta \in (0,1)$, a distribution $D$ is called \emph{$\beta$-close to being $(\alpha,r)$-clusterable} if there exists an $(\alpha,r)$-clusterable distribution $D_0 $ such that $d_{EM}(D,D_0)\leq \beta$.
\end{enumerate}
\end{defi}


\begin{theo}[{\bf Learning $(\zeta,\delta,r)$-clusterable distributions}]\label{theo:mainthm1}
There exists a (non-adaptive) algorithm \learn, as described in Algorithm~\ref{alg:testcluster--}, that has sample and query access to an unknown distribution $D$ over $\{0,1\}^n$ for $n \in \N$, takes parameters $\zeta, \delta, r$ as inputs such that, $\zeta,\delta \in (0,1)$ and $\eps=17(\delta+\zeta)<1$~\footnote{The constant $17$ is arbitrary, and can be improved to a smaller constant. We did not try to optimize.} and $r \in \N$, makes a number of queries that only depends on $\zeta,\delta$ and $r$, and either reports a full description of a distribution over $\{0,1\}^n$ or reports {\sc Fail}, satisfying both of the following conditions:

\begin{enumerate}
    \item[(i)] If $D$  $(\zeta, \delta, r)$-clusterable, then with probability at least $\frac{2}{3}$, the algorithm outputs a full description of a distribution $D'$ over $\{0,1\}^n$ such that $d_{EM}(D, D'_{\sigma}) \leq \eps$ for some permutation $\sigma:[n] \rightarrow [n]$.
    
    \item[(ii)] For any $D$, the algorithm will not output a distribution $D'$ such that $d_{EM}(D,D'_\sigma) > \eps$ for every permutation $\sigma:[n]\rightarrow [n]$, with probability more than $\frac{1}{3}$. However, if the distribution $D$  is not $(\zeta,\delta, r)$-clusterable, the algorithm may output {\sc Fail} with any probability.
\end{enumerate}
\end{theo}

The algorithm corresponding to learning $(\zeta,\delta,r)$-clusterable distributions is described in Algorithm~\ref{alg:testcluster--} as \learn. It calls a subroutine \aprox, as described in Algorithm~\ref{alg:approx}.


\begin{rem}
    The sample complexity of \learn is polynomial in $r$, and the query complexity of \learn is exponential in $r$. Moreover, for the case of query complexity, the exponential dependency in $r$ is required. In particular, in Section~\ref{sec:bounded-vc}, we construct a distribution with support size $r$ that requires $2^{\Omega(r)}$ queries to test for the property of being a permutation thereof.
\end{rem}

\color{black}

\input{algo.tex}

\color{black}

{To prove the correctness of \learn (which we will do in Section~\ref{sec:event} and Section~\ref{sec:lem-inetr}), we will need the notion of an $(\eta,\xi)$-clustered distribution around a sequence of vectors $\cS$ (see Definition~\ref{def:concdist}), and an associated observation (see Observation~\ref{obs:extra})}. 

\begin{defi}[{\bf $(\eta, \xi)$-clustered distribution around a sequence}]\label{def:concdist}
Let $\eta ,\xi \in (0,1)$ and $n \in \N$.  Also, for $\mathbf{X} \in \{0,1\}^n$, let  $\mbox{{\sc NGB}}_{\eta}({\bf X})$ denote the set of vectors in $\{0,1\}^n$ that are at a distance of at most $\eta$ from ${\bf X}$. Let $\cS= \{\mathbf{S}_1, \ldots, \mathbf{S}_{t}\}$ be a sequence of $t$ vectors in $\{0,1\}^n$ and define $\mbox{{\sc NGB}}_{\eta}(\cS)=\bigcup\limits_{S \in \cS} \mbox{{\sc NGB}}_{\eta}(S)$. Then:

\begin{enumerate}
\item[(i)]  A distribution $D$ over $\{0,1\}^n$ is called \emph{$(\eta,\xi)$-clustered around $\cS$  with weights} $w_0, \ldots, w_{t} \in [0,1]$ satisfying $\sum\limits_{i=0}^{t} w_i =1$ and $w_0 \leq \xi$, if {there exist $t$ pairwise disjoint sets $\cC_i$, such that $\cC_i \subseteq \mbox{{\sc \mbox{{\sc NGB}}}}_{\eta}(\mathbf{S}_i)$ and $D(\cC_i) \geq w_i$ for every $i \in [t]$}.



\item[(ii)] A distribution $D$ over $\{0,1\}^n$ is called  \emph{$(\eta,\xi)$-clustered around $\cS$} if $D$ is $(\eta,\xi)$-clustered around $\cS$ with weights $w_0, \ldots, w_t \in [0,1]$, for some $w_0, \ldots, w_t$ such that $\sum\limits_{i=0}^t w_i =1$ and $w_0 \leq \xi$.

\end{enumerate}
\end{defi}
\begin{obs}\label{obs:extra}
Let $D$ be any distribution over $\{0,1\}^n$ and $\cS$ be a sequence of vectors in $\{0,1\}^n$ such that $\mbox{{\sc NGB}}_{\eta}(\cS) \geq 1-\xi$. Then  $D$ is $(\eta,\xi)$-clustered around $\cS$.
\end{obs}

\begin{proof}
Let us partition $\mbox{NGB}_\eta(\cS)$ into $t$ parts such that $\cC_i = \mbox{NGB}_\eta(\mathbf{X}_i) \setminus \bigcup\limits_{j=1}^{i-1}\mbox{NGB}_\eta(\mathbf{X}_j)$ for every $i\in [t]$. For every $i \in [t]$, note that $\cC_i \subseteq \mbox{NGB}_\eta(\mathbf{X}_i)$, and let us define $w_i=D(\cC_i)$. Also, set $w_0=1-\sum\limits_{i=1}^{t}w_i$, and observe that $w_0=1-\mbox{{\sc NGB}}_{\eta}(\cS)\leq \xi$. This shows that $D$ is $(\eta, \xi)$-clustered around $\cS$ with weights $w_0,\ldots,w_{t}$, and we are done.
\end{proof}

{The correctness proof of \learn is in Subsection~\ref{sec:lem-inetr}. Leading to it, in Subsection~\ref{sec:event}, we consider some important lemmas and define a set of events. These lemmas, and the events whose probability they bound from below, provide the infrastructure for the proof of \learn in Subsection~\ref{sec:lem-inetr}.}

\color{black}




\subsection{Useful lemmas and events to prove the correctness of \learn}\label{sec:event}

{The central goal of this section is to define an event {\sc GOOD} and show that $\pr\left(\mbox{{\sc GOOD}}\right) \geq 2/3$. The event {\sc GOOD} is defined in such a fashion that, if it holds, then the algorithm \learn produces the desired output as stated in Theorem~\ref{theo:mainthm1}. Note that this bounds the error probability of \learn. The event {\sc GOOD} is formally defined in Definition~\ref{defi:eventgood}.}
{To define the event {\sc GOOD}, we first consider four lemmas: Lemma~\ref{lem:hit}, Lemma~\ref{lem:dist}, Lemma~\ref{lem:mass} and  Lemma~\ref{lem:struct}.}

{We will first state a lemma (Lemma~\ref{lem:hit}) which says that, with high probability, the first set of samples $\cS$ (obtained in Step (i) of \learn) intersects all the large clusters when $D$ is $(\zeta, \delta, r)$-clusterable.}



\begin{lem}[{\bf Hitting large clusters}]\label{lem:hit}
Assume that the input distribution $D$ over $\{0,1\}^n$ is $(\zeta, \delta, r)$-clusterable with the clusters $\cC_1, \ldots, \cC_r$. The cluster $\cC_i$ is said to be \emph{large} if $D(\cC_i)\geq \frac{\zeta}{10 r}$. With probability at least $99/100$, the sequence of vectors $\cS=\{\mathbf{X}_1,\ldots {\mathbf{X}_{t_1}}\}$ (found in Step (i) of \learn) contains at least one vector from every large cluster.
\end{lem}

\begin{proof}
Consider any large cluster $\cC_i$. As $D(\cC_i) \geq \frac{\zeta}{10r}$, the probability that no vector in $\cS$ belongs to $\cC_i$ is at most $(1 - \frac{\zeta}{10r})^{\size{\cS}}\leq \frac{99}{100 r}$. This follows  for a suitable choice of the hidden coefficient since $\size{\cS}=t_1=\Oh\left(\frac{r}{\zeta} \log \frac{r}{\zeta}\right)$. Since there are at most $r$ large clusters, using the union bound, the lemma follows.
\end{proof}


{Recall that \learn obtains a second set of sample vectors $\cT$ in Step (ii), takes a random set of indices $R \subset [n]$ without replacement in Step (iii), and tries to assign each vector in $\cT$ to some vector in $\cS$, based on the distance between the vectors when projected to the indices of $R$. Intuitively, the step of assigning vectors performs as desired if $R$ preserves the distances between the vectors in $\cS$ and $\cT$. For technical reasons, we also need $R$ to preserve most (but not all) distances between $\cS$ and the entirety of $\{0,1\}^n$. The following lemma says that indeed $R$ achieves this with high probability.}


\begin{lem}[{\bf Distance preservation}]\label{lem:dist}
Let us consider the input distribution $D$ over $\{0,1\}^n$, and $\cS=\{\mathbf{X}_1,\ldots,\mathbf{X}_{t_1}\} $ and $R \subset [n]$ drawn in Step (i) and (iii) of \learn.  $R$ is said to be \emph{distance preserving}
if the following conditions hold:
\begin{itemize}
\item[(i)] $\size{d_H (\mathbf{S},\mathbf{T})-d_H (\mathbf{S} \mid_R , \mathbf{T} \mid_R)}\leq \delta$ for every $\mathbf{S} \in \cS$ and $\mathbf{\bf T} \in \cT$.

\item[(ii)] Let $\mathcal{W} \subseteq \{0,1\}^n$ be such that, for every $\mathbf{W} \in \mathcal{W}$, $\size{d_H (\mathbf{W},\mathbf{S})-d_H (\mathbf{W} \mid_R , \mathbf{S} \mid_R)}\leq \delta$. Then $D(\mathcal{W})\geq 1- \frac{\zeta}{ t_1}$.
\end{itemize}
The set $R$ chosen in Step (iii) of \learn is distance preserving with probability 
at least $99/100$.
\end{lem}

\begin{proof}
For (i), consider a particular $\mathbf{S} \in \cS$ and $\mathbf{T} \in \cT$. Applying 
Observation~\ref{cl:distchernoff} with $K=R$, $\mathbf{U}=\mathbf{S}$ and $\mathbf{V}=\mathbf{T}$, the probability that $\size{d_H(\mathbf{S},\mathbf{T})-d_H(\mathbf{S} \mid_R,\mathbf{T} \mid_R)} \leq \delta$ is at least 
$1-\frac{\zeta}{200 t_1^2 t_2}$. Applying the union bound over all possible choices over $(\mathbf{S},\mathbf{T})$
pairs, we have Part (i) with probability at least $199/200.$

To prove (ii), let us consider an arbitrary vector $\mathbf{V} \in \{0,1\}^n$. Similarly to (i), we know that $\size{d_H(\mathbf{V},\mathbf{S})-d_H(\mathbf{V} \mid_R,\mathbf{S} \mid_R)} \leq \delta$ holds with probability at least $1-\frac{\zeta}{200 t_1^2 t_2}$. Applying the union bound, we can say that the same holds over all $\mathbf{S} \in \cS$ with probability at least $1 - \frac{\zeta}{200t_1}$. So, the expected value of $D(\{0,1\}^n \setminus \mathcal{W})$ is at most $\frac{\zeta}{200 t_1}$. By Markov's inequality, the probability that Part (ii) holds, that is, $D( \{0,1\}^n \setminus \mathcal{W}) \leq \frac{\zeta}{t_1}$ is at least $199/200$. 
Putting everything together, we have the result.
\end{proof}

\color{black}
{By Lemma~\ref{lem:hit}, we know that $\cS$ intersects with all large clusters with high probability, and we are trying to assign the vectors in $\cT$ to some vectors in $\cS$ based on their projected distances on the indices of $R$. To learn the input distribution, we want the second set of sample vectors $\cT$ to preserve the mass of all the large clusters, and it is enough for us to approximate it, as well as be able to detect the case where approximation is impossible and we should output {\sc Fail}. The following lemma takes care of this.}

\begin{lem}[{\bf Weight representation}]\label{lem:mass}
Let us consider the input distribution $D$ over $\{0,1\}^n$ to \learn, $\cS=\{\mathbf{X}_1,\ldots,\mathbf{X}_{t_1}\} $ in Step (i),   $\cT=\{\mathbf{Y}_1,\ldots,\mathbf{Y}_{t_2}\} $ in Step (ii), and consider fixed $t_1$ pairwise disjoint subsets $\cC=\{\cC_1,\ldots,\cC_{t_1}\}$ of $\{0,1\}^n$. $\cT$ is said to be \emph{weight preserving} for $\cS$ and $\cC$ if
\begin{enumerate}
    
\item[(i)]   $\frac{\size{\cT ~\cap~ \mbox{{\sc NGB}}_{\delta}(\cS) }}{\size{\cT}}  \geq D(\mbox{{\sc NGB}}_{ \delta}(\cS))-{\zeta}$.

\item[(ii)]  $ \frac{\size{\cT ~\cap ~\mbox{{\sc NGB}}_{3\delta}(\cS) }}{\size{\cT}}  \leq D(\mbox{{\sc NGB}}_{ 3\delta}(\cS))+ {\zeta}$.
        
\item[(iii)] {for every $i \in [t_1]$, $\frac{\size{\cT ~\cap~ \cC_i}}{\size{\cT}} \leq D( \cC_i)+\frac{\zeta}{t_1}$.}
\end{enumerate} 

Then with probability at least $99/100$, $\cT$ is weight reserving for $\cS$ and $\cC$.
\end{lem}

\begin{proof}
\color{black}


To prove (i), let $Z_j$ be the indicator random variable such that $Z_j=1$ if and only if $\mathbf{Y}_j$ is in $\mbox{{\sc NGB}}_{\delta}(\cS)$, where $j \in [t_2]$. Observe that $\size{\cT \cap NGB_{\delta}(\cS)}=\sum\limits_{j=1}^{t_2} Z_j$. As $\pr(Z_j=1)=D(\mbox{{\sc NGB}}_{\delta}(\cS))$, the expected value of $\frac{\size{\cT \cap \mbox{{\sc NGB}}_{\delta}(\cS)}}{\size{\cT}}$ is also $D(\mbox{{\sc NGB}}_{\delta}(\cS))$. Applying Hoeffding's inequality (see Lemma~\ref{lem:hoeffdingineq}), we conclude that (i) holds with probability at least $299/300$.
     
Proving (ii) is similar to (i). Again     applying  Hoeffding's inequality (Lemma~\ref{lem:hoeffdingineq}), we can show that (ii) holds with probability at least $299/300.$

In order to prove (iii), we proceed in similar fashion as (i), and after applying Hoeffding's inequality (Lemma~\ref{lem:hoeffdingineq}), we apply the union bound over all $j \in [t_1]$ to get the desired result.
\end{proof}


{Consider the weights $w_1, \ldots, w_{t_1}$ obtained in Step (vi) of \learn. To argue that these weights are good enough to report the desired distribution $D'$ (if we know the vectors in $\cS$ exactly), we consider the following observation which says that there exist $t_1$ pairwise disjoint subsets $\cC_1^*,\ldots,\cC_{t_1}^*$ such that $w_i$ is the fraction of vectors in $\cT$ that are in $\cC_i^*$ for every $i \in [t_1]$.} Also, let us define $\cC^*=\{\cC_1^*,\ldots,\cC_{t_1}^*\}$.

\begin{obs}\label{obs:crucial}
Let us consider assigning each vector in $\{0,1\}^n$ either to some $S \in \cS$ or not assigning to any vector in $\cS$, using the same procedure that has been used to assign the set of vectors in $\cT$ in Steps (iii) and (iv) of \learn. Let $\cC^*_i \subseteq \{0,1\}^n$ be the set of all vectors that are assigned to $\mathbf{X}_i$, for every $i\in [t_1]$. Then, for every $i \in [t_1]$, we have  $w_i = \frac{\size{\cT ~ \cap~ \cC_i^*}}{\size{\cT}}$.
\end{obs}

\begin{proof}
This follows from the definition of $\cC_i^*$.
\end{proof}

{Note that $\cC^*$ is formed following the procedure that \learn performs to assign the vectors of $\cT$ to the vectors in $\cS$. So, a vector far away from $\mathbf{X}_i \in \cS$ might be assigned $\mathbf{X}_i$, and $w_i$ is considered in this case. This is not a problem as the mass on $\cC_i^*$ is close to being bounded by the total mass of the vectors in $\mbox{{\sc NGB}}_{3\delta}(\mathbf{X}_i)$. This follows from the fact that the set $R$ is distance preserving (see Part (ii) of Lemma~\ref{lem:dist}) with high probability. Now let us define $\cC^{**}=\{\cC_i^* \cap \mbox{{\sc NBG}}_{3\delta}(\mathbf{X}_i):i \in [t_1]\}$. Finally, we will upper bound $w_i$ by $D(\cC^{**}_i)$ in the following observation. This will be useful for proving the correctness of \learn in Section~\ref{sec:lem-inetr}.}

\begin{obs}\label{obs:crucial-2}
Let us assume that $R$ is distance preserving and $\cT$ is weight representative of $\cS$ and $\cC^*$. Then for every $i \in [t_1]$, 
$  w_i \leq D(\cC_i^*)+\frac{\zeta}{t_1}\leq D(\cC_i^{**})+\frac{2\zeta}{t_1}$, where we define $\cC_i^{**}=\cC_i^* \cap \mbox{{\sc NBG}}_{3\delta}(\mathbf{X}_i)$.
\end{obs}

\begin{proof}
As $R$ is distance preserving, consider $\cC^*=\{\cC_1^*,\ldots,\cC_{t_1}^*\}$ as guaranteed by Observation~\ref{obs:crucial}. Now, as $\cT$ is weight representative of $\cS$ and $C^*$ and $w_i=\frac{\size{\cT~\cap~\cC_i^*}}{\size{T}}$ for every $i \in [t_1]$, by Lemma~\ref{lem:mass} (iii), $w_i\leq D(\cC_i^{*})+\frac{\zeta}{t_1}$. By the definition of $\cC_i^*$ and by Lemma~\ref{lem:dist} (ii), $D(\cC_i^* \setminus \mbox{{\sc NGB}}_{3\delta}(\mathbf{X}_i)) \leq \frac{\zeta}{t_1}$, that is, $D(\cC_i^{*}) \leq D(\cC_i^{**})+\frac{\zeta}{t_1}$.
\end{proof}

{Note that the above observation only gives upper bounds on the set of weights $w_1, \ldots,w_{t_1}$. As Lemma~\ref{lem:mass} provides upper as well as lower bounds on the mass around $\cS$, this will not be a problem.}
\color{black}

{Consider the distribution $D^*$ supported over $\cS$ such that 
$D(\mathbf{X}_i)\geq w_i$ for every $i \in [t_1]$, which we can view as an approximation of $D$. Note that we still can not 
report $D^*$ as the output distribution, since in order to do so, we need to perform $\Omega(n)$ queries to know the exact vectors of $\cS$. Instead we will report a distribution $D'$ such that $D'_\sigma$ is close to $D^*$ for some permutation $\sigma:[n] \rightarrow [n]$. The idea is to construct a new set of vectors $\mathbf{S}_1, \ldots, \mathbf{S}_{t_1}$ in Step (vii) such that the Hamming distance between $\mathbf{X}_i$ and $\sigma(\mathbf{S}_i)$ is small for every $i \in [t_1]$ for some permutation $\sigma:[n]\rightarrow [n]$. Lemma~\ref{lem:struct} implies that this is possible from the projection of the vectors in $\cS$ onto the indices of $R$ (the implication itself will be proved later in Lemma~\ref{lem:approxguarantee}).} Before proceeding to Lemma~\ref{lem:struct}, we need the following definition and observation.

\begin{defi}
Given any sequence of vectors $\cS= \{\mathbf{X}_1, \ldots, \mathbf{X}_{t_1}\} \subseteq \{0,1\}^n$ and $j\in [n]$, we define the vector $C_j^\cS\in \{0,1\}^{t_1}$ as $$\mbox{for every } i\in [t_1],\ C_j^\cS (i) = \mathbf{X}_i(j).$$ 
For any $J\in \{0,1\}^{t_1}$, we define $$\alpha_J = \frac{|\{j\in [n]\mid C_j^{\mathcal{S}} = J\}|}{n}.$$
\end{defi}

Intuitively, let us consider a matrix $M$ of order $t_1\times n$ such that the $i$-th row vector corresponds to the vector $\mathbf{X}_i$. Then observe that $C_j^{\cS}$ represents the $j$-th column vector of the matrix $M$ and $\alpha_J$ denotes the fraction of column vectors of $M$ that are identical to $J$.



\begin{lem}[{\bf Structure preservation}]\label{lem:struct}
Let us consider the input distribution $D$ over $\{0,1\}^n$,  $\cS=\{\mathbf{X}_1,\ldots,\mathbf{X}_{t_1}\} $ and $R \subset [n]$ drawn in Step (i) and (iii) of \learn. Also, let us consider the values of $\Gamma_J$ found in Step (iii) of \aprox (called from Step (vii) of \learn). The set $R$ is said to be \emph{structure preserving} for $\cS$ if $\size{\alpha_J-\frac{\Gamma_J}{n}}\leq \frac{\delta}{10 \cdot 2^{t_1}}$ holds for every $J \in \{0,1\}^{t_1}$. Then the set $R$ chosen in Step (iii) of \learn is structure preserving for $\cS$ with probability at least $99/100$.
\end{lem}

\begin{proof}
Consider any particular $J\in \{0,1\}^{t_1}$ and $\gamma_J$ determined by Step (ii) of \aprox. By applying Hoeffding's bound for sampling without replacement (Lemma~\ref{lem:hoeffdingineq_without_replacement}), we obtain, for any $\eta >0$, $$\Pr\left[|\gamma_J - \alpha_J| \geq \frac{\eta}{20}\right] \leq e^{-2\eta^2|R|/400}.$$
By substituting the value of $|R|$ (for a suitable choice of the hidden coefficient) and $\eta=\frac{\delta}{2^{t_1}}$, and using the union bound over all possible $J \in \{0,1\}^{t_1}$, we conclude that with probability at least $99/100$,
for all $J \in \{0,1\}^{t_1}$, $|\gamma_J - \alpha_J|\leq \frac{\delta}{20 \cdot 2^{t_1}}$. 

Note that \aprox constructs $\Gamma_J$'s from $\gamma_j$'s by applying Observation~\ref{obs:integral}. From the way Observation~\ref{obs:integral} generates $\Gamma_J$'s from $\gamma_j$'s, we conclude that for all $J \in \{0,1\}^{t_1}$, $|\gamma_J- \frac{\Gamma_J}{n}| \leq \frac{1}{n}$, completing the proof, assuming that $n$ is larger than $\frac{20 \cdot 2^{t_1}}{\delta}$.
\end{proof}


Now we are ready to define the event {\sc GOOD}.
\begin{defi}[{\bf Definition of the event $\mbox{{\sc GOOD}}$}]\label{defi:eventgood}

Let us define an event $\mbox{{\sc GOOD}}$ as $\cE_1 \wedge \cE_2 \wedge \cE_3 \wedge \cE_4$, where 
\begin{enumerate}
    \item[(i)] $\cE_1:$ If $D$ is $(\zeta, \delta, r)$-clusterable with the clusters $\cC_1, \ldots, \cC_r$, then $\cS=\{\mathbf{X}_1,\ldots {\mathbf{X}_{t_1}}\}$ (found in Step (i) of \learn) contains at least one vector from every large cluster.
    
    \item[(ii)] $\cE_2:$ $R$ (picked in Step (ii) of \learn) is distance preserving\remove{ between $\cS$ and $\cT$}.
    
    \item[(iii)] $\cE_3:$ $R$ is structure preserving for $\cS$.
   
    \item[(iv)] $\cE_4$:  $\cT$ is weight preserving for $\cS$ and $\cC^{*}$,  where $\cC^*=\{\cC_1^*,\ldots,C_{t_1}^*\}$ is as defined in Observation~\ref{obs:crucial}. 
\end{enumerate}
\end{defi}

Note that the event $\cE_1$ follows from Lemma~\ref{lem:hit}, $\cE_2$ follows from Lemma~\ref{lem:dist}, $\cE_3$ follows from Lemma~\ref{lem:struct}, and $\cE_4$ follows from Lemma~\ref{lem:mass}. Thus, from the respective guarantees of the aforementioned lemmas, we can say that $\pr(\cE_1), \pr(\cE_2), \pr(\cE_3), \pr(\cE_4)\geq \frac{99}{100}$.  
{To address a subtle point, note that Lemma~\ref{lem:dist} gives a probability lower bound on $R$ being distance preserving for any choice of $\cT$, and hence the lower bound also holds for $\cT$ sampled according to the distribution. Similarly, Lemma~\ref{lem:mass} provides a probability lower bound on $\cT$ being weight representative for any choice of $R$ (which affects $C^{*}$) regardless of whether $R$ is distance preserving, and hence the lower bound also holds for the $R$ chosen at random by the algorithm. So, we have the following lemma.}

\begin{lem}\label{eqn:cevent}
    $\pr\left(\mbox{{\sc GOOD}}\right)\geq \frac{2}{3}.$ 
\end{lem}



\subsection{Proof of Theorem~\ref{theo:mainthm1} (Correctness of \learn)}\label{sec:lem-inetr}

In the first three lemmas below (Lemma~\ref{cl:concfar}, Lemma~\ref{lem:stepv} and Lemma~\ref{lem:approxguarantee}), we prove the correctness of the internal steps of the algorithm. These lemmas are stated under the conditional space that the event $\mbox{{\sc GOOD}}$ defined in Definition~\ref{defi:eventgood} occurs. Using these lemmas along with Lemma~\ref{lem:distbyvec}, which helps us combine them, allows us to prove Theorem~\ref{theo:mainthm1}.


\begin{lem}[{\bf Guarantee till Step (v) of \learn}]\label{cl:concfar}
Assume that the $\mbox{event {\sc GOOD}}$ holds.
\begin{enumerate}
    \item[(i)]  If $D$ is $(\zeta, \delta, r)$-clusterable, then $D$ is $(\delta,2\zeta)$-clustered around $\mathcal{S}$, and the fraction of samples in $\mathcal{T}_y$ that are not assigned to any vector in $\mathcal{S}_x$ will be at most $3\zeta$. That is, \learn does not output $\mbox{{\sc Fail}}$ in Step (v) and proceeds to Step (vi).
    
    \item[(ii)] If $D$ is not $(3\delta,5\zeta)$-clustered around $\mathcal{S}$, then the fraction of samples in $\mathcal{T}_y$ that are not assigned to any vector in $\mathcal{S}_x$ will be at least  $3\zeta$. That is, \learn outputs $\mbox{{\sc Fail}}$ and does not proceed to Step (vi).
\end{enumerate}
\end{lem}

\begin{proof}



\begin{enumerate}
    \item[(i)] 
    
    
    For the first part, as $\cE_1$ holds (see Lemma~\ref{lem:hit}), the set $\cS$ contains at least one vector from every large cluster. Now, if we consider the $\delta$-neighborhood of $\cS$, that is, $\mbox{\sc NGB}_{\delta}(\cS)$, we infer  that all vectors in large clusters are in $\mbox{\sc NGB}_{\delta}(\cS)$. By the definition of a large cluster, the mass on the vectors that are not in any large cluster is at most $2 \zeta$. Hence, we conclude that $D(\mbox{\sc NGB}_{\delta}(\cS))\geq (1 - 2 \zeta)$. Thus, by Observation~\ref{obs:extra}, $D$ is $(\delta, 2\zeta)$-clustered around $\cS$. For the second part, as the event $\cE_4$ holds (see Lemma~\ref{lem:mass}(i)),  $\cT$ is weight representative for $\cS$. This follows since $D$ is $(\delta, 2\zeta)$-clustered, and in particular is $(3\delta, 5\zeta)$-clustered around $\cS$. Thus, $\frac{\size{\cT \cap \mbox{{\sc NGB}}_{\delta}(\cS) }}{\size{\cT}}  \geq D(\mbox{{\sc NGB}}_{ \delta}(\cS))-{\zeta}$. Also, as the event $\cE_2$ holds (see Lemma~\ref{lem:dist}), $R$ is distance preserving between $\cS$ and $\cT$, meaning that if $\mathbf{Y}_i$ in $\cC_j$, then $\mathbf{y}_i$ is assigned to $\mathbf{x}_j$. Hence, $$\sum_{i=1}^{t_1}w_i \geq  \frac{\size{\cT \cap \mbox{{\sc NGB}}_{\delta}(\cS)}}{\size{\cT}} \geq D(\mbox{{\sc NGB}}_{\delta}(\cS))-\zeta \geq 1-3\zeta.$$ That is, $w_0\leq 3\zeta$, and the algorithm \learn does not report {\sc Fail} and proceeds to Step (vi).

\item[(ii)] Since the distribution $D$ is not $(3\delta, 5\zeta)$-clustered around $\cS$, by Observation~\ref{obs:extra}, $D(\mbox{{\sc NGB}}_{3\delta}(\cS)) < 1-5 \zeta$. As the event $\cE_4$ holds (see Lemma~\ref{lem:mass} (ii)), $\frac{\size{\cT \cap \mbox{{\sc NGB}}_{3\delta}(\cS) }}{\size{\cT}}  \leq D(\mbox{{\sc NGB}}_{ 3 \delta}(\cS))+{\zeta} \leq 1-4\zeta$. Also, as the event $\cE_2$ holds (see Lemma~\ref{lem:dist}), $R$ is distance preserving between $\cS$ and $\cT$. This implies that $$\sum_{i=1}^{t_1}w_i \leq  \frac{\size{\cT \cap \mbox{{\sc NGB}}_{3\delta}(\cS)}}{\size{\cT}} \leq D(\mbox{{\sc NGB}}_{3\delta}(\cS))+\zeta < 1-3 \zeta .$$ That is, $w_0 > 3 \zeta$, and the algorithm \learn reports {\sc Fail}. So, \learn does not proceed to Step (vi).
\end{enumerate}
\end{proof}

\begin{lem}[{\bf Guarantee from Step (vi) of \learn}]\label{lem:stepv} 
Assume that the $\mbox{event {\sc GOOD}}$ holds. Also, assume that $D$ is $(3\delta,5\zeta)$-clustered around $\mathcal{S}$ and $w_0 \leq 3 \zeta$ holds in Step (vi) of \learn. Consider the following distribution $D''$ over $\{0,1\}^n$, constructed from the weights obtained from Step (vi) of \learn, such that 
\begin{enumerate}
    \item[(i)] For each $i \in [t_1]$,  $D''(\mathbf{X}_i) = w(\mathbf{x}_i)=w_i$.
    
    \item[(ii)] $D''(\mathbf{X}_0) = 1 - \sum\limits_{i=1}^{t_1} w({\mathbf{x}_i})$ for some arbitrary $\mathbf{X}_0$.
    
    \item[(iii)] $D''(\mathbf{X})=0$ for every $\mathbf{X} \in \{0,1\}^n \setminus \{\mathbf{X}_0, \ldots, \mathbf{X}_{t_1}\}$.
\end{enumerate}

Then $D''$ is $(5\delta,5 \zeta)$-clustered around $\cS$ with weights $w_0, \ldots, w_{t_1}$, where $w_0=1-\sum\limits_{i=1}^{t_1}w_i$, and the EMD between $D$ and $D''$ satisfies $d_{EM}(D, D'') \leq {10 \delta + 12 \zeta}$. 
\end{lem}

We will prove Lemma~\ref{lem:stepv} in Subsection~\ref{sec:pf-inter}. Now we proceed to prove the guarantee regarding Step (vii) of \learn.

\begin{lem}[{\bf Guarantee from Step (vii) of \learn}]
\label{lem:approxguarantee}
Assume that the $\mbox{event {\sc GOOD}}$ holds. Then, in Step (vii),
the algorithm \aprox (if called as described in Algorithm~\ref{alg:approx}) outputs a sequence of vectors $\{\mathbf{S}_1, \ldots, \mathbf{S}_{t_1}\}$ in $\{0,1\}^n$, such that there exists a permutation $\sigma:[n] \rightarrow [n]$ for which $d_H(\sigma(\mathbf{X}_i), \mathbf{S}_i) \leq \frac{\delta}{10}$ holds for every $i \in [t_1]$.
\end{lem}

\begin{proof}


Here we assume that the event $\mbox{{\sc GOOD}}$ holds. In particular, we assume that the event $\cE_3$ holds.

Let us consider a matrix $M$ of order $t_1\times n$ such that the $i$-th row vector corresponds to the vector $\mathbf{X}_i$. Then observe that $C_j^{\cS}$ represents the $j$-th column vector of matrix $M$ and $\alpha_J$ denotes the fraction of column vectors of $M$ that are identical to the vector $J$.


Let us consider the matrix $A$ of order $t_1\times n$ constructed by our algorithm,  by putting $\Gamma_J$ many column vectors identical to $J$, for every $J \in \{0,1\}^{t_1}$. Note that  $\{\mathbf{S}_1,\ldots,\mathbf{S}_{t_1}\}$ are the row vectors corresponding to $A$. As we are assuming that the event $\cE_3$ holds (see Lemma~\ref{lem:struct}),  $|\alpha_J - \frac{\Gamma_J}{n}| \leq \frac{\delta n}{10 \cdot 2^{t_1}}$ holds for every $J \in \{0,1\}^{t_1}$. Observe that we can permute the columns of the matrix $M$ using a permutation $\sigma: [n] \rightarrow [n]$ and create a matrix $M_\sigma$, such that there exists a bad set $I\subset[n]$ of size at most $\frac{\delta \cdot n}{10}$, where after the removal of the columns corresponding to indices of $I$ from both matrices $M_\sigma$ and $A$ become identical. Hence, we infer that $d_H(\sigma(\mathbf{X}_i), \mathbf{S}_i) \leq \frac{\delta}{10}$ for every $i \in [t_1]$, where $\sigma$ is the permutation corresponding to $M_\sigma$. This completes the proof of Lemma~\ref{lem:approxguarantee}.
\end{proof}

Finally, to prove Theorem~\ref{theo:mainthm1}, we need to show that the Earth Mover Distance between two distributions defined over close vectors is bounded when one distribution is clustered around a sequence of vectors and the other distribution has similar weights compared to the first distribution.  

\begin{lem}[{\bf EMD between distributions having close cluster centers}]\label{lem:distbyvec}
Let $\eta, \kappa, \xi \in (0,1)$ be three parameters such that $\eta+ \kappa+ \xi <1$. Suppose that $\cS= \{\mathbf{X}_1, \ldots, \mathbf{X}_{t_1}\}$ and $\cS'=\{\mathbf{X}'_1,\ldots,\mathbf{X}'_{t_1}\}$ are two sequences of vectors over $\{0,1\}^n$ such that $d_H(\mathbf{X}_i,\mathbf{X}'_i) \leq \kappa$ for every $i \in [t_1]$. Moreover, let $D$ be an $(\eta,\xi)$-clustered distribution around $\cS$ with weights $w_0, \ldots, w_{t_1}$ and $D'$ be another distribution such that $D'(\mathbf{X}'_i) \geq w_i$ for every $i \in [t_1]$. Then $d_{EM}(D, D') \leq \eta+\xi+\kappa$.


\end{lem}

\begin{proof}


Recall that the  EMD between $D$ and $D'$ is the solution to the following LP:
\begin{eqnarray*}
    &&\mbox{Minimize}\quad{\sum_{\mathbf{X},\mathbf{Y} \in \{0,1\}^n} f_{\mathbf{X}\mathbf{Y}} d_H(\mathbf{X},\mathbf{Y})}\\
    &&\mbox{Subject to}\quad\sum\limits_{\mathbf{Y} \in \{0,1\}^n} f_{\mathbf{X}\mathbf{Y}} = D(\mathbf{X})~\forall \mathbf{X} \in \{0,1\}^n \; \mbox{,} \;
    \sum\limits_{\mathbf{X} \in \{0,1\}^n} f_{\mathbf{X}\mathbf{Y}} = D'(\mathbf{Y})~\forall \mathbf{Y} \in \{0,1\}^n \\
    &&\mbox{and} ~~~~~~~~~~~~\quad0 \leq f_{\mathbf{X}\mathbf{Y}} \leq 1, ~~~~~\forall \, \mathbf{X}, \mathbf{Y} \in \{0,1\}^n.
\end{eqnarray*}

Here $D$ is $(\eta, \xi)$-clustered around $\cS$. Let $\cC_1,\ldots,\cC_{t_1}$ be the pairwise disjoint subsets of $\{0,1\}^n$ such that $\cC_i \subseteq \mbox{{\sc NGB}}_\eta(\mathbf{X}_i)$  and $D(\cC_i)\geq w_i$ for every $i \in [t_1]$. 

Consider a particular solution $\{f^*_{\bf X Y}:{\bf X ,Y}\in \{0,1\}^n\}$ that also satisfies the constraint 

\begin{center}
$\sum\limits_{{\bf X} \in \cC_i} f_{{\bf X} {\bf X}'_i} \geq  w_i$ for every $i \in [t_1]$.
\end{center}
  

The above constraint is feasible as  $D(\cC_i)\geq w_i$ and $D'({\bf X}_i')\geq w_i$, where $i \in [t_1]$. 

Now, 
\begin{eqnarray*}
\mbox{EMD}(D,D') &\leq& \sum\limits_{{\bf X},{\bf Y} \in \{0,1\}^n} f^*_{{\bf X Y}}d_H({\bf X},{\bf Y}) \\
&\leq& \sum\limits_{i=1}^{t_1} \sum\limits_{{\bf X} \in \cC_i } f^*_{{\bf X} {\bf X}_i'}d_H({\bf X},{\bf X}_i') + \sum\limits_{{\bf X} \notin \bigcup\limits_{i=1}^{t_1} \cC_i, {\bf Y} \in \{0,1\}^n} f^* _{{\bf X} {\bf Y}}d_H({\bf X},{\bf Y})\\
&\leq& \sum\limits_{i=1}^{t_1}w_i \cdot (\eta +\kappa) + w_0 \cdot 1\\ 
 &\leq& \eta + \kappa + \xi.
\end{eqnarray*}
\end{proof}

\subsection*{Proof of Theorem~\ref{theo:mainthm1} }\label{sec:main-pf}
To prove Theorem~\ref{theo:mainthm1}, we need the following lemma.
\begin{lem}\label{cl:distsigma}
If $D$ is $(3\delta,5\zeta)$-clustered around $S$, and \learn executes Step (vi), then $d_{EM}(D,D'_\sigma) \leq 17(\delta + \zeta)$ for some permutation $\sigma: [n] \rightarrow [n]$.
\end{lem}

\begin{proof}
As $D$ is $(3\delta, 5 \zeta)$-clustered around $\cS$, by Lemma~\ref{lem:stepv}, we have that $D''$ is $(5\delta,5\zeta)$-clustered around $\cS$ with weights $w_0, \ldots, w_{t_1}$ and $d_{EM}(D, D'') \leq 10 \delta + 12 \zeta.$

Now consider Step (vii) of \learn, where we call \aprox with $R$ and ${\bf x}_1,\ldots,{\bf x}_{t_1}$ to obtain $\mathbf{S}_1, \ldots, \mathbf{S}_{t_1}$. By  Lemma~\ref{lem:approxguarantee}, $d_H(\sigma(\mathbf{X}_i), \mathbf{S}_i) \leq \frac{\delta}{10}$ for every $i \in [t_1]$ for some permutation $\sigma: [n] \rightarrow [n]$. Consider the sequence of vectors $\mathbf{X}_1^{\sigma}\ldots,\mathbf{X}_{t_1}^{\sigma}$ where $\mathbf{X}_i^\sigma = \sigma(\mathbf{X}_i)$ for every $i \in [t_1]$.


Let us now consider the distribution $D''_\sigma$ over $\{0,1\}^n$ such that $D''_\sigma (\mathbf{X})=D''(\sigma(\mathbf{X}))$ for every $\mathbf{X} \in \{0,1\}^n$. As $D''$ is $(5\delta, 5\zeta)$-clustered around $\cS=\{\mathbf{X}_1,\ldots, \mathbf{X}_{t_1}\}$ with weights $w_0,\ldots,w_{t_1}$, $D''_\sigma$ is $(5\delta, 5\zeta)$-clustered around $\{\mathbf{X}_1^\sigma,\ldots,\mathbf{X}_{t_1}^\sigma\}$ with weights $w_0,\ldots,w_{t_1}$. In the output distribution $D'$, $D'(\mathbf{S}_i) \geq w_i$ for every $i \in [t_1]$. So, 
by Lemma~\ref{lem:distbyvec}, we have $d_{EM}(D',D''_\sigma)\leq 5 \delta + \frac{\delta}{10} + 5 \zeta$. Combining this with the fact that $d_{EM}(D, D'') \leq 10 \delta + 12 \zeta$, \complain{we conclude that $d_{EM}(D,D_\sigma')\leq 17 (\delta+\zeta)$}.
\end{proof}

To prove Theorem~\ref{theo:mainthm1}, we first prove that the guarantees of the two parts follow assuming that the event {\sc GOOD} holds. We will be done since $\pr\left(\mbox{{\sc GOOD}}\right)\geq 2/3$ (see Lemma~\ref{eqn:cevent}). The query complexity of the algorithm follows from the parameters in its description.

\begin{proof}[Proof of Part (i):]

Here 
$D$ is $(\zeta, \delta,r)$-clusterable. By Lemma~\ref{cl:concfar},  $D$ is $(\delta,2\zeta)$-clustered around $\cS$ and the fraction of samples in $\cT_y$ that are not assigned to any vector in $\cS_x$ is at most $3 \zeta$. That is, \learn does not output $\mbox{{\sc Fail}}$ for $D$ in Step (v). By Lemma~\ref{cl:distsigma}, we conclude that $d_{EM}(D,D_\sigma')\leq 17(\delta + \zeta)$ for some permutation $\sigma: [n] \rightarrow [n]$. This completes the proof of Part (i).
\end{proof}

\begin{proof}[Proof of Part (ii):]


Recall that we are working under the conditional space that the $\mbox{event {\sc GOOD}}$ holds. Now consider the following:
\begin{itemize}
    \item If $D$ is not $(3\delta,5\zeta)$-clustered around $\cS$, then by Lemma~\ref{cl:concfar}, the algorithm \learn reports {\sc Fail}.
    
    \item If $D$ is $(3\delta,5\zeta)$-clustered around $\cS$, then the algorithm \learn either reports {\sc Fail} in Step (v) or continues to Step (vi). In case we go to Step (vi), following Lemma~\ref{cl:distsigma}, we again conclude that $d_{EM}(D,D'_\sigma)\leq \eps$.
\end{itemize}

Observe that the above two statements imply Part (ii). This completes the proof of Theorem~\ref{theo:mainthm1}.
\end{proof} 
 

\subsection*{Proof of Lemma~\ref{lem:stepv}}
\label{sec:pf-inter}

Here we assume that the event $\mbox{{\sc GOOD}}$ holds. In particular, the events $\cE_2$ and $\cE_4$ hold. To prove Lemma~\ref{lem:stepv}, we will prove some associated claims and lemmas about the weights $w_0,\ldots,w_{t_{1}}$ obtained in Step (vi) of \learn, and the distribution $D''$ defined in Lemma~\ref{lem:stepv}. Let us start with the following claim.

\begin{cl}\label{cl:distD}
The distribution $D^{''}$ (defined in the statement of Lemma~\ref{lem:stepv}) is $(5\delta, 5\zeta)$-clustered around $\cS$ with weights $w_0,w_1\ldots,w_{t_1}$, where $w_0 =1- \sum\limits_{i=1}^{t_1} w_i$.
\end{cl}

\begin{proof}
This follows from the definition of $D''$, along with the fact that $w_0 \leq 3 \zeta < 5 \zeta$.
\end{proof}



Now we have the following claim.

\begin{cl}\label{cl:existw}
There exists a sequence of weights $w'_0,\ldots,w_{t_1}'$ such that $D$ is $(5\delta, 5\zeta)$-clustered around $\cS$ with weights $w'_0,\ldots,w_{t_1}'$, and 
$\sum\limits_{i=1}^{t_1}\size{w_i-w'_{i}}\leq 2 \zeta$.


\end{cl}


\begin{proof}
As events $\cE_2$ and $\cE_4$ hold, consider $\cC^*=\{\cC_1^*,\ldots,\cC_{t_1}^*\}$ (as guaranteed by Observation~\ref{obs:crucial}) and $\cC^{**}=\{\cC_1^{**},\ldots,\cC_{t_1}^{**}\}$ such that, for every $i \in [t_1]$, $ \cC_i^{**}= \cC_i^{*} \cap \mbox{{\sc NGB}}_{3\delta}(\mathbf{X}_i)$ and $ w_i \leq D(\cC_i^{**})+ \frac{2\zeta}{t_1}$ (see Observation~\ref{obs:crucial-2}).


Let us define $w_i'=\max\{w_i-\frac{2 \zeta}{t_1},0\}$ and $w_0' =1-\sum\limits_{i=1}^{t_1}w_i'$. So, $w_i' \leq D\left(\cC_i^{**}\right) $.

Now $$w_0'=1-\sum\limits_{i=1}^{t_1}w_i'\leq 1- \sum\limits_{i=1}^{t_1}\left(w_i- \frac{2\zeta}{t_1}\right) \leq \left(w_0+2\zeta\right) \leq 3 \zeta +2 \zeta= 5 \zeta. $$




Putting everything together, the above $\cC^{**}$ satisfies $\cC_i^{**} \subseteq \mbox{{\sc NGB}}_{3\delta}(\mathbf{X}_i) \subseteq \mbox{{\sc NGB}}_{5\delta}(\mathbf{X}_i)$ and has weights $w_0',\ldots,w_{t_1}'$ such that $w_0' \leq 5 \zeta$ and $w_i'\leq D(\cC_i^{**})$ for every $i \in [t_1]$.
Hence, $D$ is $(5\delta, 5 \zeta)$-clustered around $\cS$ with weights $w_0',\ldots,w_{t_1}'$. Moreover, $\sum\limits_{i=1}^{t_1}\size{w_i-w'_{i}}\leq 2 \zeta$ holds following the definition of $w_i'$s.
\color{black}
\end{proof}

\color{black}

\begin{lem}[{\bf Comparison-by-weights}]\label{lem:comparebywt}
Let $D_1$ and $D_2$ be two distributions defined over $\{0,1\}^n$ that are $(\eta, \xi)$-clustered around a sequence of vectors $\cS= \{\mathbf{X}_1, \ldots, \mathbf{X}_{t_1}\}$ with weights $v_0,\ldots,v_{t_1}$ and $w_0,\ldots,w_{t_1}$, respectively.
Then the Earth Mover Distance between $D_1$ and $D_2$ is $d_{EM}(D_1, D_2) \leq 2\eta+\sum\limits_{i=1}^{t_1}|v_i-w_i| + 2\xi$.
\end{lem}

\begin{proof}


Let $\mathbf{U}$ be an arbitrary vector from $\{0,1\}^n$. Let us define a distribution $D_1'$ (supported over $\cS \cup \{\mathbf{U}\}$) from the distribution $D_1$ as follows:
\[ D_1'(\mathbf{Y})=  \left\{
\begin{array}{ll}
       v_i & \mathbf{Y} = \mathbf{X}_i \ \mbox{for every $i \in [t_1]$} \\
      1 - \sum\limits_{i=1}^{t_1}v_i & \mathbf{Y}=\mathbf{U} \\
      0 &\mbox{otherwise}
\end{array} 
\right. \]

Similarly, we define a distribution $D_2'$ from $D_2$.
First we have the following claim, which follows from the definitions. From the definitions of $D_1'$ and $D_2'$, we can say that 
\begin{enumerate}
    \item[(i)] $d_{EM}(D_1, D_1') \leq \eta + \xi$ and $d_{EM}(D_2, D_2') \leq \eta + \xi$ (by Lemma~\ref{lem:distbyvec}).
    
    \item[(iii)] $d_{EM}(D_1', D_2')\leq \sum\limits_{i=1}^{t_1} |v_i-w_i|$.
\end{enumerate}

Using the triangle inequality, we have
\begin{eqnarray*}
d_{EM}(D_1, D_2) &\leq& d_{EM}(D_1, D_1') + d_{EM}(D_1', D_2') + d_{EM}(D_2, D_2') \\ &\leq& 2 \eta + \sum\limits_{i=1}^{t_1} |v_i-w_i| + 2 \xi.
\end{eqnarray*}

This completes the proof of Lemma~\ref{lem:comparebywt}.
\end{proof}

Now we proceed to prove Lemma~\ref{lem:stepv}.

\begin{proof}[Proof of Lemma~\ref{lem:stepv}]


By the description of $D''$ in Lemma~\ref{lem:stepv}, using Claim~\ref{cl:distD}, we know that $D''$ is $\left(5 \delta, 5 \zeta\right)$-clustered around $\cS$ with weights $w_1, \ldots, w_{t_1}$. By applying Claim~\ref{cl:existw}, $D$ is $(5 \delta,5\zeta)$-clustered around $\cS$ with weights $w_0',\ldots,w_{t_1}'$ such that $\sum\limits_{i=1}^{t_1}\size{w_i-w_i'}\leq 2\zeta$. Now, by applying Lemma~\ref{lem:comparebywt} with $\eta=5\delta$, $\xi=5\zeta$, we obtain that the Earth Mover Distance between $D$ and $D''$ is bounded as follows:
$$d_{EM}(D, D'') \leq 10 \delta + 2 \zeta + 10 \zeta \leq 10 \delta + 12 \zeta .$$

This completes the proof of Lemma~\ref{lem:stepv}.
\end{proof}

%% file: algo.tex
\begin{algorithm}

\SetAlgoLined

\caption{\learn}
\label{alg:testcluster--}
\KwIn{Sample and Query access to a distribution $D$ over $\{0,1\}^n$, and parameters $\zeta,\delta, r$ with $\zeta ,\delta \in (0,1)$ and $r \in \N$. }
\KwOut{Either reports a full description of a distribution over $\{0,1\}^n$ or {\sc Fail}, satisfying (i) and (ii) as stated in Theorem~\ref{theo:mainthm1}.}
\begin{description}
\item[(i)] Take $t_1=\Oh(\frac{r}{\zeta} \log \frac{r}{\zeta})$ samples $\mathcal{S} = \mathbf{X}_1, \dots, \mathbf{X}_{t_1}$ from $D$.

\item[(ii)] Take $t_2=\Oh(\frac{t_1^2}{\zeta^2} \log t_1)$ samples $\cT=\mathbf{Y}_1, \dots, \mathbf{Y}_{t_2}$ from $D$.

\item[(iii)]  Pick a random subset $R \subset [n]$ with $\size{R}=\Oh(\frac{4^{t_1}}{\delta^2\zeta} \log \frac{r}{\delta\zeta})$. Query the indices corresponding to $R$ in each sample of $\cS$, to obtain the sequence of vectors $\mathcal{S}_x = \mathbf{x}_1, \dots, \mathbf{x}_{t_1}$,  where $\mathbf{x}_{i} = \mathbf{X}_i\mid_R$ for each $i \in [t_1]$. Also,  query the indices corresponding to $R$ in each sample in $\cT$, to obtain the sequence of vectors $\mathcal{T}_y = \mathbf{y}_1, \dots, \mathbf{y}_{t_2}$, where $\mathbf{y}_{j} = \mathbf{Y}_j\mid_R$ for every $j \in [t_2]$.

\item[(iv)] For each $j \in \{ 1, \ldots, t_2\}$,  if there exists an $i \in [t_1]$ such that $d_H(\mathbf{y}_j, \mathbf{x}_i)\leq 2 \delta$, assign $\mathbf{y}_j$ to $\mathbf{x}_i$, breaking ties by assigning $\mathbf{y}_j$ to the vector in $\mathcal{S}_x$ with the minimum index.

If for some $\mathbf{y}_j$ no suitable $\mathbf{x}_i$ is found, then $\mathbf{y}_j$ remains unassigned.

\item[(v)] If the total number of unassigned vectors in $\mathcal{T}_y$ is more than $3\zeta t_2$, then output {\sc Fail}.

\item[(vi)] For every $i \in \{1, \dots, t_1\}$, the weight of $\mathbf{x}_i$ is defined as 
$$w_i=w(\mathbf{x}_i)=\frac{\mbox{Number of vectors in } \mathcal{T}_y \mbox{ assigned to } \mathbf{x}_i}{t_2}.$$

\item[(vii)] Use \aprox (as described in Algorithm~\ref{alg:approx}) with $R$ and $\mathbf{x}_1, \ldots, \mathbf{x}_{t_1}$ to obtain $\mathbf{S}_1, \ldots, \mathbf{S}_{t_1} \in \{0,1\}^n$ (as stated in Lemma~\ref{lem:approxguarantee}).

    
    
\item[(viii)] Construct and return any distribution $D'$ over $\{0,1\}^n$ such that 
\begin{itemize}
    \item For each $i = 1, \dots, t_1$,  $D'(\mathbf{S}_i) \geq w(\mathbf{x}_i)$.
   
   \item $\sum\limits_{i=1}^{t_1} D'(\mathbf{S}_i)=1$.
   
   \item $D'(\mathbf{S})=0$ for every $\mathbf{S} \in \{0,1\}^n \setminus \{\mathbf{S}_1, \ldots, \mathbf{S}_{t_1}\}$.
\end{itemize}
\end{description}
\end{algorithm}

\begin{algorithm}

\SetAlgoLined

\caption{\aprox}
\label{alg:approx}
\KwIn{A random subset $R \subseteq [n]$ with $\size{R}=\Oh(\frac{4^{t_1}}{\delta^2\zeta} \log \frac{r}{\delta\zeta})$, and a sequence of vectors $\mathbf{x}_1, \ldots,\mathbf{x}_{t_1} \in \{0,1\}^{\size{R}}$ drawn from the distribution $D\mid_{R}$.}
\KwOut{Sequence of vectors $\mathbf{S}_1, \ldots, \mathbf{S}_{t_1}$ such that with probability at least $99/100$ over the random choice of $R$, for every $i \in [t_1]$, $d_H(\sigma(\mathbf{X}_i), \mathbf{S}_i) \leq \delta/10$, where $\sigma: [n] \rightarrow [n]$ is a permutation.}
\LinesNumbered
\begin{description}
     \item[(i)] For each $i \in R$,  construct the vector $\mathbf{C}_i \in \{0,1\}^{t_1}$ such that $\mathbf{C}_i(j) = \mathbf{x}_j(i)$.

    \item[(ii)] For any $J\in \{0,1\}^{t_1}$, determine 
    $\gamma_J=\frac{|\{i\in R\mid \mathbf{C}_i = J\}|}{\size{R}}$.

    \item[(iii)] Apply Observation~\ref{obs:integral}, to obtain for any $J\in \{0,1\}^{t_1}$ an approximation $\Gamma_J$, such that $\Gamma_J \in \{\lfloor \gamma_J \cdot n\rfloor, \lceil \gamma_J \cdot n\rceil \}$ and $\sum\limits_{J\in \{0,1\}^{t_1}}\Gamma_J = n$.
    
    \item[(v)]
    Construct a matrix $A$ of dimension $t_1 \times n$ by putting $\Gamma_J$ many $J$ column vectors, for each $J \in \{0,1\}^{t_1}$. 
    
    \item[(vi)] Return the row vectors of $A$ as $\mathbf{S}_1,\ldots,\mathbf{S}_{t_1}$.


    
    

\end{description}

\end{algorithm}

%% file: vc_cluster.tex
\section{Testing properties with bounded VC-dimension }\label{sec:vc}


{In this section, we will prove that 
distributions over $\{0,1\}^n$ whose support have bounded VC-dimension can be learnt (up to permutations) by performing a number of queries that is independent of the dimension $n$, and depends only on the proximity parameter $\eps$ and the VC-dimension $d$ (Theorem~\ref{theo:main}). In fact, we will prove a generalization, that any distribution $D$ that is $\beta$-close to bounded VC-dimension can be learnt efficiently up to permutations (with a proximity parameter depending on $\beta$) by performing a set of queries whose size is independent of $n$ (Theorem~\ref{theo:learn-vc}). As a consequence of the learning result of Theorem~\ref{theo:learn-vc}, we also obtain a tester for properties having a bounded VC-dimension (Corollary~\ref{coro:vc-learn2}) which is a restatement of Corollary~\ref{coro:testvc}.}

{In Subsection~\ref{sec:corotheo-vc}, we connect the notions of $(\zeta, \delta,r)$-clusterablity  and being $\beta$-close to $(\alpha,r)$-clusterablity (Definition~\ref{defi:cluster}) in Lemma~\ref{lem:cluster12} and prove Corollary~\ref{coro:mainthm1} regarding learning distributions that are $\beta$-close to $(\alpha,r)$-clusterable. Then, in Subsection~\ref{sec:learn-vcprelim}, we recall some standard results from VC theory to connect the notions of bounded VC-dimension and clusterability, to obtain Corollary~\ref{coro:cover}, which is crucially used in Subsection~\ref{sec:learn-vc} to prove Theorem~\ref{theo:learn-vc}.}


\begin{theo}[{\bf Learning a distribution $\beta$-close to bounded {VC-dimension}}]\label{theo:learn-vc}
Let $d\in \N$ be a constant. 
There exists a (non-adaptive) algorithm, that given sample and query access to an unknown distribution $D$ over $\{0,1\}^n$, takes $ \alpha, \beta \in (0,1)$ with $\beta <\alpha$ as input such that $\eps=17(3 \alpha+\beta/\alpha)<1$, makes number of queries that depends only on $\alpha,\beta$ and $d$, and either reports a full description of a distribution, or {\sc Fail}, satisfying both of the following conditions:
\begin{enumerate}
\item[(i)] If $D$ is $\beta$-close to {VC-dimension} $d$, then with probability at least $2/3$, the algorithm outputs a distribution $D'$ such that $d_{EM}(D,D'_\sigma)\leq \eps$ for some permutation $\sigma:[n]\rightarrow [n]$.

\item[(ii)] For any $D$, the algorithm will not output a distribution $D'$ such that $d_{EM}(D,D'_\sigma) > \eps$ for every permutation $\sigma:[n]\rightarrow [n]$ with probability more than $\frac{1}{3}$. However, if the distribution $D$  is not $\beta$-close to {VC-dimension} $d$, the algorithm may output {\sc Fail} with any probability.
\end{enumerate}
\end{theo}


\begin{rem}
Note that $\alpha$ above does not appear anywhere outside the expression for $\eps$, and hence it is tempting to minimize $\eps$ by taking $\alpha=\sqrt{\beta/3}$. However, this is a bad strategy since the number of queries of the algorithm depends on $1/\alpha$.  In the common scenario, we would be given $\beta$ and $\eps \geq 34\sqrt{3\beta}$, and solve for $\alpha$.    
\end{rem}

\color{black}

\begin{coro}
[{\bf Testing properties with bounded VC-dimension}]\label{coro:vc-learn2} 
Let $d \in \N$ be a constant, and $\cP$ be an index-invariant property with VC-dimension $d$. There exists an algorithm that has sample and query access to an unknown distribution $D$, takes a parameter $\eps \in (0,1)$, and distinguishes whether $D \in \cP$ or $D$ is $\eps$-far from $\cP$ with probability at least $2/3$, where the total number of queries made by the algorithm is a function of only $d$ and $\eps$.
\end{coro}


\begin{rem}
Note that the algorithm for testing the index-invariant property with constant VC-dimension $d$ takes $\exp(d)$ samples, and performs $\exp(\exp(d))$ queries. It turns out that similarly to the case of \learn, the dependencies of the sample and query complexities on $d$ are tight, in the sense that there exists a property of VC-dimension $d$ such that testing it requires $2^{\Omega(d)}$ samples, and $\Omega(2^{2^{d- \Oh(1)}})$ queries. We will construct such a property and prove its lower bound in Section~\ref{sec:bounded-vc}.
\end{rem}

\color{black}

We will give the proof of Theorem~\ref{theo:learn-vc} in Subsection~\ref{sec:learn-vc}.

\subsection{A corollary of Theorem~\ref{theo:mainthm1} to prove Theorem~\ref{theo:learn-vc}}\label{sec:corotheo-vc}

{In this subsection, we first connect the notions of $(\zeta, \delta,r)$-clusterablity  and being $\beta$-close to $(\alpha,r)$-clusterablity (Definition~\ref{defi:cluster}) in Lemma~\ref{lem:cluster12}. Then using Lemma~\ref{lem:cluster12} with our algorithm for learning $(\zeta,\delta,r)$-clusterable distributions (Theorem~\ref{theo:mainthm1}), we prove Corollary~\ref{coro:mainthm1} regarding learning distributions that are $\beta$-close to $(\alpha,r)$-clusterable. This corollary will be used later to prove Theorem~\ref{theo:learn-vc}.}

\begin{coro}[{\bf Learning distributions $\beta$-close to  $(\alpha,r)$-clusterable}]\label{coro:mainthm1}
Let $n \in \N$. There exists a (non-adaptive) algorithm, that has sample and query access to an unknown distribution $D$ over $\{0,1\}^n$, takes parameters $\alpha, \beta, r$ as inputs such that $\alpha > \beta$ and $\eps=17(3\alpha+\beta/\alpha)<1$ and $r \in \N$, makes a number of queries that only depends on $\alpha,\beta$ and $r$, and either reports a full description of a distribution over $\{0,1\}^n$ or reports {\sc Fail}, satisfying both of the following conditions:
\begin{enumerate}
    \item[(i)] If $D$ is $\beta$-close to $(\alpha, r)$-clusterable, then with probability at least $2/3$, the algorithm outputs a full description of a distribution $D'$ over $\{0,1\}^n$ such that $d_{EM}(D, D'_{\sigma}) \leq \eps$ for some permutation $\sigma:[n] \rightarrow [n]$.
    
    \item[(ii)] For any $D$, the algorithm will not output a distribution $D'$ such that $d_{EM}(D,D'_\sigma) > \eps$ for every permutation $\sigma:[n]\rightarrow [n]$, with probability more than $1/3$. However, if the distribution $D$  is not $\beta$-close to $(\alpha, r)$-clusterable, the algorithm may output {\sc Fail} with any probability.
\end{enumerate}
\end{coro}


To prove the above corollary, we need the following lemma, that connects the two notions of clusterability, that is, $(\zeta,\delta,r)$-clusterablity and being $\beta$-close to $(\alpha,r)$-clusterability (see Definition~\ref{defi:cluster}).

\begin{lem}\label{lem:cluster12}
Let $\alpha, \beta \in (0,1)$ be such that $\alpha > \beta$, and $D$ be a distribution over $\{0,1\}^n$ that is $\beta$-close to being $(\alpha, r)$-clusterable. Then $D$ is $\left(3\alpha, r,\beta/\alpha \right)$-clusterable.
\end{lem}

\begin{proof}
Let $D_0$ be the distribution such that $D_0$ is $(\alpha,r)$-clusterable and $d_{EM}(D,D_0)\leq \beta$. Let $\cC_1, \ldots, \cC_s$ be the partition of the support of $D_0$ that realizes the $(\alpha,r)$-clusterability of $D_0$, and let $\{f_{\mathbf{XY}}:\mathbf{X,Y}\in \{0,1\}^n\}$ be the flow that realizes $d_{EM}(D,D_0)\leq \beta$.

Let $\cC=\bigcup\limits_{i=1}^s \cC_i$, and $\cC_{>\alpha}$ be the set of vectors in $\{0,1\}^n$ that have distance of at least $\alpha$ from all the vectors in $\cC$. Now we have the following claim.

\begin{cl}\label{cl:clusterdelta}
$D(\cC_{>\alpha}) \leq \frac{\beta}{\alpha}$.
\end{cl}

\begin{proof}
By contradiction, let us assume that $D(\cC_{>\alpha}) > \frac{\beta}{\alpha}$. Then we have the following:
\begin{eqnarray*}
d_{EM}(D,D_0) &\geq& \sum\limits_{\mathbf{X} \in \cC _{> \alpha} , \mathbf{Y} \in \cC} f_{\mathbf{XY}}d_H(\mathbf{X},\mathbf{Y}) 
\geq \alpha \cdot D(\cC _{> \alpha}) > \beta.
\end{eqnarray*}
This is a contradiction as we have assumed $d_{EM}(D,D_0)\leq \beta$.
\end{proof}

Now for every $i$, let $\cC_i^{\leq \alpha}$ be the vectors that have distance at most $\alpha$ from at least one vector $\cC_i$, where $i \in [s]$. Let $\cC_i'=\cC_i^{\leq \alpha} \setminus \bigcup\limits_{j=1}^{i-1} \cC_j'$ for $1 \leq i \leq s$. Now we have the following observation.


\begin{obs}
For any $1 \leq i \leq s$,  $d_H(\mathbf{U},\mathbf{V})\leq 3\alpha$  for any $\mathbf{U},\mathbf{V} \in \cC_i'$.    
\end{obs}

\begin{proof}
Since $\mathbf{U}, \mathbf{V} \in \cC_i'$, let $\mathbf{U'}$ and $\mathbf{V'}$ be the vectors in $\cC_i$ such that $d_H(\mathbf{U}, \mathbf{U'}) \leq \alpha$, and $d_H(\mathbf{V}, \mathbf{V'}) \leq \alpha$. As $\mathbf{U'}, \mathbf{V'} \in \cC_i$, and $D_0$ is $(\alpha, r)$-clusterable, using the triangle inequality, we can say that $d_H(\mathbf{U}, \mathbf{V}) \leq d_H(\mathbf{U}, \mathbf{U'}) + d_H(\mathbf{U'}, \mathbf{V'}) + d_H(\mathbf{V'}, \mathbf{V}) \leq 3\alpha$.
\end{proof}

\color{black}
 

Consider $\cC_0'=\cC_{>\alpha}$, and by Claim~\ref{cl:clusterdelta}, note that $D(\cC_0')\leq \beta/\alpha$. The existence of $\cC_0',\cC_1',\ldots,\cC_s'$ as above implies that $D$ is $\left(3\alpha,r, \beta/\alpha \right)$-clusterable (see Definition~\ref{defi:cluster}).
\end{proof}

\begin{proof}[Proof of Corollary~\ref{coro:mainthm1} using Theorem~\ref{theo:mainthm1} and Lemma~\ref{lem:cluster12}]
The algorithm here (say {\sc ALG}) calls algorithm \learn (as described in Algorithm~\ref{alg:testcluster--}) with parameters $\zeta=\beta/\alpha$ and $\delta=3\alpha$, and reports the output returned by \learn as the output of {\sc ALG}. Now we prove the correctness of {\sc ALG}. 

\begin{description}
\item[Part (i):] \complain{Here we consider the case where $D$ is $\beta$-close to $(\alpha,r)$-clusterable. By Lemma~\ref{lem:cluster12}, $D$ is $(\zeta, \delta,r)$-clusterable}. By Theorem~\ref{theo:mainthm1} (i), we get a distribution $D'$  such that $d_{EM}(D,D_\sigma')\leq 17(\zeta+\delta) =17(3\alpha+\beta/\alpha)=\eps$ for some permutation $\sigma:[n]\rightarrow [n]$, with probability at least $2/3$. This completes the proof of Part (i).
 
\item[Part (ii):] This follows from Theorem~\ref{theo:mainthm1} (ii) along with our choices of  $\delta=3\alpha$ and $\zeta=\beta/\alpha$.
\end{description}
\end{proof}

\color{black}

\subsection{A corollary from VC theory required to prove Theorem~\ref{theo:learn-vc}}\label{sec:learn-vcprelim}

{In this subsection, we recall some definitions from VC-dimension theory, and use a well known result of Haussler~\cite{haussler1995sphere} to obtain Corollary~\ref{coro:cover}, which states that if the VC-dimension of a set of vectors $V$ is bounded, then the vectors of $V$ can be covered by bounded number of Hamming balls. This corollary will be crucially used to prove Theorem~\ref{theo:learn-vc} in Subsection~\ref{sec:learn-vc}.}


Let us start by defining the notion of an {$\alpha$-separated set}.

\begin{defi}[{\bf $\alpha$-separated set}]
Let $\alpha \in (0,1)$ and $W \subset \{0,1\}^n$ be a set of vectors. $W$ is said to be \emph{$\alpha$-separated} if for any two vectors $\mathbf{X}, \mathbf{Y} \in W$, $d_{H}(\mathbf{X}, \mathbf{Y}) \geq \alpha$.
\end{defi}

Now let us define the notion of the $\alpha$-packing number of a set of vectors.

\begin{defi}[{\bf $\alpha$-packing number}]
Let $\alpha \in (0,1)$, and $V \subset \{0,1\}^n$ be a set of vectors. The \emph{$\alpha$-packing number} of $V$, denoted by $\cM(\alpha, V)$, is defined as the cardinality of the largest $\alpha$-separated subset $W$ of $V$.
\end{defi}

Now we define the notion of an $\alpha$-cover of a set of vectors.

\begin{defi}[{\bf $\alpha$-cover}]
Let $\alpha \in (0,1)$ and $V \subset \{0,1\}^n$ be a set of vectors. A set $M\subseteq V$ is an \emph{$\alpha$-cover} of $V$ if $V \subseteq \bigcup\limits_{{\bf p} \in M} \mbox{{\sc NGB}}_{\alpha}({\bf p})$,
where $\mbox{{\sc NGB}}_{\alpha}({\bf p}) : = \left\{ \mathbf{q} \, : \, d_{H}(\mathbf{p},\mathbf{q}) \leq \alpha\right\}$ denotes the set of vectors that are within Hamming distance $\alpha$ from the vector $\mathbf{p}$.
\end{defi}




Now let us consider the following theorem from \cite{haussler1995sphere}, which says that if the VC-dimension of a set of vectors $V$ is $d$, then the size of the $\alpha$-packing number of $V$, that is, $\cM(\alpha, V)$, is bounded by a function of $d$ and $\alpha$.

\begin{theo}[{\bf Haussler's packing theorem~\cite[Theorem~1]{haussler1995sphere}}]\label{haussler1995sphere}
Let $\alpha \in (0,1)$ be a parameter. If the VC-dimension of a set of vectors $V$ is $d$, then the $\alpha$-packing number of $V$ is bounded as follows:
$$\cM(\alpha, V) \leq e(d+1)\left(\frac{2e}{\alpha}\right)^d
$$
\end{theo}


The following observation is immediate.

\begin{obs}
Let $\alpha \in (0,1)$ be a parameter and $M$ be a maximal $\alpha$-packing of a set of vectors $V \subset \{0,1\}^{n}$. Then $M$ is also an $\alpha$-cover of $V$.
\end{obs}


With this observation, along with Theorem~\ref{haussler1995sphere},
we get the 
following bound on the size of a cover of a set of vectors in terms of its VC-dimension. 

\begin{coro}[{\bf Existence of a small $\alpha$-cover}]\label{coro:cover}
Let $d \in \N$. If the VC-dimension of a set of vectors $V$ is $d$, then for all $\alpha \in (0,1)$,
there exists a set $M \subseteq V$ such that $M$ is an $\alpha$-cover of $V$ and
$|M| \leq e(d+1)\left(\frac{2e}{\alpha}\right)^d$.
\end{coro}

\subsection{Proof of Theorem~\ref{theo:learn-vc} and testing bounded VC-dimension 
properties}\label{sec:learn-vc}

{In this subsection, using Corollary~\ref{coro:mainthm1}, we prove that any distribution that is $\beta$-close to bounded VC-dimension can be learnt (up to permutation) by performing a number of queries that depends only on the VC-dimension $d$ and the proximity parameter $\eps$, and is independent of the dimension of the Hamming cube $\{0,1\}^n$ (Theorem~\ref{theo:learn-vc}). The crucial ingredient of the proof is Theorem~\ref{haussler1995sphere}, through its Corollary~\ref{coro:cover}.
From Theorem~\ref{theo:learn-vc}, we obtain a tester for testing distribution properties with bounded VC-dimension (Corollary~\ref{coro:vc-learn2}).}



\begin{proof}[Proof of Theorem~\ref{theo:learn-vc}]
\complain{We call the algorithm {\sc ALG} corresponding to Corollary~\ref{coro:mainthm1} with $D$ as the input distribution, the same $\alpha$ and $\beta$ as here, and $r=\lfloor e(d+1)\left(\frac{2e}{\alpha}\right)^d \rfloor$. Note that the output of {\sc ALG} is either the full description of a distribution $D'$ or {\sc Fail}. We output the same output returned by {\sc ALG}. Now we prove the correctness of this procedure.}
\begin{itemize}
\item[(i)] \complain{
Here $D$ is $\beta$-close to having {VC-dimension} $d$. Let $D_0$ be the distribution such that $D_0$ has VC-dimension at most $d$ and $d_{EM}(D,D_0)\leq \beta$. By Corollary~\ref{coro:cover}, we can partition the support of $D_0$ into $r$ parts $\cC_1,\ldots,\cC_r$ such that $r\leq e(d+1)\left(\frac{2e}{\alpha}\right)^d$ and the Hamming distance between any pair of vectors in the same cluster $\cC_i$ is at most $\alpha$. This means that $D_0$ is $(\alpha,r)$-clusterable.
So, with probability at least 2/3, \learn outputs a distribution $D'$ such that $d_{EM}(D,D'_\sigma)\leq 17(3\alpha + \beta/\alpha)$ for some permutation $\sigma : [n] \rightarrow [n]$, and we are done with the proof.}
\item[(ii)] \complain{This follows from the guarantee provided by \learn, see Corollary~\ref{theo:mainthm1} (ii).}
\end{itemize}      
\end{proof}      

We conclude this section with the proof of Corollary~\ref{coro:vc-learn2} regarding the testing of properties with bounded VC-dimension.


\begin{proof}[Proof of Corollary~\ref{coro:vc-learn2}]
We call the algorithm (say {\sc ALG}) corresponding to Theorem~\ref{theo:learn-vc} with the input distribution $D$, $\alpha=\eps/102$, and $\beta=0$. Let $D'$ be the output of {\sc ALG}.  We check if there exists a distribution $D'' \in \cP$ such that $d_{EM}(D',D'') \leq \eps/2$. If yes, we accept $D$. Otherwise, we reject $D$.

Now we argue the correctness. For completeness, let us assume that $D \in \cP$, hence $D$ has VC-dimension $d$. By the guarantee for {\sc ALG} following Theorem~\ref{theo:learn-vc}, with probability at least $2/3$, {\sc ALG} does not report {\sc Fail}, and the output distribution $D'$ by {\sc ALG} satisfies $d_{EM}(D,D'_\sigma)\leq \eps/2$ for some permutation $\sigma:[n]\rightarrow [n]$. Since $\cP$ is an index-invariant property, $D'$ and $D'_\sigma$ have the same distance from the property $\cP$. Also, as $D \in \cP$, $D_\sigma \in \cP$ as well. Hence, there exists a distribution $D'' \in \cP$ (here $D_\sigma$ in particular) such that $d_{EM}(D',D'') \leq \eps/2$, and we accept $D$ with probability at least $2/3$.

For soundness, consider the case where $D$ is $\eps$-far from $\cP$. If {\sc ALG} reports {\sc Fail}, we are done. Otherwise, by Theorem~\ref{theo:learn-vc}, the output distribution $D'$ is such that $d_{EM}(D,D'_\sigma)\leq \eps/2$ for some permutation $\sigma:[n]\rightarrow [n]$. Now we consider any distribution $D''$ with $d_{EM}(D',D'')\leq \eps/2$ and argue that $D''$ is not in $\cP$. By contradiction, let us assume that $D'' \in \cP$. As $\cP$ is index-invariant, $D''_\sigma \in \cP$. Note that  $d_{EM}(D'_\sigma,D''_\sigma) \leq \eps/2$ as $d_{EM}(D',D'')\leq \eps/2$. So, $D'_\sigma$ is $\eps/2$-close to property $\cP$. As $d_{EM}(D,D'_\sigma)\leq \eps/2$, by the triangle inequality, $D$ is $\eps$-close to $\cP$, a contradiction. This completes the proof of Corollary~\ref{coro:vc-learn2}.
\color{black}
\end{proof}

%% file: vc-lb.tex
\section{Tightness of the bounds for bounded VC-dimension properties}\label{sec:bounded-vc}

As mentioned in the introduction, our tester for testing a VC-dimension property takes $\exp(d)$ samples, and performs $\exp(\exp(d))$ queries for VC-dimension $d$. Now we show that there exists an index-invariant property of VC-dimension at most $d$ which requires such sample and query complexities, proving Theorem~\ref{theo:lb-vc_intro}.

\color{black}

\begin{theo}[{\bf Restatement of Theorem~\ref{theo:lb-vc_intro}}]\label{theo:lb-vc}
Let $d,n \in \N$. There exists an index-invariant property $\cP_{\vc}$ with VC-dimension at most $d$ such that any (non-adaptive) tester for $\cP_{\vc}$ requires $2^{\Omega(d)}$ samples and $2^{2^{d - \Oh(1)}}$ queries.
\end{theo}

Since the query complexity of non-adaptive testers can be at most quadratic as compared to adaptive ones (Theorem~\ref{theo:lb-main}), arguing only for non-adaptive testers is sufficient for our purpose. {We would like to point out that the property of having support size at most $2^d$ is a property with VC-dimension bounded by $d$, for which the authors of \cite{GoldreichR21a} proved a lower bound of $\Omega(2^{(1- o(1))d})$ samples~\cite[Observation 2.7]{GoldreichR21a}. Although the sample lower bound of the property $\cP_{\vc}$ of Theorem~\ref{theo:lb-vc} is weaker in comparison to that of the  support size property, here we prove both sample and query lower bounds for the same property $\cP_{\vc}$. Moreover, $\cP_{\vc}$ is defined by being a permutation of a single distribution.}

Without loss of generality, in what follows, we assume that $d$ is large enough.

\paragraph*{Property $\cP_{\vc}$:}
Let $k=2^d$ and $\ell=2^{2^{d - 10}}$ be two integers and assume that $\ell$ divides $n$. Consider a matrix $A$ of dimension $k \times \ell$ such that the Hamming distance between any pair of column vectors of $A$ is at least $1/3$~\footnote{
One way to construct such a matrix is to select  $2^{d - 10}$  vectors from $\{0,1\}^{2^d}$ uniformly at random, and let the columns of $A$ be the set of all their linear combinations over the field $\mathbb{Z}_2$.}. Let $D_A$ be a  distribution supported over the vectors $\mathbf{V}_1,\ldots,\mathbf{V}_{k}$ such that, for every $i \in [k]$, the following holds:

\begin{itemize}
    \item $\mathbf{V}_i$ is the $n/\ell$ times ``blow-up'' of the $i$-th row of $A$, that is, for $j \in [\ell]$ and $j'$ with $(j-1)\cdot \frac{n}{\ell} < j' \leq j \cdot \frac{n}{\ell} $, $\left(\mathbf{V}_i\right)_{j'}=a_{ij}$, where $a_{ij}$ denotes the element of the matrix $A$ present in the $i$-th row and the $j$-th column.
    
    \item $D_A(\mathbf{V}_i)=\frac{1}{k}=\frac{1}{2^d}$.
\end{itemize}

Now we are ready to define the property $\cP_{\vc}$.
$$\cP_{\vc} = \{D: D= D_A^\sigma \ \mbox{for some permutation}\ \sigma: [n] \rightarrow [n]\}.$$


Now we have the following observation.

\begin{obs}
The VC-dimension of $\cP_{\vc}$ is at most $d$.
\end{obs}

This follows from the fact that the support size of the distribution $D_A$ is $2^d$. We will prove first the query complexity lower bound, and then prove the (easier) sample complexity lower bound.


\paragraph*{Query complexity lower bound:}

Let us define the first pair of hard distributions over distributions over $\{0,1\}^n$, that is, $D_{yes}$ and $D_{no}$.

\paragraph*{Distribution $D_{yes}$:} We choose a permutation $\sigma:[n] \rightarrow [n]$ uniformly at random, and pick the distribution $D_A^{\sigma}$ over $\{0,1\}^n$.
\newline

The distribution $D_{no}$ is constructed from the matrix $A$ that is used to define $D_{yes}$ as follows:
\paragraph*{Distribution $D_{no}$:}
We first choose $\ell'=2^{2^{d - 20}}$ many column vectors uniformly at random from $A$ and let $B$ be the resulting matrix of dimension $k \times \ell'$. Let $D_{B}$ be the distribution supported over the vectors $\mathbf{W}_1,\ldots,\mathbf{W}_{k}$ such that, for every $i \in [k]$, the following holds:

\begin{itemize}
    \item $\mathbf{W}_i$ is the $n/\ell'$ times blow-up of the $i$-th row of $B$, that is, for $j \in [\ell']$ and $j'$ with $(j-1)\cdot \frac{n}{\ell'} < j' \leq j \cdot \frac{n}{\ell'} $, $\left(\mathbf{W}_i\right)_{j'}=b_{ij}$, where $b_{ij}$ denotes the element of matrix $B$ present in the $i$-th row and the $j$-th column.
    
    \item $D_{no}(\mathbf{W}_i)=\frac{1}{k}=\frac{1}{2^d}$.
    
\end{itemize}

We choose a permutation $\sigma:[n] \rightarrow [n]$ uniformly at random, and pick the distribution $D_B^{\sigma}$ over $\{0,1\}^n$.

\begin{lem}\label{cl:emd-vc-low}
$D_{yes}$ is supported over $\cP_{\vc}$ and $D_{no}$ is supported over distributions that are $1/8$-far from  $\cP_{\vc}$.
\end{lem}

\begin{proof}
Following the definition of $\cP_{\vc}$ and $D_{yes}$, it is clear that $D_{yes}$ is supported over $\cP_{\vc}$. To prove the claim about $D_{no}$, consider the following definition and observation.

\begin{defi}
Let us consider a distribution $D$  over $\{0,1\}^n$. A matrix $M$ of dimension $s \times n$ is said to be a \emph{corresponding} matrix of $D$ if $D$ is the distribution resulting from picking uniformly at random a row of $M$.~\footnote{Note that, if $M$ has no duplicate rows, then $D$ is a uniform distribution over its support.} For a permutation $\pi:[s] \rightarrow [s]$, $M^\pi$ denotes the matrix obtained by permuting the rows of $M$ according to the permutation $\pi$, that is, the $\pi(i)$-th row of $M^\pi$ is same as the $i$-th row of $M$ for every $i \in [s]$.  
\end{defi}


Note that if $M$ is a corresponding matrix of $D$ with $s$ rows and $s'$ is a multiple of $s$, then the matrix $M'$ constructed by repeating every row of $M$ $s'/s$ many times is also a corresponding matrix of $D$.

Now the following observation connects the Earth Mover Distance between two distributions with the Hamming distance between their corresponding matrices.

\begin{cl}\label{obs:emd-far}
Let $D_1$ and $D_2$ be two distributions over $\{0,1\}^n$. Also, let $L$ and $M$ be corresponding matrices of $D_1$ and $D_2$, respectively, both of dimension $s \times n$. Then the Earth Mover Distance between $D_1$ and $D_2$ is the same as the minimum Hamming distance between $L$ and $M$ over all row permutations.

Formally, let the Hamming distance between $L$ and $M$ be defined as  $$d_H(L,M)=\frac{\size{\{(i,j) \in [s]\times [n]~:~l_{ij}\neq m_{ij}\}}}{s \cdot n}$$ Then
$$d_{EM}(D_1,D_2)=\min\limits_{\pi: [s] \rightarrow [s]} d_H({L}^\pi,M).$$
\end{cl}

\begin{proof}
We first note that any solution $f_{\mathbf{XY}}$ for the EMD between $D_1$ and $D_2$ can be translated to a doubly stochastic matrix $S$ of dimension $s \times s$ as follows: 

For every $i$, let $\mathbf{L}_i$ be the $i$-th row of $L$ and $l_i$ be the number of rows of $L$ that are identical to $\mathbf{L}_i$. Similarly, let $\mathbf{M}_i$ be the $i$-th row of $M$ and $m_i$ be the number of rows of $M$ that are identical $\mathbf{M}_i$. To construct the matrix $S$, we set the value of its entry at $i$-th row and $j$-th column as follows: $$s_{ij}=\frac{f_{\mathbf{L}_i \mathbf{M}_j} \cdot s}{l_i \cdot m_j}$$


Now we claim that the matrix $S$ defined above is a doubly stochastic matrix.

\begin{obs}
    The matrix $S$ defined above is doubly stochastic.
\end{obs}

\begin{proof}
We will prove that the every row of $S$ sum to $1$, and omit the identical proof for the columns of $S$. Note that if we sum the $i$-th row of $S$, we obtain the following:
$$\sum_{j=1}^ss_{ij}=\sum_{j=1}^s\frac{f_{\mathbf L_i\mathbf M_j}\cdot s}{l_i \cdot m_j}=\sum_{\mathbf Y \in \mathrm{Supp}(D_2)}\frac{f_{\mathbf L_i\mathbf Y}\cdot s}{l_i}=\frac{D_1(\mathbf L_i)\cdot s}{l_i}=1$$

This completes the proof of the observation.
\end{proof}

\color{black}

Now we will apply the Birkhoff-Newmann theorem~\cite{birkhoff1946three,von1953certain}, which states that the doubly stochastic matrix $S$ defined above can be expressed as a weighted average of permutation matrices. By translating the EMD expression from $f_{\mathbf{XY}}$ to $S$ and using an averaging argument, we can infer that there exists a permutation $\pi$ (among those in the representation of $S$) such that $d_H(L^{\pi},M)$ is equal to $d_{EM}(D_1,D_2)$. This completes the proof of the claim.
\end{proof}

Note that $D_{no}$ is supported over the set of distributions $ D_B^{\sigma}$ for any permutation $\sigma$ and any matrix $B$ which consists of $2^{2^{d - 20}}$ columns of $A$. We will be done by showing that the Earth Mover Distance between $D$ and $D_B^\sigma$ is at least $1/8$, where $D \in \cP_{\vc}$, $\sigma:[n] \rightarrow [n]$ is any permutation, and $B$ is any matrix with $2^{2^{d - 20}}$ columns.


Note that both $D$ and $D_B^\sigma$ admit respective corresponding matrices $L$ and $M$, respectively, both of dimension $2^d \times n$, where the rows of L are the vectors $\mathbf{V}_i$, and the rows of M are the respective permutations of the vectors $\mathbf{W}_i$. By Claim~\ref{obs:emd-far}, we note that: $$d_{EM}(D_B^\sigma,D)=\min_{\pi : [2^d]\rightarrow [2^d]} d_H(L^\pi,M).$$
The following claim will imply that $d_{EM}(D_B^\sigma,D) \geq 1/8$.

\begin{cl}\label{cl:emd}
For any permutation $\pi:[2^d]\rightarrow [2^d]$, $d_H(L^\pi,M)$ is at least $1/8$.
\end{cl}

\begin{proof}
Let us partition the index set $[n]$ into $\ell'$ many equivalence classes $C_1, \ldots, C_{\ell'}$ such that two indices of $[n]$ belong to the same equivalence class if the corresponding column vectors  in $L^\pi$ are identical. Observe that
$$d_H(L^\pi,M)=\frac{\sum\limits_{i \in [\ell']} \sum\limits_{j \in C_i}d_H(\mathbf{L}^\pi_j,\mathbf{M}_j)\cdot k}{k \cdot n}=\frac{\sum\limits_{i \in [\ell']} \sum\limits_{j \in C_i}d_H(\mathbf{L}^\pi_j,\mathbf{M}_j)}{n},$$
where $\mathbf{L}^\pi_j$ and $\mathbf{M}_j$ denote the $j$-th column vectors of $L^\pi$ and $M$, respectively.
 
Hence we will be done by showing $\sum\limits_{j \in C_i} d_H\left(\mathbf{L}^\pi_j,\mathbf{M}_j\right) \geq \frac{n}{8 \ell '},$ for every $i \in [\ell']$.
 
Note that $\size{C_i}=\frac{n}{\ell'}$. Also, all the columns in $\{\mathbf{L}_j^\pi : j \in C_i\}$  are identical. Consider a column vector $\mathbf{v} \in \{0,1\}^k$. Observe that there can be at most $\frac{n}{\ell}$ many columns in $\{\mathbf{M}_j : j \in C_i\}$ that are $1/7$-close to $\mathbf{v}$. This follows from the construction of $\cP_{\vc}$, which implies that for every column $\mathbf{M}_j$ of $M$, there are no more than $n/\ell - 1$ many other columns of $L^{\pi}$ whose distance from $\mathbf{M}_j$ is at most $2/7 < 1/3$.

So, in the expression $\sum\limits_{j \in C_i} d_H\left(\mathbf{L}^\pi_j,\mathbf{M}_j\right)$, there are at least $(\frac{n}{\ell'}-\frac{n}{\ell})$ many terms that are at least $1/7$. Hence, $\sum\limits_{j \in C_i} d_H\left(\mathbf{L}^\pi_j,\mathbf{M}_j\right)\geq \ell' \cdot \frac{1}{7}\left(\frac{n}{\ell'}-\frac{n}{\ell}\right) \geq \frac{n}{8\ell'}$.
\end{proof}
The above two claims conclude the proof of Lemma~\ref{cl:emd-vc-low}.
\end{proof}

\begin{lem}[{\bf Query complexity lower bound part of Theorem~\ref{theo:lb-vc}}]\label{cl:main-vc-low}
Any (non-adaptive) tester, that has sample and query access to either $D_{yes}$ or $D_{no}$ and performs $2^{2^{d - \omega(1)}}$ queries, can not distinguish between $D_{yes}$ and $D_{no}$.
\end{lem}

\begin{proof}
Let $A'$ and $B'$ be the matrices of dimension $k \times n$ such that the $i$-th row of $A'$ corresponds to the vector $\mathbf{V}_i^{\sigma}$ (for the permutation $\sigma$ drawn according to $D_{yes}$) and the $i$-th row of $B'$ corresponds to the vector $\mathbf{W}_i^{\sigma}$ (for the permutation $\sigma$ drawn according to $D_{no}$), where $i \in [k]$.

Let us divide the index set $[n]$ into $\ell$ equivalence classes $C_1,\ldots,C_{\ell}$ such that two indices belong to the same equivalence class if the corresponding column vectors in $A'$ are identical. Similarly, let us divide the index set $[n]$ into $\ell'$ equivalence classes $C'_1,\ldots,C'_{\ell'}$ such that two indices belong to the same equivalence class if the corresponding column vectors  in $B'$ are identical.

Let $Q \subseteq [n]$ be the set of all distinct indices queried by the tester to any sample (that is, the union of the sets $J_1, \ldots, J_s$ as they appear in Definition~\ref{defi:non-adaptive-tester}). If $\size{Q}=2^{2^{d - \omega(1)}}$, then the probability that there exist two indices in $Q$ that belong to the same $C_i$ or the same $C'_i$ is $o(1)$. Observe that, conditioned on the event that $Q$ does not contain two indices from the same equivalence class $C_i$ or $C'_i$, the distributions over the responses to the queries of the tester are identical for both $D_{yes}$ and $D_{no}$. The reason is that in both the cases of $D_{yes}$ and $D_{no}$, the distribution over the responses is identical to the one derived from picking a uniformly random subset of size $\size{Q}$ of the columns of the matrix $A$, and taking uniformly independent samples of the rows of the resulting matrix.
\end{proof}



Now we will prove the sample complexity lower bound for testing $\cP_{\vc}$.

\paragraph*{Sample complexity lower bound:}

Let us define the second pair of hard distributions over distributions over $\{0,1\}^n$, $D'_{yes}$ and $D'_{no}$.

\paragraph*{Distribution $D'_{yes}$:} Identically to $D_{yes}$ above, we choose a permutation $\sigma:[n] \rightarrow [n]$ uniformly at random, and pick the distribution $D_A^{\sigma}$ over $\{0,1\}^n$.
\newline

\vspace{2pt}

The distribution $D'_{no}$ is constructed from the matrix $A$ used to define $D'_{yes}$ as follows:


\paragraph*{Distribution $D'_{no}$:}
We first choose $k'=2^{d - 20}$ many row vectors uniformly at random from $A$ and construct a matrix $B'$ of dimension $k' \times \ell$. Let $D_{B'}$ be the distribution supported over the vectors $\mathbf{W}'_1,\ldots,\mathbf{W}'_{k'}$ such that, for every $i \in [k']$, the following hold:


\begin{itemize}
    \item $\mathbf{W}'_i$ is the $n/\ell$ times blow-up of the $i$-th row of $B'$, that is, for $j \in [\ell]$ and $j'$ with $(j-1)\cdot \frac{n}{\ell} < j' \leq j \cdot \frac{n}{\ell} $, $\left(\mathbf{W}'_i\right)_{j'}=b_{ij}$, where $b_{ij}$ denotes the element of matrix $B'$ present in the $i$-th row and the $j$-th column.
    
    \item $D_{no}(\mathbf{W}'_i)=\frac{1}{k'}=\frac{1}{2^{d-20}}$.
    

\end{itemize}

We choose a permutation $\sigma:[n] \rightarrow [n]$ uniformly at random, and pick the distribution $D_{B'}^{\sigma}$ over $\{0,1\}^n$.

\begin{lem}\label{cl:emd-vc-low_sample}
$D'_{yes}$ is supported over $\cP_{\vc}$ and $D'_{no}$ is supported over distributions that are $1/8$-far from  $\cP_{\vc}$.
\end{lem}

\begin{proof}
Following the definition of $\cP_{\vc}$ and $D'_{yes}$, it is clear that $D'_{yes}$ is supported over $\cP_{\vc}$. To prove the claim about $D'_{no}$, we will apply Claim~\ref{obs:emd-far}.

Note that $D'_{no}$ is supported over the set of distributions $ D_{B'}^{\sigma}$ for any permutation $\sigma$ and any matrix $B'$ which consists of $2^{d - 20}$ rows of $A$. We will be done by showing the Earth Mover Distance between $D$ and $D_{B'}^{\sigma}$ is at least $1/8$, where $D \in \cP_{\vc}$ and $\sigma:[n] \rightarrow [n]$ be any permutation, and $B'$ is any matrix with $2^{d - 20}$ distinct rows.

Let $L$ and $M$ be corresponding matrices of $D$ and $D_{B'}^{\sigma}$, respectively, of dimension $k \times n$, where $k=2^d$ (where the rows of $L$ are the vectors $\mathbf{V}_i$, and the rows of $M$ are $2^{20}$-fold repetitions of the respective permutations of the vectors $\mathbf{W'}_i$). By Claim~\ref{obs:emd-far}, we know that $$d_{EM}\left(D_{B'}^{\sigma}, D \right)=\min\limits_{\pi: [k] \rightarrow [k]} d_H(L^\pi,M).$$ Thus, the following claim will imply that $d_{EM}(D_{B'}^{\sigma}, D) \geq 1/8$.

\begin{cl}
For any permutation $\pi:[2^d]\rightarrow [2^d]$, $d_H(L^\pi,M)$ is at least $1/8$.    
\end{cl}

\begin{proof}
Our proof will follow a similar vain to that of Claim~\ref{cl:emd}.
Let us first partition the index set $[n]$ into $\ell'$ many equivalence classes $C_1, \ldots, C_{\ell'}$ such that two indices of $[n]$ belong to the same equivalence class if the corresponding column vectors  in $L^\pi$ are identical. Observe that
$$d_H(L^\pi,M)=\frac{\sum\limits_{i \in [\ell']} \sum\limits_{j \in C_i}d_H(\mathbf{L}^\pi_j,\mathbf{M}_j)\cdot k}{k \cdot n}=\frac{\sum\limits_{i \in [\ell']} \sum\limits_{j \in C_i}d_H(\mathbf{L}^\pi_j,\mathbf{M}_j)}{n},$$
where $\mathbf{L}^\pi_j$ and $\mathbf{M}_j$ denote the $j$-th column vectors of $L^\pi$ and $M$, respectively.

Since $B'$ has only $2^{d - 20}$ distinct rows, the number of its equivalence classes is bounded by $\ell'=2^{2^{d - 20}}$. Note that unlike the proof of the query lower bound, the sizes of the equivalence classes here may be different from each other. Also, note that the sizes of the equivalence classes of $L$ are $n/\ell$, as $D \in \cP_{\vc}$. Thus we have the following:
$$d_H(L^{\pi}, M) \geq \frac{1}{7} \cdot \frac{\sum_{i=1}^{\ell'}\max\{0, |C_i| - n/\ell\}}{n} \geq  \frac{1}{7} \cdot \left(1- \frac{1}{2^{10}}\right) \cdot n >\frac{1}{8}n.$$
The inequality follows from the facts that $\ell = 2^{2^{d-10}}$ and $\ell' = 2^{2^{d-20}}$, and the columns of $M$ corresponding to each $C_i$ can be $1/7$-close to at most $n/\ell$ columns of $L$.
\end{proof}
This concludes the proof of Lemma~\ref{cl:emd-vc-low_sample}.
\end{proof}

\noindent
The sample lower bound for testing $\cP_{\vc}$ now follows from the following lemma.

\begin{lem}[{\bf Sample complexity lower bound part of Theorem~\ref{theo:lb-vc}}]\label{lem:sample}
Any tester that takes at most $2^{o(d)}$ samples from the input distribution can not distinguish between the distributions $D'_{yes}$ and $D'_{no}$.    
\end{lem}

\begin{proof}
Let $\cS$ be the set of samples taken by the algorithm. Note that if $\size{\cS}=2^{o(d)}$, then the probability that $\cS$ contains two samples of the same $\mathbf{V}_i$ or the same $\mathbf{W'}_i$ is $o(1)$. Conditioned on the event that $\cS$ does not contain two samples from the same vector ($\mathbf{V}_i$ or $\mathbf{W'}_i$), even if the tester queries the samples of $\cS$ in their entirety, the distributions over the responses to the queries of the tester are identical for both $D'_{yes}$ and $D'_{no}$. This follows from the fact that the distribution over the responses is identical to a distribution obtained by drawing uniformly without repetitions a sequence of row vectors from $\mathbf{V}_1, \ldots, \mathbf{V}_{2^d}$, and querying the row vectors completely. This completes the proof.
\end{proof}

\color{black}

%% file: Lbound.tex
\color{black}

%% file: non-index-invariant.tex
\section{Exponential gap between adaptive and non-adaptive testers for general properties}\label{sec:exponentialgap}

In this section, we prove that unlike the index-invariant properties, there can be an exponential gap between the query complexities of adaptive and non-adaptive algorithms for non-index-invariant properties. In Subsection~\ref{sec:ubgen} we prove an exponential upper bound on the relation between the non-adaptive and adaptive query complexities of general properties. In Subsection~\ref{sec:lbgen}, we provide an exponential separation between them, and also use the same method to prove Proposition~\ref{theo:non-index}. 

\subsection{Relation between adaptive and non-adaptive testers for general properties}~\label{sec:ubgen}
Let us assume that $\cA$ is the adaptive algorithm that $\eps$-tests $\cP$ using $s$ samples $\{\mathbf{V}_1, \ldots, \mathbf{V}_s\}$ and $q$ queries, along with tossing some random coins. Before directly proceeding to the description of the non-adaptive algorithm, let us first consider the following observation.




\begin{obs}\label{obs:nonadaptquerygen}
For any given outcome sequence of the random coin tosses of $\cA$, there are at most $2^q-1$ possible internal states of $\cA$.

\end{obs}

\begin{proof}
Consider the $k$-th step of $\cA$, where $\cA$ queries the $j_k$-th index of $\mathbf{V}_{i_k}$ for some $i_k \in [s]$, $j_k \in [n]$, and $k \in [q]$. Note that ${i_1}$ and ${j_1}$ are functions of only the random coins, and ${i_k}$ and ${j_k}$ are functions of the random coins, as well as $\mathbf{V}_{i_1}\mid_{j_1}, \ldots, \mathbf{V}_{i_{k-1}}\mid_{j_{k-1}}$, where $2 \leq k \leq q$. Due to the $2^{k-1}$ possible values of $\mathbf{V}_{i_1}\mid_{j_1}, \ldots, \mathbf{V}_{i_{k-1}}\mid_{j_{{k-1}}}$, there are $2^k$ possible states of the algorithm $\cA$ at Step $k$, for each $1 \leq k \leq q$. Finally, the state of $\cA$ depending on the random coins and the values of $\mathbf{V}_{i_1}\mid_{j_1}, \ldots, \mathbf{V}_{i_q}\mid_{j_{q}}$ will decide the final output. This implies that for any fixed set of outcomes of the random coin tosses used by $\cA$, there can be a total of at most $\sum\limits_{i=0}^{q-1} 2^i=2^q-1$ internal states, each making one query, as well as $2^q$ final (non-query-making) states.
\end{proof}

Now we proceed to present the non-adaptive algorithm $\cA'$ that simulates $\cA$ by using $s$ samples and at most $2^q$ queries.

\begin{lem}\label{pre:adaptiveexp}
Let $\cP$ be any property that is $\eps$-testable by an adaptive algorithm using $s$ samples and $q$ queries. Then $\cP$ can be $\eps$-tested by a non-adaptive algorithm using $s$ samples and at most $2^q - 1$ queries, where $s$ and $q$ are integers.
\end{lem}


\begin{proof}
Let $\cA$ be the adaptive algorithm that $\eps$-tests $\cP$ using $s$ samples $\{\mathbf{V}_1, \ldots, \mathbf{V}_s\}$ and $q$ queries. Now we show that a non-adaptive algorithm $\cA'$ exists that uses $s$ samples and makes at most $2^q-1$ queries, such that the output distributions of $\cA$ and $\cA'$ are identical for any unknown distribution $D$.

The idea of $\cA'$ in a high level is to enumerate all possible internal steps of $\cA$, and list all possible queries $\cQ$ that might be performed by $\cA$. Note that $\cQ$ depends on the random coins used by $\cA$. We then query all the indices of $\cQ$ non-adaptively, and finally simulate $\cA$ using the full information at hand, with the same random coins that were used to generate $\cQ$. As $\cA$ has query complexity $q$, the number of possible internal states of $\cA$ is at most $2^q-1$, and the query complexity of $\cA'$ follows. Now we formalize the above intuition below.


The algorithm $\cA'$ has two phases:
\paragraph{Phase 1:}
\begin{enumerate}
    \item[(i)] $\cA'$ first takes $s$ samples $\mathbf{V}_1, \ldots, \mathbf{V}_s$.
    
    \item[(ii)] $\cA'$ now tosses some random coins (same as $\cA$) and determines the set of all possible indices $J_i$ of $\mathbf{V}_i$ that might be queried by $\cA$, for every $i \in [s]$. The sets of indices $J_i$'s are well defined after we fix the random coins, and follows from Observation~\ref{obs:nonadaptquerygen}.
\end{enumerate}

Thus at the end of Phase $1$, $\cA'$ has determined $s$ sets of indices $J_1, \ldots, J_s$ of the vectors $\mathbf{V}_1, \ldots, \mathbf{V}_s$ such that $\sum\limits_{i=1}^s|J_i|\leq 2^{q} -1$. Now $\cA'$ proceeds to the second phase of the algorithm.

\paragraph{Phase 2:}
\begin{enumerate}
    \item[(i)] For every $i \in [s]$ and $j \in J_i$, query the $j$-th index of $\mathbf{V}_i$, where $J_i$ denotes the set of indices of $\mathbf{V}_i$ that might be queried at the internal states of $\cA$, determined in Phase $1$.
    
    \item[(ii)] Simulate the algorithm $\cA$ using the same random coins used in Phase $1$, and report {\sc Accept} or {\sc Reject} according to the output of $\cA$.
    
\end{enumerate}


Note that the set of random coins that are used to determine $J_1, \ldots, J_s$ in Step $(ii)$ of Phase $1$ of the algorithm are the same random coins that are used to simulate $\cA$ in Step $(ii)$ of Phase $2$. Thus the correctness of $\cA'$ follows from to the correctness of $\cA$ along with Observation~\ref{obs:nonadaptquerygen}.
\end{proof}

\remove{
Given $\cA$, the algorithm $\cA'$ follows $\cA$ with the following difference:

Instead of choosing $i_k$ and $j_k$, it goes over all $2^{k-1}$ possibilities
for the answers $v_{i_1}|{j_1},...,v_{i_{k-1}}|{j_{k-1}}$, and queries all of them. After querying them, it uses $v_{i_k}|j_k$.

It is clear that the new algorithm has $s$ samples and $2^q-1$ queries in
total. Also, it is non-adaptive, because all $q$ query rounds use only the
internal coins, and hence they can be first calculated, and then you can
set each $I_i$ be the set of all queries made to $v_i$ from all $q$ "rounds".
}

\subsection{Exponential separation between adaptive and non-adaptive query complexities}\label{sec:lbgen}

Now we prove that the gap of Lemma~\ref{pre:adaptiveexp} is almost tight, in the sense that there exists a property such that the adaptive and non-adaptive query complexities for testing it are exponentially separated. 


Before proceeding to the proof, let us consider any property $\cP$ of strings of length $n$ over the alphabet $\{0,1\}$. Now we describe a related property $1_{\cP}$ over distributions as follows:

\paragraph*{{\bf Property $1_{\cP}$:}}

For any distribution $D \in 1_{\cP}$, the size of the support of $D$ is $1$, and the single string in the support of $D$ satisfies $\cP$.\\

Let us first recall the following result from \cite{GoldreichR21a}, which states that $\widetilde{\Oh}(\frac{1}{\eps})$ queries are enough to $\eps$-test whether any distribution has support size $1$.

\begin{lem}[{\bf Restatement of Corollary $2.3.1$ of \cite{GoldreichR21a}}]\label{lem:supportsizebound}
There exists a non-adaptive algorithm that $\eps$-tests whether an unknown distribution $D$ has support size $1$ and uses $\widetilde{\Oh}(\frac{1}{\eps})$ queries, for any $\eps \in (0,1)$.

\end{lem}

We now prove that the query complexity of $\eps$-testing $1_\cP$ is at least the query complexity of $\eps$-testing $\cP$, and can be at most the query complexity of $\frac{\eps}{2}$-testing of $\cP$, along with an additional additive factor of $\widetilde{\Oh}(\frac{1}{\eps})$ for testing whether the distribution has support size $1$. The result is formally stated as follows:

\begin{lem}\label{lem:stringprel}
Let $q_N$ and $q_A$ denote  the non-adaptive and adaptive query complexities for $\eps$-testing $\cP$, respectively. Similarly, let $Q_N$ and $Q_A$ denote the non-adaptive and adaptive query complexities of $\eps$-testing $1_{\cP}$, respectively. Then the following hold:
\begin{enumerate}
    \item \label{eqn:adaptivestring} 
        $q_A(\eps) \leq Q_A(\eps) \leq \widetilde{\Oh}(\frac{1}{\eps})+ \Oh\left(q_A(\frac{\eps}{2})\right)$~\footnote{We are using $\Oh(\cdot)$ as we are amplifying the success probability of the tester for the property $\cP$ to $9/10$ as compared to the usual success probability of $2/3$.}.

    \item \label{eqn:nonadaptivestring}$q_N(\eps) \leq Q_N(\eps) \leq \widetilde{\Oh}(\frac{1}{\eps})+ \Oh\left(q_N( \frac{\eps}{2})\right)$.
\end{enumerate}
\end{lem}

\begin{proof}
We prove here \eqref{eqn:adaptivestring}, and omit the nearly identical proof of \eqref{eqn:nonadaptivestring}.

\paragraph*{Proof of $q_A(\eps) \leq Q_A(\eps)$:} Consider an adaptive algorithm $\cA$ that $\eps$-tests $1_\cP$ by using $Q_A(\eps)$ queries. We construct an algorithm $\cA'$ that $\eps$-tests $\cP$ using the same number of queries. Let $\mathbf{V}$ be the unknown string of length $n$, where we want to test whether $\mathbf{V} \in \cP$ or $\mathbf{V}$ is $\eps$-far from $\cP$.

Let us define an unknown distribution $D'$ (over the Hamming cube $\{0,1\}^n$) such that we want to distinguish whether $D'\in 1_{\cP}$ or $D'$ is $\eps$-far from $1_{\cP}$. The distribution $D'$ is defined as follows:
\[ D'(\mathbf{X}) = \left\{
\begin{array}{ll}
     
      1 & \mathbf{X}=\mathbf{V} \\
      \vspace{3pt}
            0 & \mbox{otherwise}  
\end{array} 
\right. \]

Observe that $\mathbf{V} \in \cP$ if and only if $D' \in 1_\cP$. Similarly, it is not hard to see that $\mathbf{V}$ is $\eps$-far from $\cP$ if and only if $D'$ is $\eps$-far from $1_\cP$.
We simulate the algorithm $\cA$ by $\cA'$ as follows: when $\cA$ takes a sample, $\cA'$ does nothing, and when $\cA$ queries an index $i \in [n]$ of any sample, $\cA'$ queries the index $i$ of $\mathbf{V}$. Finally, $\cA'$ provides the output received from the simulation of $\cA$. 


From the description, it is clear that $\cA'$ performs exactly $Q_{\cA}(\eps)$ queries and is indeed simulated by running $\cA$ over $D'$.



\paragraph*{Proof of {\bf $Q_A(\eps) \leq \tOh\left(\frac{1}{\eps}\right) + \Oh\left(q_A(\frac{\eps}{2})\right)$:}}

Let us consider an adaptive algorithm $\cA_1$ that $\frac{\eps}{2}$-tests $\cP$ using $\Oh\left(q_A(\frac{\eps}{2})\right)$ queries to the unknown string $\mathbf{X} \in \{0,1\}^n$, with success probability at least $\frac{9}{10}$.  Now we design an adaptive algorithm $\cA_1'$ that $\eps$-tests $1_\cP$ using $\tOh\left(\frac{1}{\eps}\right) + \Oh\left(q_A(\frac{\eps}{2})\right)$ many queries.

\paragraph*{Algorithm $\cA_1'$:}
Assume that $D$ is the distribution that we want to $\eps$-test for $1_\cP$. The algorithm $\cA_1'$ performs the following steps:

\begin{enumerate}
    \item[(i)] Run the tester corresponding to Lemma~\ref{lem:supportsizebound} to $\frac{\eps}{20}$-test whether $D$ has support size $1$, with success probability at least $\frac{9}{10}$. If the tester decides that $D$ has support size $1$, then go to the next step. Otherwise, {\sc Reject}.
    \item[(ii)] Take one more sample from $D$ and let it be $\mathbf{U} \in \{0,1\}^n $. Run algorithm $\cA_1$ to $\frac{\eps}{2}$-test $\cP$ considering $\mathbf{X}=\mathbf{U}$ as the unknown string. If $\cA_1$ accepts, {\sc Accept}. Otherwise {\sc Reject}.
\end{enumerate}

Note that the query complexity for performing Step $(i)$ is $\widetilde{\Oh}(\frac{1}{\eps})$, which follows from Lemma~\ref{lem:supportsizebound}. Additionally, the number of queries performed in Step $(ii)$ is $\Oh\left(q_A(\frac{\eps}{2})\right)$, which follows from the assertion of the lemma. Thus, the algorithm $\cA_1'$ performs $\widetilde{\Oh}(\frac{1}{\eps}) + \Oh\left(q_A(\frac{\eps}{2})\right)$ many queries in total.

Now we will argue the correctness of $\cA_1'$. For completeness, assume that $D \in 1_\cP$. Let $\mathbf{V} \in \{0,1\}^n$ be the string such that $D(\mathbf{V})=1$ and $\mathbf{V} \in \cP$. Note that, by Lemma~\ref{lem:supportsizebound}, $\cA_1'$ proceeds to Step $(ii)$ with probability at least $\frac{9}{10}$. In Step $(ii)$, $\cA'$ sets $\mathbf{U}=\mathbf{V}$, and runs algorithm $\cA_1$ to $\frac{\eps}{2}$-test $\cP$ considering $\mathbf{X}=\mathbf{V}$ as the unknown string. Since $\mathbf{V} \in \cP$, by the assumption on the algorithm $\cA_1$, $\cA_1'$ accepts with probability at least $\frac{9}{10}$, given that $\cA_1'$ does not report {\sc Reject} in Step $(i)$. Thus, by the union bound, $\cA_1'$ accepts $D$ with probability at least $\frac{4}{5}$.

Now consider the case where $D$ is $\eps$-far from $1_{\cP}$. If $D$ is $\frac{\eps}{20}$-far from having support size $1$, $\cA_1'$ reports {\sc Reject} in Step $(i)$ with probability at least $\frac{9}{10}$, and we are done. So, assume that $D$ is $\frac{\eps}{20}$-close to having support size $1$. Then there exists a distribution $D'$ with support size $1$, and the distance between $D$ and $D'$ is at most $\frac{\eps}{20}$. Let us assume that $D'$ is supported on the string $\mathbf{V}$. By the Markov inequality, this implies that with probability at least $\frac{4}{5}$, a string $\mathbf{U}$ sampled according $D$ will be $\frac{9 \eps}{20}$-close to $\mathbf{V}$.

\begin{enumerate}
\item[(i)] If $\mathbf{V}$ is $\frac{19 \eps}{20}$-close to $\cP$, using the triangle inequality, this implies that $D$ is $\eps$-close to $1_{\cP}$, which is a contradiction.

\item[(ii)] Now consider the case where $\mathbf{V}$ is $\frac{19 \eps}{20}$-far from $\cP$. Recall that with probability at least $\frac{4}{5}$, the sample $\mathbf{U}$ taken at Step (ii) above is $\frac{9 \eps}{20}$-close to $\mathbf{V}$. As we are considering the case where $\mathbf{V}$ is $\frac{19 \eps}{20}$-far from $\cP$, using the triangle inequality, $\mathbf{U}$ is $\frac{\eps}{2}$-far from $\cP$ with the same probability. In this case, the algorithm will {\sc Reject} in Step $(ii)$, with probability at least $\frac{9}{10}$. Together, this implies that the  algorithm will {\sc Reject} the distribution $D$, with probability at least $\frac{7}{10}$.
\end{enumerate}
\end{proof}

In the following, we will construct the property $\cP_{Pal}$ of strings over the alphabet $\{0,1,2,3\}$. It will then be encoded as a property of strings over $\{0,1\}$ by using two bits per letter.

\paragraph*{{\bf Property $\cP_{Pal}$:}}

A string $\mathbf{S}$ of length $n$ is in $\cP_{Pal}$ if $\mathbf{S} = \mathbf{X}\mathbf{Y}$, where $\mathbf{X}$ is a palindrome over the alphabet $\{0,1\}$, and $\mathbf{Y}$ is a palindrome over the alphabet $\{2,3\}$.

There is an exponential gap between the query complexities of adaptive and non-adaptive algorithms to $\eps$-test $\cP_{Pal}$. The result is stated as follows:

\begin{lem}\label{lem:stringlb}
There exists an adaptive algorithm that $\eps$-tests $\cP_{Pal}$ by performing $\Oh(\log n)$ queries for any $\eps \in (0,1)$. However, there exists an $\eps \in (0,1)$ such that $\Omega(\sqrt{n})$ non-adaptive queries are necessary to $\eps$-test $\cP_{Pal}$.
\end{lem}

\begin{proof}
The lower bound proof (using Yao's lemma), which we omit here, is nearly identical to the one from \cite{AlonKNS99} (see Theorem $2$ therein).


Let us assume that $\mathbf{V}$ is the string that we want to $\eps$-test for $\cP_{Pal}$. The adaptive algorithm to $\eps$-test $\cP_{Pal}$ uses binary search, and is described below:

\begin{enumerate}
    \item[(i)] Use binary search for an index of $\mathbf{V}$ that has ``value 1.5'' (which is not present in the input). This returns an index $0 \leq i \leq n$, such that (a) $\mathbf{V}_i \in \{0,1\}$ unless $i=0$, and (b) $\mathbf{V}_{i+1} \in \{2,3\}$ unless $i=n$.
     
    
    \item[(ii)]Repeat $\Oh(\frac{1}{\eps})$ times:
    
    \begin{enumerate}
        \item[(a)] Sample an index $j \in [n]$ uniformly at random.
        
        \item[(b)] If $j \leq i$, then query $\mathbf{V}_j$ and $\mathbf{V}_{i+1 - j}$. {\sc Reject} if they are not both equal to the same value in $\{0,1\}$.
        
        \item[(c)] Otherwise query $\mathbf{V}_j$ and $\mathbf{V}_{n+i +1 - j}$. {\sc Reject} if they are not both equal to the same value in $\{2,3\}$.
        
    \end{enumerate}
    
    \item[(iii)] If the input has not been rejected till now, {\sc Accept}.
    
\end{enumerate}


We first argue the completeness of the algorithm. Assume that $\mathbf{V}$ is a string such that $\mathbf{V} \in \cP_{Pal}$, and $i$ is the index returned by Step $(i)$ of the algorithm. As $\mathbf{V}=\mathbf{XY}$ for some palindrome $\mathbf{X}$ over $\{0,1\}$ and palindrome $\mathbf{Y}$ over $\{2,3\}$, the index $i$ will be equal to $\size{\mathbf{X}}$. This implies that the algorithm will {\sc Accept} $\mathbf{V}$ with probability $1$.


Now consider the case where $\mathbf{V}$ is $\eps$-far from $\cP_{Pal}$. 
We call an index $j$ \emph{violating} if it does not satisfy the condition appearing either in Step (ii)(b) or Step (ii)(c) above, where $i$ is the index returned in Step $(i)$. The number of violating indices is at least $\eps n$, because otherwise we can change the violating indices such that the modified input is a string of the form $\mathbf{XY}$ following the definition of $\cP_{Pal}$, where $\size{\mathbf{X}}=i$. Since the loop in Step (ii) runs for $\Oh(\frac{1}{\eps})$ times, we conclude that with probability at least $\frac{2}{3}$ at least one such violating index will be found. So, the algorithm will {\sc Reject} $\mathbf{V}$ with probability at least $\frac{2}{3}$.
\end{proof}






Now we are ready to formally state and prove the main result of this section.

\begin{pre}\label{lem:nonindexlb}
There exists a property of distributions over strings that can be $\eps$-tested adaptively using $\Oh(\log n)$ queries for any $\eps \in (0,1)$, but $\Omega(\sqrt{n})$ queries are necessary for any non-adaptive algorithm to $\eps$-test it for some $\eps \in (0,1)$.

\end{pre}

\begin{proof}
Consider the property $1_{\cP_{Pal}}$. From Lemma~\ref{lem:stringlb}, we know that $q_A(\frac{\eps}{2}) = \Oh(\log n)$, for any fixed $\eps \in (0,1)$. Using the upper bound of Lemma~\ref{lem:stringprel}, we conclude that $Q_A(\eps) = \Oh(\log n)$, for any fixed $\eps \in (0,1)$, ignoring the additive $\widetilde{\Oh}(\frac{1}{\eps})$ term.

On the other hand, according to Lemma~\ref{lem:stringlb}, $q_N(\eps) = \Omega(\sqrt{n})$ for some $\eps \in (0,1)$. Thus, following Lemma~\ref{lem:stringprel}, we conclude that $Q_N(\eps) = \Omega(\sqrt{n})$ holds for some $\eps \in (0,1)$. Together, Proposition~\ref{lem:nonindexlb} follows.
\end{proof}



Now we present a sketch of a proof of Proposition~\ref{theo:non-index}, which shows that for a property to be constantly testable, it is not sufficient that the property has constant VC-dimension, unless it is index-invariant as well.

\begin{pre}[{\bf Restatement of Proposition~\ref{theo:non-index}}]\label{pre:linlb}
There exists a non-index-invariant property $\cP$ such that any distribution $D \in \cP$ has VC-dimension $O(1)$ and the following holds. There exists a fixed $\eps>0$, such that distinguishing whether $D \in \cP$ or $D$ is $\eps$-far from $\cP$, requires $\Omega(n)$ queries, where the distributions in the property $\cP$ are defined over the $n$-dimensional Hamming cube $\{0,1\}^n$.
\end{pre}

\begin{proof}
Note that the VC-dimension of $1_{\cP}$ is $0$, where $1_{\cP}$ is the property corresponding to $\cP$ as defined before. String properties which are hard to test, for which there is a fixed $\eps>0$ such that $\eps$-testing them requires $\Omega(n)$ queries, are known to exist. Examples are properties studied in the work of Ben{-}Eliezer, Fischer, Levi and Rothblum~\cite{DBLP:conf/innovations/Ben-EliezerFLR20}, and in the work of Ben{-}Sasson, Harsha and Raskhodnikova~\cite{DBLP:journals/siamcomp/Ben-SassonHR05}. Defining $1_{\cP}$ for such a property $\cP$ provides us the example proving Proposition~\ref{pre:linlb}. 
\end{proof}

%% file: adaptive_ub.tex

%% file: quadratic_lb_new.tex
\section{Quadratic gap between adaptive and non-adaptive testers for index-invariant properties}\label{sec:quadratic_gap1}


In this section, we first prove Theorem~\ref{theo:lb-main} in Subsection~\ref{sec:quadratic_gap2}, that is, there can be at most a quadratic gap between the query complexities of adaptive and non-adaptive algorithms for testing index-invariant properties. Then in Subsection~\ref{sec:quadratic_gap}, we prove Theorem~\ref{theo:lb-main_new_intro}, that is,  we demonstrate a quadratic separation between them, which is one of the main results of the paper and the main content of this section.

\subsection{Quadratic relation between adaptive and non-adaptive testers for index-invariant properties}\label{sec:quadratic_gap2}

\begin{theo}[{\bf Restatement of Theorem~\ref{theo:lb-main}}]\label{theo:ub-main}
Let $\cP$ be any index-invariant property that is $\eps$-testable by an adaptive algorithm using $s$ samples and $q$ queries. Then $\cP$ can be $\eps$-tested by a non-adaptive algorithm using $s$ samples and $sq \leq q^2$ queries, where $s$ and $q$ are integers.
\end{theo}

\begin{proof}

The main idea of the proof is to start with an adaptive algorithm $\cA$ as stated above, and then argue for another semi-adaptive algorithm $\cA'$ with sample complexity $s$ but query complexity $qs$, such that the output distributions of $\cA$ and $\cA'$ are the same for any unknown distribution $D$.
Finally, we construct a non-adaptive algorithm $\cA''$ such that (i) the sample and query complexities of $\cA''$ are the same as that of $\cA'$, and (ii) the probability bounds of accepting and rejecting distributions depending on their distances to $\cP$ are preserved from $\cA'$ to $\cA''$. Now we proceed to formalize this argument.

Let $\cA$ be the adaptive algorithm that $\eps$-tests $\cP$ using $s$ samples $\{\mathbf{V}_1, \ldots, \mathbf{V}_s\}$ and $q$ queries. Now we show that a two phase algorithm $\cA'$ exists that takes $s$ samples $\{\mathbf{V}_{1},\ldots,\mathbf{V}_{s}\}$ and proceeds as follows:

\paragraph{Phase 1:} In this phase, $\cA'$ queries in an adaptive fashion. If $\cA$ queries the $j_k$-th index of $\mathbf{V}_{i_k}$ at its $k$-th step, for some $i_k \in [s]$ and $j_k \in [n]$, then we perform the following steps:
\begin{enumerate}
    \item[(i)] If $\cA'$ has queried the $j_k$-th index of all the samples before this step, then we reuse the queried value.
    
    \item[(ii)] Otherwise, we query the $j_k$-th index from all the samples $\{\mathbf{V}_1, \ldots, \mathbf{V}_s\}$.
\end{enumerate}

\paragraph{Phase 2:} Let $\cQ \subset [n]$ be the set of indices queried by $\cA'$ while running the $q$ querying steps of $\cA$.  If $\size{\cQ} < q$, we arbitrarily pick $t = q- \size{\cQ}$ distinct indices $\{j'_1, \ldots, j'_t\}$, disjoint from the set of indices $\cQ$. We query the set of indices ${j'_1, \ldots, j'_t}$ from the entire set of sampled vectors $\mathbf{V}_1, \ldots,\mathbf{V}_s$. 

The output ({\sc Accept} or {\sc Reject}) of $\cA'$ is finally set to that of $\cA$, and in particular depends only on the answers to the queries made in the first phase.





Now we have the following observation regarding the query complexity of $\cA'$, which will be used to argue the query complexity of the non-adaptive algorithm later.
\begin{obs}\label{obs:semiadaptquery}
$\cA'$ uses $s$ samples and performs exactly $qs$ queries. Moreover, for any distribution $D$, the output distribution of $\cA'$ is the same as that of $\cA$. 
    
    
\end{obs}

Let us assume that $\cA'$ proceeds in $q$ steps by querying indices $\ell_1,\ldots,\ell_q \in [n]$ in each of the $s$ samples $\mathbf{V}_1, \ldots, \mathbf{V}_s$ (when the unknown distribution is $D$). Equivalently, we can think that the algorithm proceeds in $q$ steps, where in Step $k$ ($k \in [q]$), we query the $\ell_k$-th index of $\{\mathbf{V}_1, \ldots, \mathbf{V}_s \}$, such that $\ell_k$ depends on $\ell_1,\ldots,\ell_{k-1}$, where $2 \leq k \leq q$. 

Let us now consider an uniformly random permutation $\sigma: [n] \rightarrow [n]$ (unknown to $\cA'$). Assume that the unknown distribution is $D_\sigma$ instead of $D$. As $\cP$ is index-invariant, we can assume that the algorithm $\cA'$ runs on $D_\sigma$ for $q$ steps as follows. In Step $k$, $\cA'$ queries the $\sigma(\ell_k)$-th index of each of the $s$ samples, for $k \in [q]$. Now we have the following observation regarding the distribution of the indices queried, which follows from $\sigma$ being uniformly random.


\begin{obs}\label{obs:sigmauni}
$\sigma(\ell_1)$ is uniformly distributed over $[n]$, and $\sigma(\ell_k)$ is uniformly distributed over $[n] \setminus \{\sigma(\ell_1),\ldots,\sigma(\ell_{k-1})\}$, where $2 \leq k \leq q$. Moreover, this holds even if we condition on the values $\ell_1, \ldots, \ell_k$ as well as $\sigma(\ell_1), \ldots,\sigma(\ell_{k-1})$.
\end{obs}


Now the algorithm $\cA''$ works as follows:
\begin{itemize}
    \item First take a uniformly random permutation $\sigma: [n] \rightarrow [n]$.
    
    \item Run $\cA'$ over $D_{\sigma}$ instead of $D$.\remove{, that is, instead of querying the index $\ell_i$, query the index $\sigma(\ell_i)$ from all the samples $\mathbf{v}_1, \ldots, \mathbf{v}_s$}
    
\end{itemize}
From the above description, it does not immediately follow that $\cA''$ is a non-adaptive algorithm. But from the description along with Observation~\ref{obs:sigmauni}, it follows that $\cA''$ is the same as the following algorithm:  

\begin{itemize}
    \item First take $s$ samples $\mathbf{V}_1,\ldots,\mathbf{V}_s$, and  also pick a uniformly random non-repeating sequence of $q$ indices $r_1, \ldots,r_q \in [n]$.
    
    \item Run $\cA'$ such that, for every $i \in [q]$, when $\cA'$ is about to query $\ell_i$, query $r_i$ from all samples instead. That is, we assume $r_i$ to be the value of $\sigma(\ell_i)$.
\end{itemize}
 
\color{black}

The sample complexity and  query complexity  of algorithm $\cA''$ are $s$ and $qs$, respectively, which follows from Observation~\ref{obs:semiadaptquery} and Observation~\ref{obs:sigmauni}. The correctness of the algorithm follows from Observation~\ref{obs:semiadaptquery} and Observation~\ref{obs:sigmauni} along with the fact that $\cP$ is index-invariant. This completes the proof of Theorem~\ref{theo:ub-main}.
\end{proof}

\color{black}

\subsection{Preliminaries towards proving a quadratic separation result}\label{sec:quadratic_gap}

In this subsection, we present some preliminary results required to prove that Theorem~\ref{theo:lb-main} is almost tight, that is, there exists an index-invariant property for which there is a nearly quadratic gap between the query complexities of adaptive and non-adaptive testers. The result is formally stated as follows.

\begin{theo}[{\bf Restatement of Theorem~\ref{theo:lb-main_new_intro}}]\label{theo:lb-main_new}
There exists an index-invariant property $\cP_{\mathrm{Gap}}$ that can be $\eps$-tested adaptively using $\tOh(n)$ queries for any $\eps \in (0,1)$, while there exists an $\eps \in (0,1)$ for which $\widetilde{\Omega}(n^{2})$ queries are necessary for any non-adaptive $\eps$-tester.
\end{theo}



In what follows throughout this section, we assume that the integer $n$ is of the form $n=2^l$ for some integer $l$, and that $k=\Oh(l)$ is another integer.  We denote vectors in $\{0, 1\}^N$ by capital bold letters (for example $\mathbf{X} \in \{0, 1\}^N$) and vectors in $\{0, 1\}^n$ by small bold letters (for example $\mathbf{x} \in \{0, 1\}^n$).
For two vectors $\mathbf{X}, \mathbf{Y} \in \{0,1\}^N$, we will use $\delta_H(\mathbf{X}, \mathbf{Y}) = N \cdot d_H(\mathbf{X}, \mathbf{Y})$ to denote the absolute Hamming distance between $\mathbf{X}$ and $\mathbf{Y}$.

To construct the property $\cP_{\mathrm{Gap}}$ (as stated in Theorem~\ref{theo:lb-main_new}), we define two encodings $\sen:\{0,1\}^\ell \rightarrow \{0,1\}^k$ and $\gen:[n]^m \rightarrow [n]^n$~\footnote{{\sc SE} stands for Secret Encoding, and {\sc GE} stands for General Encoding.}. The encodings $\gen$ and  $\sen$ follow from the construction of a Probabilistically Checkable Unveiling of a Shared Secret (PCUSS) in \cite{DBLP:conf/innovations/Ben-EliezerFLR20}.  We can also construct such a function $\gen$ using the Reed-Solomon code, where we will assume that $n$ is a prime power and use polynomials of degree $m-1$ over the field $\mbox{{\sc GL}}(n)$ for $m= \Theta(n)$~\footnote{$\mbox{{\sc GL}}(n)$ stands for the finite field with $n$ elements.}.


\paragraph*{Function {\sc SE}:}

We will use a function $\sen$ of the form $\sen:\{0,1\}^l \times \{0,1\} \rightarrow \{0,1\}^k$, where $l$ and $k$ are the integers defined above. In fact, $\sen$ takes an integer $i \in [n]$ in its Boolean encoding as an $l$ bit Boolean string and a ``secret'' bit $a \in \{0,1\}$, and will output a Boolean string of length $k$. $\sen$ will have the following properties for some constant $\zeta \in (0,1/2)$. 


\begin{enumerate}
    \item[(i)] Let $i, i' \in [n]$ be two integers encoded as binary strings of length $l$~\footnote{Binary strings of length $\log n$ can actually encode only integers from $\{0, \ldots, n-1\}$, so we use the encoding of $0$ for the value $n$.}, and $a, a' \in \{0,1\}$. If $(i,a) \neq (i',a')$, then $\delta_H(\sen(i,a), \sen(i',a')) \geq \zeta \cdot k$.
    

     \item[(ii)] Let $a \in \{0,1\}$ be a fixed bit, and suppose that $i$ is an integer chosen uniformly at random from $[n]$. Then for any set of indices $I \subset [k]$ such that $\size{I} \leq \zeta \cdot k$, the restriction $\sen(i,a) \mid_I$ is uniformly distributed over $\{0,1\}^{\size{I}}$. 
     
    
\end{enumerate}

\paragraph*{Function \mbox{{\sc GE}:}}

For our construction, we will use another function $\gen$ of the form $\gen:[n]^m \rightarrow [n]^n$, where $n,m \in \mathbb{N}$ with the following properties for the same constant $\zeta \in (0,1/2)$ as above.

\begin{enumerate}

     \item[(i)] Let $\mathbf{z}, \mathbf{z'} \in [n]^m$ be two strings such that $\mathbf{z} \neq \mathbf{z'}$. For any two such strings $\mathbf{z}$ and $\mathbf{z'}$, $\delta_H(\gen(\mathbf{z}), \gen(\mathbf{z'})) = |\{i: \gen(\mathbf{z})_i \neq \gen(\mathbf{z'})_i\}| \geq \zeta \cdot n$.
     

     \item[(ii)] Consider a string $\mathbf{z} \in [n]^m$ chosen uniformly at random. For any set of indices $I \subset [n]$ such that $\size{I} \leq \zeta \cdot n$, $\gen(\mathbf{z}) \mid_I$ is uniformly distributed over $[n]^{\size{I}}$.
     
\end{enumerate}

\remove{\begin{rem}
We can construct such a function $\gen$ using the Reed-Solomon code, where we will use polynomials of degree $m-1$ over the field $GL(n)$ for $m= \Theta(n)$.~\footnote{\textcolor{blue}{$GL$ stands for General Linear group.}} A more computationally efficient way to construct $\gen$ (and the earlier function $\sen$) follows from the construction of a Probabilistically Checkable Unveiling of a Shared Secret (PCUSS) in \cite{DBLP:conf/innovations/Ben-EliezerFLR20}.  
\end{rem}}

From now on, we will use the following notation in this subsection: Let $n \in \N$ be such that $n=2^l$ for some integer $l$, $k=\Oh(l)$ and $\zeta \in (0,1/2)$ as above, $b = \lfloor \log(\lceil \log kn \rceil) \rfloor +1$, $N = 1+b+kn$ and 
$\alpha = 1/\log n$. Note that in particular $N=\Oh(n \log n)$. For a vector $\mathbf{X} \in \{0,1\}^N$ and a permutation $\pi:[N] \rightarrow [N]$, $\mathbf{X}_{\pi}$ denotes the vector obtained from $\mathbf{X}$ by permuting the indices of $\mathbf{X}$ with $\pi$, that is, $\mathbf{X}_{\pi}=(\mathbf{X}_{\pi(1)}, \ldots, \mathbf{X}_{\pi(N)})$. 


Let $B$ be the sequence of integers $B = \{2, \ldots, b+1\}$, and for every $j \in [n]$, let $C_j$ denote the sequence of integers
$C_j = \{b+2+k(j-1), \ldots, b+1+kj\}$. For a sequence of integers $A$ and a vector $\mathbf{X}$, we denote by $\mathbf{X} \mid_A$ the vector obtained by projecting $\mathbf{X}$ onto the set of indices of $A$ preserving the sequence order. For a sequence $A \subseteq [N]$ and a permutation $\pi:[N] \rightarrow [N]$, we denote by $\pi(A)$ the sequence obtained after permuting every element of $A$ with respect to the permutation $\pi$, that is, if $A=(a_1, \ldots, a_l)$, then $\pi(A)=(\pi(a_1), \ldots, \pi(a_l))$. In particular, we have $\mathbf{X}_{\pi} \mid_A = \mathbf{X} \mid_{\pi(A)}$. By abuse of notation and for simplicity, for a set of integers $A$ and a vector $\mathbf{X}$, we denote by $\mathbf{X} \mid_A$ the vector obtained by projecting $\mathbf{X}$ onto the set of indices of $A$, whenever the ordering in which we consider the indices in $A$ will be clear from the context~\footnote{A common scenario is when the indexes of $A$ are considered as a monotone increasing sequence.}.

In the following, we use string notation. For example, $\mathbf{1}^k\mathbf{0}^k$ denotes the vector in $\{0,1\}^{2k}$ whose first $k$ coordinates are $1$ and whose last $k$ coordinates are $0$. Now we formally define the notion of encoding of a vector which will be crucially used to define $\cP_{\mathrm{Gap}}$.



\begin{defi}[{\bf Encoding of a vector}]\label{defi:encoding}
Let $n,k,b \in \mathbb{N}, N=1+b+kn$, and $\mathbf{x} = (\mathbf{x}_1, \ldots, \mathbf{x}_n) \in \{0,1\}^{n}$ and $\mathbf{Y}\in \{0,1\}^N$ be two vectors. $\mathbf{Y}$ is said to be an \emph{encoding} of $\mathbf{x}$ with respect to the functions $\sen:\{0,1\}^l \times \{0,1\} \rightarrow \{0,1\}^k$ and $\gen:[n]^m \rightarrow [n]^n$ if the following hold:
\begin{enumerate}
    \item[(i)] The first index of $\mathbf{Y}$ is $0$.
    
    \item[(ii)] $\mathbf{Y} \mid_B$ is the all-$1$ vector.
    \item[(iii)] $\mathbf{Y} \mid_{[N] \setminus \{1\} \cup B}$ is of the form $\sen(\gen(\mathbf{z})_1,\mathbf{x}_1) \ldots \sen(\gen(\mathbf{z})_n,\mathbf{x}_n)$ for some string $\mathbf{z} \in [n]^m$. In other words, 
    $\mathbf{Y} \mid_{C_j} = \sen(\gen(\mathbf{z})_j,\mathbf{x}_j)$ for every $j \in [n]$.
    
    
    
    
\end{enumerate}

\end{defi}

\color{black}

For ease of presentation, we will denote this encoding by $\fen$, that is, $\fen: [n]^m \times \{0,1\}^n \rightarrow \{0,1\}^N$ is the function~\footnote{{\sc FE} stands for Final Encoding.} such that $\fen(\mathbf{z},\mathbf{x})= \mathbf{0}(\mathbf{1}^b)\sen(\gen(\mathbf{z})_1, \mathbf{x}_1) \ldots \sen(\gen(\mathbf{z})_n,\mathbf{x}_n)$, for $\mathbf{z} \in [n]^m$ and $\mathbf{x} = (\mathbf{x}_1, \ldots, \mathbf{x}_n) \in \{0,1\}^n$. We also say that $\mathbf{X} \in \{0,1\}^N$ is a \emph{valid encoding} of some $\mathbf{x} \in \{0,1\}^n$, if there exists some $\mathbf{z} \in [n]^m$ for which $\mathbf{X} = \fen(\mathbf{z},\mathbf{x})$. {The image of {{\sc FE}} will be called the set of all valid encodings.}

Now let us infer two properties of the function $\fen$, which will be crucial to our proofs, as stated in the following two claims. These properties of {\sc FE} are analogous to the properties of {\sc SE} and {\sc GE}. As {\sc FE} is formed by combining {\sc SE} and {\sc GE}, the proofs of these observations use their respective properties.


The following observation, particularly Items (i) and (ii), will allow us to prove that certain distributions are indeed far from the property $\cP_{\mathrm{Gap}}$ (to be defined later) in the EMD metric. Item (iii) will be useful to prove the soundness of our adaptive algorithm in Subsection~\ref{sec:quadratic_adaptive_ub}, and in particular in Lemma~\ref{cl:emd-far}.

\color{black}

\begin{obs}[{\bf Distance properties of $\fen$}]\label{obs:FEproperties}
Let $\fen: [n]^m \times \{0,1\}^n \rightarrow \{0,1\}^N$ be the function from Definition~\ref{defi:encoding}. Then $\fen$ has the following properties:

\begin{enumerate}

\item[(i)] Let $\mathbf{x}, \mathbf{x'} \in \{0,1\}^n$ be any two strings and $\mathbf{z}, \mathbf{z'} \in [n]^m$ be two vectors such that $\mathbf{z} \neq \mathbf{z'}$. Then $\delta_H(\fen(\mathbf{z},\mathbf{x}), \fen(\mathbf{z'},\mathbf{x'})) \geq \zeta^2 \cdot N/2$ holds.

\item[(ii)] Let $\mathbf{z},\mathbf{z'} \in [n]^m$ be any two strings, and $\mathbf{x}, \mathbf{x'} \in \{0,1\}^n$ be two other strings such that $\mathbf{x} \neq \mathbf{x'}$. Then $\delta_H(\fen(\mathbf{z},\mathbf{x}), \fen(\mathbf{z'},\mathbf{x'})) \geq \zeta k \cdot \delta_H(\mathbf{x}, \mathbf{x'})$. 

\item[(iii)] 
Let $\mathbf{x}, \mathbf{x'} \in \{0,1\}^n$ be two strings and $\mathbf{z} \in [n]^m$ be a vector. Then $\delta_H(\fen(\mathbf{z},\mathbf{x}), \fen(\mathbf{z},\mathbf{x'})) \leq k \cdot \delta_H(\mathbf{x},\mathbf{x'})$, and in particular $d_H(\fen(\mathbf{z},\mathbf{x}), \fen(\mathbf{z},\mathbf{x'})) \leq d_H(\mathbf{x},\mathbf{x'})$ holds.

\end{enumerate}
\end{obs}

\begin{proof}

We prove each item separately below.
    
\begin{enumerate}

\item[(i)]


Following the properties of $\gen$ (Property (i)),  for $\mathbf{z} \neq \mathbf{z'}$, we know that $\delta_H(\gen(\mathbf{z}), \gen(\mathbf{z'})) \geq \zeta \cdot n\geq \zeta N/2k$. That is, the number of indices $j \in [n]$ such that $\gen(\mathbf{z})_j \neq \gen(\mathbf{z'})_j$, is at least $\zeta N/2k$. For every index $j \in [n]$ such that $\gen(\mathbf{z})_j \neq \gen(\mathbf{z'})_j$, we can say that $\delta_H(\sen(\gen(\mathbf{z})_j, \mathbf{x}_j), \sen(\gen(\mathbf{z'})_j, \mathbf{x'}_{j})) \geq \zeta \cdot k$. This is due to  Property (i) of $\sen$. Hence,
\begin{eqnarray*}
\delta_H(\fen{\left(\mathbf{z},\mathbf{x}\right)},\fen{\left(\mathbf{z}',\mathbf{x'}\right)})
&\geq& \sum\limits_{j \in [n]: \mathbf{z}_j \neq \mathbf{z'}_j}\delta_H(\sen(\gen(\mathbf{z})_j, \mathbf{x}_j), \sen(\gen(\mathbf{z})_j, \mathbf{x'}_j))\\
&\geq& \zeta N/2k \cdot \zeta k=\zeta^2 \cdot N/2.
\end{eqnarray*}

\color{black}
        

\item[(ii)]

Consider two strings $\mathbf{x}, \mathbf{x'} \in \{0,1\}^n$ such that $\mathbf{x} \neq \mathbf{x'}$. Using Property (i) of $\sen$, we know that $\delta_H(\sen(\gen(\mathbf{z})_j, \mathbf{x}_j), \sen(\gen(\mathbf{z})_j, \mathbf{x'}_j)) \geq \zeta \cdot k$ for every $j$ for which $\mathbf{x}_j \neq \mathbf{x'}_j$. Note that the number of such indices $j$ is $\delta_H(\mathbf{x}, \mathbf{x'})$. Summing over them, we have the result.

\item[(iii)] 
Consider any two strings $\mathbf{x}, \mathbf{x'} \in \{0,1\}^n$. Observe that $$ \delta_H(\fen{\left(\mathbf{z},\mathbf{x}\right)},\fen{\left(\mathbf{z},\mathbf{x'}\right)})=\sum\limits_{j \in [n]}\delta_H(\sen(\gen(\mathbf{z})_j, \mathbf{x}_j), \sen(\gen(\mathbf{z})_j, \mathbf{x'}_j)).$$

Note that $\delta_H(\sen(\gen(\mathbf{z})_j, \mathbf{x}_j), \sen(\gen(\mathbf{z})_j, \mathbf{x'}_j))$ is at most $k$ for every $j \in [n]$. Moreover, $\delta_H(\sen(\gen(\mathbf{z})_j, \mathbf{x}_j), \sen(\gen(\mathbf{z})_j, \mathbf{x'}_j))= 0$ for every $j \in [n]$ with $\mathbf{x}_j = \mathbf{x}_j'$. Since the number of indices $j$ such that $\mathbf{x}_j \neq \mathbf{x}_j'$ is $\delta_H(\mathbf{x}, \mathbf{x'})$,  we conclude the following:
$$\delta_H(\fen{\left(\mathbf{z},\mathbf{x}\right)},\fen{\left(\mathbf{z},\mathbf{x'}\right)})\leq k \cdot \delta_H(\mathbf{x}, \mathbf{x'}).$$


Note that this immediately implies 
$d_H(\fen{\left(\mathbf{z},\mathbf{x}\right)},\fen{\left(\mathbf{z},\mathbf{x'}\right)}) \leq d_H(\mathbf{x}, \mathbf{x'}). ~~~~~~~~~~~~~~~~~~~~~~~~~~\qedhere $ 
\end{enumerate}
\end{proof}

\color{black}
        
        



The following lemma will provide us a way to construct distributions that cannot be easily distinguished using non-adaptive queries (following a uniformly random index-permutation which we will deploy).

\color{black}

\begin{lem}[{\bf Projection property of $\fen$}]\label{lem:restrictionuniform}
Consider a fixed vector $\mathbf{x} \in \{0,1\}^n$, and let $\mathbf{z} \in [n]^m$ be a string chosen uniformly at random. For any set of indices $Q \subseteq [N]$ such that $\size{Q} \leq \zeta \cdot N/2k$ and $|Q \cap C_j| \leq \zeta \cdot k$ for every $j \in [n]$, the restriction of $\fen(\mathbf{z},\mathbf{x}) \mid_{Q \setminus [b+1]}$ is uniformly distributed over $\{0,1\}^{|Q \setminus [b+1]|}$~\footnote{Recall that the
restriction $\fen(\mathbf{z},\mathbf{x}) \mid_{[b+1]}$ is always the vector $\mathbf{0}\mathbf{1}^b$.}.
\end{lem}

\begin{proof}
For the set of indices $Q$, consider the set $J =\{j: Q \cap C_j \neq \emptyset\}$. From the statement of the lemma, we know that $\size{Q \cap C_j} \leq \zeta \cdot k$ for every $j \in J$. Noting that $\size{J} \leq \size{Q} \leq \zeta \cdot n$, if we consider the restriction $\gen(\mathbf{z}) \mid_{J}$, following Property (ii) of the function $\gen$, we know that $\gen(\mathbf{z}) \mid_{J}$ is uniformly distributed over $[n]^{\size{J}}$.

Now when we call $\sen(i_j, \mathbf{x}_j)$ with $i_j \in [n]$ obtained from $\gen(\mathbf{z})\mid_{J}$, following the above argument, we can say that $i_j$ has been chosen uniformly at random from $[n]$ (and independently from the other $i_{j'}$). Since $\size{Q \cap C_j} \leq \zeta \cdot k$, applying Property (ii) of the function $\sen$, we know that the corresponding restriction of $\sen(i_j, \mathbf{x}_j)$ will be uniformly distributed over $\{0,1\}^{\size{Q \cap C_j}}$. Since $\fen(\mathbf{z},\mathbf{x})= \mathbf{0}(\mathbf{1}^b)\sen(\gen(\mathbf{z})_1,\mathbf{x}_1) \ldots \sen(\gen(\mathbf{z})_n,\mathbf{x}_n)$, combining the above arguments, we conclude that $\fen(\mathbf{z},\mathbf{x}) \mid_{Q \setminus [b+1]}$ is uniformly distributed over $\{0,1\}^{|Q \setminus [b+1]|}$.
\end{proof}




Now we are ready to formally define the property, first constructing a non-index-invariant version to be used in the next index-invariant definition.

\paragraph*{Property $\cP^0_{\mathrm{Gap}}$:}
A distribution $D$ over $\{0,1\}^N$ is in $\cP^0_{\mathrm{Gap}}$ if and only if $D$ satisfies the following conditions:

\begin{enumerate}
    \item[(i)] $D(\mathbf{U}) = \alpha$, where $\mathbf{U}= \mathbf{1}\mathbf{0}^{N-1}$ is the indicator vector for the index $1$.

    \item[(ii)] Consider the set of vectors $\cS= \{\mathbf{V}_1, \ldots, \mathbf{V}_b\}$ in $\{0,1\}^N$ such that for every $i\in [b]$, the $i$-th vector $\mathbf{V}_i$ is of the form $\mathbf{1}^{i+1}\mathbf{0}^{N-1-i}$. Note that $\mathbf{V}_i \mid _B=1^i0^{b-i}$ for $B=\{2,\ldots,b+1\}$. We require that $D(\mathbf{V}_i) = \alpha/b$ for every $i \in [b]$.
     

    \item[(iii)] Consider the set of vectors $\cT= \{\mathbf{W}_0, \ldots, \mathbf{W}_{\lceil \log kn \rceil -1}\}$ (disjoint from $\cS$) in $\{0,1\}^N$ such that for every $\mathbf{W}_i \in \cT$, $\mathbf{W}_i$ is of the form $\mathbf{0}(b(i))(\mathbf{0}^{2^i}\mathbf{1}^{2^i})^{kn/2^{i+1}}$, where $b(i)$ denotes the length $b$ binary representation of $i$~\footnote{If $kn/2^{i+1}$ is not an integer, we trim the rightmost copy of $\mathbf{0}^{2^i}\mathbf{1}^{2^i}$ so that the total length of ``$(\mathbf{0}^{2^i} \mathbf{1}^{2^i})^{kn/2^{i+1}}$'' is exactly $kn$.}. Note that for $i=b+2+j$, with $0 \leq j \leq kn-1$, the sequence $(\mathbf{W}_{0})_i, \ldots,(\mathbf{W}_{{(\lceil \log kn \rceil -1)}})_i$ holds the binary representation of $j$. Also, note that there is an one-to-one correspondence between $\mathbf{W}_i \mid_B$ and $\mathbf{W}_i \mid_{[N] \setminus \{B\} \cup \{1\}}$. We require that $D(\mathbf{W}_i)=\alpha/|\cT|$ for every $\mathbf{W}_i \in \cT$.

    \item[(iv)] $\mbox{Supp}(D)\setminus (\{\mathbf{U}\} \cup \cS \cup \cT)$ consists of valid encodings of at most $n$ vectors from $\{0,1\}^n$ with respect to the functions $\sen:\{0,1\}^l \times \{0,1\} \rightarrow \{0,1\}^k$ and $\gen:[n]^m \rightarrow [n]^n$, for the integers $l,m, k \in \N$ as defined in Definition~\ref{defi:encoding}. That is, there exist vectors $\mathbf{x}_1, \ldots, \mathbf{x}_n \in \{0,1\}^n$ for which $\mbox{Supp}(D) \setminus (\{\mathbf{U}\} \cup \cS \cup \cT) \subseteq \{\fen(\mathbf{z},\mathbf{x}_i) : \mathbf{z} \in [n]^m, i \in [n]\}$. Note that for $D$ to be a distribution, we must have $D(\mbox{Supp}(D)\setminus (\{\mathbf{U}\} \cup \cS \cup \cT))=1-3\alpha$.

\end{enumerate}

\paragraph*{Property $\cP_{\mathrm{Gap}}$:}

A distribution $D$ over $\{0,1\}^N$ is said to be in the property $\cP_{\mathrm{Gap}}$ if $D_{\pi}$ is in $\cP^0_{\mathrm{Gap}}$ for some permutation $\pi:[N]\rightarrow [N]$. \\


\begin{rem}[{\bf Intuition behind the definition of $\cP_{\mathrm{Gap}}$}]
If a distribution $D$ is in $\cP_{\mathrm{Gap}}^0$, then we can easily check (by querying the indexes in $B$) whether a sample from $D$ would be equal to $\mathrm{FE}(\mathbf{z}, \mathbf{x})$ for some $\mathbf{z} \in [n]^m$ and $\mathbf{x} \in \{0,1\}^n$. In that case, individual bits of $\mathbf{x}$ can be decoded by querying the appropriate $C_j$ and then passed to a tester of distributions over $\{0,1\}^n$.

On the other hand, if we take a uniformly random permutation of such a distribution $D$, which keeps it in $\cP_{\mathrm{Gap}}$ (though no longer in $\cP_{\mathrm{Gap}}^0$),  a non-adaptive algorithm will need many queries to capture sufficiently many bits from any $C_j$, and this will enable us to fully hide the identity of $\mathbf{x}$ if fewer queries are performed.

By contrast, an adaptive tester will use relatively few samples that are queried in their entirety to obtain the (permutations of the) special vectors in Items (i) to (iii) of the definition of $\cP_{\mathrm{Gap}}^0$, from which it will be able to fully learn the index-permutation applied to the distribution, and continue to successfully decode individual bits. A few further samples queried in their entirety will ensure that there is very little total weight on vectors that are neither special vectors nor equal to $\fen(\mathbf{z}, \mathbf{x})$ for some $\mathbf{z} \in [n]^m$ and $\mathbf{x} \in \{0,1\}^n$.
\end{rem}

\color{black}

\color{blue}


\color{black}

\paragraph*{Known useful results about support estimation:}
Now we state a lemma which will be required later to describe the adaptive tester for $\cP_{\mathrm{Gap}}$. Informally, it says that whether a distribution $D$ over $\{0,1\}^n$ has support size $s$ or is $\eps$-far from any such distribution, can be tested by taking $\tOh(s)$ samples from $D$, and performing $\tOh(s)$ queries on them.


\begin{lem}[{\bf Support size estimation, Theorem $1.9$ and Corollary $2.3$ of \cite{GoldreichR21a} restated}]\label{pre:suppest}
There exists an algorithm {\sc Supp-Est}$(s,\eps)$ that uses $\tOh(s/\eps^2)$ queries to an unknown distribution $D$ defined over $\{0,1\}^n$, and with probability at least $\frac{9}{10}$ distinguishes whether $D$ has at most $s$ elements in its support or $D$ is $\eps$-far from all such distributions with support size at most $s$. 


\end{lem}

\color{black}

We will also use a lower bound on the support size estimation problem to prove the lower bound on non-adaptive testers for testing $\cP_{\mathrm{Gap}}$. Informally speaking, given a distribution $D$ over $\{1, \ldots, 2n\}$, in order to distinguish in the traditional (non-huge-object) model whether the size of the support of $D$ is $n$, or $D$ is far from all such distributions, $\Omega(\frac{n}{\log n})$ samples are necessary. More formally, we have the following theorem.


\begin{theo}[{\bf Support Estimation Lower bound, Corollary $9$ of \cite{valiant2010clt} restated}]\label{theo:valiantlb}
There exist two distributions $D_{yes}^{\mathrm{Supp}}$ and $D_{no}^{\mathrm{Supp}}$ over distributions
over $\{1, \ldots, 2n\}$, and an $\eta \in (0,1/8)$ such that the following holds:
\begin{enumerate}

\item[(i)] The probability mass of every element in the support of $D_{yes}^{\mathrm{Supp}}$ as well as $D_{no}^{\mathrm{Supp}}$ is a multiple of $1/2n$.

\item[(ii)] $D_{yes}^{\mathrm{Supp}}$ is supported over distributions whose support size is $n$.

\item[(iii)] $D_{no}^{\mathrm{Supp}}$ is supported over distributions whose support size is at least $(1+2\eta)n$, and in particular are $\eta$-far in variation distance from any distribution defined over $\{1, \ldots, 2n\}$ whose support size is $(1+2\eta)n$.

\item[(iv)] If a sequence of $o(\frac{n}{\log n})$ samples from a distribution are drawn according to either $D_{yes}^{\mathrm{Supp}}$ or $D_{no}^{\mathrm{Supp}}$, the resulting distributions over the sample sequences are $1/4$-close to each other.


\end{enumerate}
\end{theo}

We present an adaptive algorithm to test $\cP_{\mathrm{Gap}}$ in Subsection~\ref{sec:quadratic_adaptive_ub} and we prove the lower bound for non-adaptive testers in Subsection~\ref{sec:quadratic_nonadapt_lb}. In Subsection~\ref{sec:quadratic_determine_permutation}, we describe a subroutine to determine the unknown permutation that will be used in our adaptive algorithm in Subsection~\ref{sec:quadratic_adaptive_ub}.

\color{black}

%% file: quadratic_lb_adaptive_algo.tex
\subsection{Determining the permutation $\pi$}\label{sec:quadratic_determine_permutation}

Here we design an algorithm that, given a distribution $D \in \cP_{\mathrm{Gap}}$, can learn with high probability the permutation $\pi$ for which $D_{\pi} \in P^0_{\mathrm{Gap}}$.

The crux of the algorithm is that if $D \in \cP_{\mathrm{Gap}}$, then there exist $\mathbf{U}'=\mathbf{U}_\pi \in \{0,1\}^n$, $\cS'= \mathcal{S}_{\pi} = \{\mathbf{\mathbf{V}}_i'=(\mathbf{V}_i)_\pi: i \in [b]\}$ and $\cT'= \cT_{\pi} =\{\mathbf{W}_j'=(\mathbf{W}_j)_\pi:j \in \{0\}\cup [\lceil \log kn\rceil]-1\}$ in the support of $D$ such that $D(\mathbf{U}')=\alpha$, $D(\mathbf{V}_i')=\alpha/b$ for every $i \in [b]$ and  $D(\mathbf{W}_j')=\alpha/\lceil \log kn\rceil$. Note that $\mathbf{U}$, $\cS$ and $\cT$ are as defined in the property $\cP_{\mathrm{Gap}}^0$.

The main observation is that, if we are given the set of special vectors $\{\mathbf{U'}\} \cup \cS' \cup \cT'$, then we can determine the permutation $\pi$. Our algorithm can find  $\mathbf{U}', \cS'$ and $\cT'$ with high probability, if they exist, by taking $\Oh(\log^2 N/\alpha)=\Oh(\log^2 n/\alpha)$ samples and reading them in their entirety. This is due to the fact that the probability mass of every vector in the set of special vectors is at least $\Omega(\alpha/\log n)$.


The algorithm is described in the following subroutine \findpi (see Algorithm~\ref{algo:perm})~\footnote{This algorithm is not adaptive in itself, but its output is used adaptively in the testing algorithm described later.}.

\begin{algorithm}[H]\label{algo:perm}

\SetAlgoLined

\caption{\findpi}
\label{alg:findpi}
\KwIn{Sample and Query access to a distribution $D$ over $\{0,1\}^N$.}
\KwOut{Either a permutation $\pi:[N] \rightarrow [N]$, or {\sc Fail}.}
\LinesNumbered
\begin{description}
     \item[(i)] First take a multi-set $\cX$ of $\Oh(\log^2 N/\alpha)$ samples from $D$, and query all the entries of the sampled vectors of $\cX$ to know the vectors of $\cX$ completely.

    \item[(ii)] Find the set of distinct vectors in $\cX$ that have exactly one $1$. If no such vector exists or there \\ is more than one such vector, {\sc Fail}. Otherwise, denote by $\mathbf{U'}$ the vector that has exactly one $1$, and denote the corresponding index by $i^*$. Set $\pi(i^*)=1$, and proceed to the next step.    

    \item[(iii)] Find the set of distinct vectors $\cS' \subseteq \cX \setminus \{\mathbf{U'}\}$ such that every vector in $\cS'$ has $1$ at the index $i^*$ and has at least another $1$ among other indices. If no such vector exists, or $\size{\cS'} \neq b$, {\sc Fail}. Otherwise, if the vectors of $\cS'$ form a chain $\mathbf{V}'_1, \ldots, \mathbf{V}'_b$, where $\mathbf{V}'_j$ has exactly $j+1$ many $1$, then set $\pi(i_j)=j+1$, where $i_j$ is the index where $\mathbf{V}'_j$ has $1$, but $\mathbf{V}'_{j-1}$ has $0$ there, for every $j \in [b]$ (denoting $\mathbf{V'}_0=\mathbf{U'}$ for the purpose here). Also, set $B'=(i_1, \ldots,i_b)$. If $\cS'$ does not form a chain $\mathbf{V}'_1, \ldots, \mathbf{V}'_b$ as mentioned above, {\sc Fail}.

    \item[(iv)] Let $\cT' \subseteq \cX$ be the set of distinct vectors such that every vector in $\cT'$ has $0$ at the index $i^*$, and does not have $1$ in all indices of $B'$. If no such vector exists, {\sc Fail}. For every $j$, denote by $\mathbf{W}'_j$ the vector in $\cT'$ for which $\mathbf{W}'_j \mid_{B'} = b(j)$, where $b(j)$ denotes the binary representation of $j$. For every $j \in \{0\} \cup [\lceil \log kn \rceil -1]$, if either there are no vectors $\mathbf{W}'_j \in \cT'$ or there is more than one distinct vector with $\mathbf{W'}_j \mid_{B'}=b(j)$, {\sc Fail}.
    Also, if there is any vector in $\mathbf{W}'_j \in \cT'$ such that $\mathbf{W}'_j \mid_{B'} = b(j)$ for $\log kn \leq j < 2^b -1$, {\sc Fail}.
    
    \item[(v)] For any $i \in [N]\setminus (\{i^*\} \cup B')$, let $l_i$ be the integer with binary representation $(\mathbf{W'}_0)_i, \ldots, (\mathbf{W'}_{\lceil \log kn \rceil -1})_i$.
Set $\pi(i)=b+2+l_i$ for each $i \in [N]\setminus (\{i^*\} \cup B')$. If $\pi$ is not a permutation of $[N]$, {\sc Fail}.

    \item[(vi)] Take another multi-set $\cX'$ of $\Oh(\log^2 N/\alpha)$ samples from $D$, and query all the entries of the sampled vectors of $\cX'$ to know the vectors of $\cX'$ completely. Let $\cY$ be the set of vectors in \\ $\cX'$ such that $\cY=\{\mathbf{Z} \in \cX': \mathbf{Z} \mid_{{\{i^*\} \cup B'}} \neq \mathbf{0}\mathbf{1}^b \}$. If $\size{\cY}/\size{\cX'} > 4\alpha$, {\sc Fail}. Otherwise, output \\ the permutation $\pi$.
    
    
\end{description}

\end{algorithm}

\vspace{10pt}

Let us start by analyzing  the query complexity of \findpi.

\begin{lem}[{\bf Query complexity of \findpi}]\label{lem:queryfindpi}
The query complexity of the above defined \findpi is $\tOh(N)$.
\end{lem}

\begin{proof}
Note that \findpi takes a multi-set $\cX$ of $\Oh(\log^2 N/\alpha)$ samples from $D$ in Step (i), and queries them completely. So, \findpi performs $\Oh(N \log^2 N/\alpha)$ queries in Step (i). \findpi does not perform any new queries in Step (ii), Step (iii), Step (iv) and Step (v). Finally, \findpi takes another multi-set $\cX'$ of $\Oh(\log^2 N/\alpha)$ samples from $D$ and queries them completely, similar to Step (i). Recalling that $\alpha=1/\log n$, the query complexity of \findpi is $\tOh(N)= \tOh(n)$ in total.
\end{proof}

Now we proceed to prove the correctness of \findpi.

\begin{lem}[{\bf Guarantee when $D \in \cP_{\mathrm{Gap}}$}]\label{lem:findpicompleteness}
If $D$ is a distribution defined over $\{0,1\}^N$ such that $D \in \cP_{\mathrm{Gap}}$, then with probability at least $9/10$, \findpi reports the permutation $\pi$ such that $D_{\pi} \in \cP_{\mathrm{Gap}}^0$.
\end{lem}

\color{black}

We prove the above lemma by a series of intermediate lemmas.  In the following lemmas, we consider $\mathbf{U}$, $\cS$ and $\cT$ as per the definition of $\cP^0_{\mathrm{Gap}}$. Also, consider the permutation $\pi$ such that $D_\pi \in \cP_{\mathrm{Gap}}$.



\begin{lem}[{\bf Correctly finding $\pi^{-1}(1)$}]\label{lem:hittingU}
With probability at least $1-1/N^3$, $\cX$ will contain the vector $\mathbf{U'}$ for which $\mathbf{U'}_{\pi}=\mathbf{U}$, and $i^*=\pi^{-1}(1)$ will be identified correctly. Moreover, \findpi proceeds to Step (iii).
\end{lem}

\begin{proof}
By the definition of $\cP^0_{\mathrm{Gap}}$, the vector $\mathbf{U'}$ is the only vector in the support of $D$ containing a single $1$. Since $D(\mathbf{U'}) = D(\mathbf{U}_{\pi^{-1}(1)}) = \alpha$, and we are taking $\size{\cX}$ many samples from $D$, the probability that $\mathbf{U'}$ will not appear in $\cX$ is at most $(1- \alpha)^{\size{X}} \leq \frac{1}{N^3}$. Thus, with probability at least $1 - \frac{1}{N^3}$, $\mathbf{U'} \in \cX$ and \findpi in Step (ii) proceeds to the next step.
\end{proof}

\begin{lem}[{\bf Correctly finding $B' = \pi^{-1}(B)$}]\label{lem:hittingS}
With probability at least $1-1/N^3$, the algorithm \findpi will correctly identify $\mathbf{V'}_1, \ldots, \mathbf{V'}_b$ for which $\mathbf{V'}_{i,\pi} = \mathbf{V}_i$, and $B'={\pi^{-1}(2), \ldots, \pi^{-1}(b+1)}$ will be identified correctly as well. Moreover, \findpi proceeds to Step (iv).

\end{lem}

\begin{proof}

Let $\mathbf{V'}_1, \ldots, \mathbf{V'}_b$ denote the vectors for which $\mathbf{V'}_{i,\pi}=\mathbf{V}_i$ for every $i$. Note that these are the only vectors outside $\mathbf{U'}$ in the support of $D$ that have $1$ at the index $i^*$.
As $D(\mathbf{V}'_i)= \frac{\alpha}{b}$, the probability that $\mathbf{V}'_i$ does not appear in $\cX$ is at most $(1- \frac{\alpha}{b})^{\size{X}}$. Since $\size{\cX} = \Oh(\log^2 N/\alpha)$ and $b = \Oh(\log \log kn)$, the probability that $\mathbf{V}'_i \in \cX$ is at least $1 - \frac{1}{N^4}$. Using the union bound over all the vectors of $\cS'$, with probability at least $1-1/N^3$, we know that all of these vectors are in $\cX$, in which case they are identified correctly, so $B'$ is identified correctly as well, and \findpi in Step (iii) proceeds to the next step.
\end{proof}

\begin{lem}[{\bf Correctly identifying $\pi^{-1}(b+2), \ldots, \pi^{-1}(N)$}]\label{lem:hittingT}
Let $\mathbf{W'}_1, \ldots, \mathbf{W'}_{\lceil \log kn \rceil -1}$ denote the vectors for which $\mathbf{W'}_{j,\pi} = \mathbf{W}_j$ for every $j$. With probability at least $1-1/N^3$, all these vectors appear in $\cX$, in which case they are identified correctly, and so are $\pi^{-1}(b+2), \ldots, \pi^{-1}(N)$. Moreover, \findpi proceeds to Step (vi).

\end{lem}

The proof of the above lemma is similar to the proof of Lemma~\ref{lem:hittingS} and is omitted. Note that from Lemma~\ref{lem:hittingU}, Lemma~\ref{lem:hittingS} and Lemma~\ref{lem:hittingT}, we know that with probability at least $1-o(1)$, the algorithm \findpi has correctly determined the permutation $\pi$ and proceeded to Step (vi). We will finish up the proof of Lemma~\ref{lem:findpicompleteness} using the following lemma.


\begin{lem}
The probability that \findpi outputs {\sc Fail} in Step (vi) (instead of outputting $\pi$) is at most $1/N^3$.

\end{lem}

\begin{proof}
As $D \in \cP_{\mathrm{Gap}}$, from the description of the property, we know that $D(\{\mathbf{U'}\} \cup \cS' \cup \cT')= 3 \alpha$. As $\size{\cX'}= \Oh(\log^2 N/\alpha)$, using the Chernoff bound (Lemma~\ref{lem:cher_bound1}), we have the result.
\end{proof}

Combining the above lemmas, we conclude that with probability at least $9/10$, the algorithm \findpi outputs a correct permutation $\pi$, completing the proof of Lemma~\ref{lem:findpicompleteness}.

To conclude this section, we show that with high probability, we will not output $\pi$ for which too much weight is placed outside the ``encoded part'' of the distribution.


\begin{lem}\label{lem:findpisumfourthtype}
For any distribution $D$ (regardless of whether $D$ is in $\cP_{\mathrm{Gap}}$ or not), the probability that \findpi outputs a permutation $\pi$ for which $D(\{\mathbf{X}:\mathbf{X} \mid_{\{i^*\} \cup B'}=\mathbf{0}\mathbf{1}^b\}) \leq 1-5\alpha$ is at most $1/10$.

\end{lem}

\begin{proof}
Recall the set of vectors $\cY$ as defined in Step (vi) of \findpi: $\cY=\{\mathbf{Z} \in \cX': \mathbf{Z} \mid_{{{i^*} \cup B'}} \neq \mathbf{0}\mathbf{1}^b \}$, where $\cX'$ is the multi-set of (new) samples obtained in Step (vi) of \findpi. Consider a distribution $D$ such that $D(\{\mathbf{X}:\mathbf{X}_{\{i^*\} \cup B'}=\mathbf{0}\mathbf{1}^b\}) \leq 1-5\alpha$. This implies that $\E\left[\size{\cY}/\size{\cX'}\right] \geq 5 \alpha$. As $\size{\cX'} = \Oh(\log^2 N/\alpha)$, using the Chernoff bound (Lemma~\ref{lem:cher_bound1}), we obtain that with probability at least $9/10$, the algorithm \findpi outputs {\sc Fail} in Step (vi), and does not output any permutation $\pi$. This completes the proof.
\end{proof}

\subsection{The upper bound on adaptive testing for property $\cP_{\mathrm{Gap}}$}\label{sec:quadratic_adaptive_ub}

In this subsection, we design the adaptive tester for the property $\cP_{\mathrm{Gap}}$. Given a distribution $D$ over $\{0,1\}^N$, with high probability, \algadpt outputs {\sc Accept} when $D \in \cP_{\mathrm{Gap}}$, and outputs {\sc Reject} when $D$ is far from $\cP_{\mathrm{Gap}}$. The formal adaptive algorithm is presented in \algadpt (see Algorithm~\ref{algo:adaptive}). Note that it has only two adaptive steps.


In the first adaptive step, our tester \algadpt starts by calling the algorithm \findpi (as described in Subsection~\ref{sec:quadratic_determine_permutation}) whose  query complexity is $\tOh(n)$. If $D \in \cP_{\mathrm{Gap}}$, with high probability, \findpi returns the permutation $\pi$ such that $D_{\pi} \in \cP_{\mathrm{Gap}}^0$. Once $\pi$ is known, when we obtain a sample $\mathbf{X}$ from $D$, we can consider it as $\mathbf{X}_\pi$ from $D_{\pi}$. Also, from the structure of the vectors in the support of the distributions in $P^0_{\mathrm{Gap}}$, we can decide whether $X_{\pi}$ is a special vector, that is, $X_{\pi} \in \{\mathbf{U}\} \cup \cS \cup \cT$ or $\mathbf{X}_\pi$ is an encoding vector, that is, $\mathbf{X}_\pi=\fen(\mathbf{z},\mathbf{x})$ for some $\mathbf{z}\in [n]^m$ and $\mathbf{x} \in \{0,1\}^n$. Observe that, in the later case, we can decode any bit of $\mathbf{x}$ (say $\mathbf{x}_j$) by finding $\mathbf{X}_\pi$ projected into $C_j$, which can be done by performing $\Oh(\log n)$ queries.


As the second adaptive step, our algorithm asks for a sequence $\cY$ of $\Oh(n/\eps)$ samples from $D$, that is, from $D_{\pi}$. Let $\cY'\subseteq \cY$ be the sequence of encoding vectors in $\cY$. We now call {\sc Supp-Est}$(\cY', \eps/3)$ (from Lemma~\ref{pre:suppest}), and depending on its output, \algadpt either reports {\sc Accept} or {\sc Reject}. Note that we can execute every query by {\sc Supp-Est}$(\cY', \eps/3)$, by performing $\Oh(\log n)$ queries to the corresponding sample in $\cY'$ as discussed above.

When $D \in \cP_{\mathrm{Gap}}$ (that is, $D_\pi \in \cP_{\mathrm{Gap}}^0$ for the permutation $\pi$), the set of encoding vectors in $D_\pi$ is the encoding of at most $n$ vectors in $\{0,1\}^n$. So, in that case, \algadpt reports {\sc Accept} with high probability. Now consider the case where $D$ is $\eps$-far from $\cP_{\mathrm{Gap}}$. If \algadpt has not rejected $D$ before calling {\sc Supp-Est}$(\cY', \eps/3)$, we will show that the distribution over $\{0,1\}^n$ induced by the vectors decoded from the encoding vectors in $D_{\pi}$ is ${\eps}/3$-far from having support size $n$. Then, \algadpt will still reject $D$ with high probability.


\begin{algorithm}[H]\label{algo:adaptive}
\SetAlgoLined

\caption{\algadpt}
\label{alg:adaptive}
\KwIn{Sample and Query access to a distribution $D$ over $\{0,1\}^N$, and a parameter $\eps \in (0,1)$.}
\KwOut{Either {\sc Accept} or {\sc Reject}.}
\LinesNumbered
\begin{description}
     \item[(i)] Call \findpi. If \findpi returns {\sc Fail}, {\sc Reject}. Otherwise, let $\pi$ be the permutation returned by \findpi. Denote for convenience $i^*=\pi^{-1}(1), B'=\pi^{-1}(B)$, and $C'_j=\pi^{-1}(C_j)$ for every $j \in [n]$.
     
     \item[(ii)] Take a multi-set $\cX$ of $\Oh(1/\eps)$ samples from $D$, and query all the entries of the sampled vectors of $\cX$ to know the vectors of $\cX$ completely. If there is any vector $\mathbf{X}$ for which $\mathbf{X} \mid_{\{i^*\} \cup B'} = \mathbf{0}\mathbf{1}^b$ (according to the permutation $\pi$ obtained from Step (i)) for which $\mathbf{X}_{\pi}$ is not in the image of $\fen$ (i.e. it is not a valid encoding of any vector in $\{0,1\}^n$), {\sc Reject}. Otherwise, proceed to the next step.

    \item[(iii)]Take a sequence of samples $\cY$ such that $\size{\cY}=\Oh(n/\eps)$ from $D$. Now construct the sequence of vectors $\cY'$ such that $\cY'= \{\mathbf{Y} \in \cY: \mathbf{Y} \mid_{\{i^*\} \cup B'} = \mathbf{0} \mathbf{1}^b\}$ by querying the indices corresponding to $\{i^*\} \cup B'$.

    \item[(iv)] Call {\sc Supp-Est}$(\cY', \eps/3)$ (from Lemma~\ref{pre:suppest}), where a query to an index $j$ is simulated by querying the indices of $C'_j$ and decoding the obtained vector with respect to to $\sen$ (that is, checking whether the restriction of the queried vector to $C'_j$ is equal to $\sen(i,0)$ for some $i$, or equal to $\sen(i,1)$ for some i). {\sc Reject} if any of the following conditions hold:
    \begin{itemize}
        \item[(a)] $\size{\cY'}/\size{\cY} \leq 1/2$ (due to the absence of sufficiently many samples in $\cY'$ to apply {\sc Supp-Est}).

        \item[(b)] {\sc Supp-Est}$(\cY', \eps/3)$ queries an index $j$ from some $\mathbf{Y}_i$ corresponding to an invalid encoding of $\mathbf{Y}_i \mid_{C'_j}$ (that is, when $\mathbf{Y}_i \mid_{C'_j}$ is not in the image of $\sen$).
        
        \item[(c)] {\sc Supp-Est}$(\cY',\eps/3)$ outputs {\sc Reject}.
    \end{itemize}

    Otherwise, {\sc Accept}.

\end{description}

\end{algorithm}

\vspace{10pt}

Let us first discuss the query complexity of \algadpt.

\begin{lem}[{\bf Query complexity of \algadpt}]
The query complexity of the adaptive tester \algadpt for testing the property $\cP_{\mathrm{Gap}}$ is $\tOh(N)= \tOh(n)$.    
\end{lem}

\begin{proof}
Note that \algadpt calls the algorithm \findpi in Step (i). Following the query complexity lemma of \findpi (Lemma~\ref{lem:queryfindpi}), we know that \findpi performs $\tOh(N)$ queries.

For every sample taken in Step (ii), the sampled vectors of the multi-set $\cX$ are queried completely. Since we take $\Oh(1/\eps)$ samples, this step requires $\Oh(N/\eps) = \widetilde{\Oh}(n/\eps)$ queries in total.  
    


Then in Step (iii), \algadpt takes a multi-set $\cY$ of $\Oh(n/\eps)$ samples, and queries for the indices in $\{i^*\} \cup B'$ to obtain the vectors in $\cY'$, which takes $\Oh(n \log \log kn/\eps)$ queries. Finally, in Step (iv), \algadpt calls the algorithm {\sc Supp-Est}, which performs $\tOh(n)$ queries (following Lemma~\ref{pre:suppest}), each of them simulated by $\Oh(\log n)$ queries to some $\mathbf{Y}_i \mid_{{C'_j}}$. Thus, in total, \algadpt performs $\tOh(N) = \tOh(n)$ queries.
\end{proof}

Now we prove the correctness of \algadpt. We will start with the completeness proof.

\begin{lem}[{\bf Completeness of \algadpt}]
Let $D$ be a distribution defined over $\{0,1\}^N$. If $D \in \cP_{\mathrm{Gap}}$, then the algorithm \algadpt will output {\sc Accept} with probability at least $2/3$.
\end{lem}

\begin{proof}
Consider a distribution $D \in \cP_{\mathrm{Gap}}$. From the completeness lemma of \findpi (Lemma~\ref{lem:findpicompleteness}), we infer that with probability at least $9/10$, \findpi returns the correct permutation $\pi$ in Step (i). Then, by the definition of $\cP_{\mathrm{Gap}}$, the algorithm \algadpt can never encounter any samples with invalid encodings in Step (ii) which could cause it to {\sc Reject}. Thus, with probability at least $9/10$, the algorithm proceeds, with the correct permutation $\pi$, to Step (iii) and Step (iv).
    

As $D \in \cP_{\mathrm{Gap}}$, $D(\{\mathbf{U'}\} \cup \cS' \cup \cT') = 3 \alpha$. Since $\size{\cY}= \Oh(n/\eps)$, using the Chernoff bound (Lemma~\ref{lem:cher_bound2} (ii)), we can say that, with probability at least $9/10$, $\size{\cY'}/\size{\cY} \geq 1/2$. Moreover, as the vectors in $\cY'$ are valid encodings with respect to the function $\fen$ of at most $n$ vectors from $[2n]$, following the support estimation upper bound lemma (Lemma~\ref{pre:suppest}), we obtain that {\sc Supp-Est} outputs {\sc Accept} with probability at least $9/10$. Combining these, we conclude that \algadpt outputs {\sc Accept} with probability at least $2/3$.
\end{proof}

Now we prove that when $D$ is $\eps$-far from $\cP_{\mathrm{Gap}}$, \algadpt will output {\sc Reject} with probability at least $2/3$.

\begin{lem}[{\bf Soundness of \algadpt}]\label{lem:sound}
    Let $\eps \in (0,1)$ be a proximity parameter. Assume that $D$ is a distribution defined over $\{0,1\}^N$ such that $D$ is $\eps$-far from $\cP_{\mathrm{Gap}}$. Then \algadpt outputs {\sc Reject} with probability at least $2/3$.
\end{lem}


From the description of \algadpt (Algorithm~\ref{algo:adaptive}), if the tester reports {\sc Reject} before executing all the steps of {{\sc Supp-Est}}$(\cY',\eps/3)$ in Step (iv),  then we are done. So, let us assume that \algadpt executes all the steps of {{\sc Supp-Est}}$(\cY',\eps/3)$. Let $\cY'$ be the set of samples from a distribution $\dhash$ over $\{0,1\}^n$ as it is presented to {{\sc Supp-Est}}$(\cY',\eps/3)$. Note that $\dhash$ is unknown and we are accessing $\dhash$ indirectly via decoding samples from $D$ over $\{0,1\}^N$. From the correctness {{\sc Supp-Est}}$(\cY',\eps/3)$, we will be done with the proof of Lemma~\ref{lem:sound} by proving the following lemma.

\begin{lem}[{\bf Property of the decoded distribution}]\label{lem:far-D-star}
$\dhash$ is $\eps/3$-far from having support size at most $n$. 
\end{lem}

{We prove the above lemma using a series of claims.}
Let $D$ be a distribution which is $\eps$-far from $\cP_{\mathrm{Gap}}$, and $\mathcal{V}$ denote the set $\{\mathbf{X} \in \mbox{Supp}(D) : \mathbf{X} \mid_{\{i^*\}\cup B'}=\mathbf{0}\mathbf{1}^b\}$, and let us define $\mathcal{U}= \mbox{Supp}(D) \setminus \mathcal{V}$. Let us start with the following observation.


\begin{obs}\label{obs:sumdothervec}
$D(\mathcal{U}) \leq 5\alpha$, unless the algorithm \algadpt has rejected with probability at least $1-1/N^3$ in Step (i).
\end{obs}

\begin{proof}
Since \algadpt in Step (i) invokes the algorithm \findpi, this follows immediately from Lemma~\ref{lem:findpisumfourthtype}.
\end{proof}

Let $\pi$ be the permutation returned by \findpi. Now assume $\mathcal{V}^{\mbox{inv}} \subseteq \mathcal{V}$ denotes the following set of vectors:
$$\mathcal{V}^{\mbox{inv}} = \{\mathbf{X} \in \mathcal{V}: \mathbf{X}_{\pi} \neq \fen(\mathbf{z},\mathbf{x}) \ \mbox{for all} \ \mathbf{z} \in [n]^m, \mathbf{x} \in \{0,1\}^n\}$$

For every vector $\mathbf{V} \in \mathcal{V}^{\mbox{inv}}$, let $\Gamma'_{\mathbf{V}}=\{j\in [n]: \mathbf{V}\mid_{C'_j} \ \mbox{is not in the image of} \ \sen\}$ denotes the set of indices in $[n]$ of chunks of all the ``\emph{locally invalid}'' encodings in the vector $\mathbf{V}$~\footnote{Note that it may be the case that $\Gamma'_{\mathbf{V}}=\emptyset$, for example when for every $j \in [n]$, we have $\mathbf{V} \mid_{C'_j} = \sen(i_j, \mathbf{x}_j)$, for some $i_1, \ldots, i_n$ and $\mathbf{x}_1, \ldots, \mathbf{x}_n$ for which $i_1, \ldots, i_n$ are not in the image of $\gen$.}. Now we have the following observation.

\begin{obs}\label{obs:v-enc}
$D(\mathcal{V}^{\mbox{inv}}) \leq {\eps}/10$.
\end{obs}
The above observation holds as otherwise, \algadpt would have rejected in Step (ii) with probability at least $2/3$.

\vspace{2 pt}
Let us define a distribution $D_1$ over $\{0,1\}^N$ using the following procedure:
\begin{enumerate}
     \item[(i)] Set $D_1(\mathbf{X})= D(\mathbf{X})$ for every $\mathbf{X} \in \mathcal{U}$. 
    
    \item[(ii)] Recall that $\Gamma'_{\mathbf{V}}=\{j\in [n]: \mathbf{V}\mid_{C'_j} \ \mbox{is not in the image of} \ \sen\}$ for every vector $\mathbf{V} \in \mathcal{V}^{\mbox{inv}}$. For every vector $\mathbf{V}\in \mathcal{V}^{\mbox{inv}}$, we perform the following steps:

    \begin{enumerate}
        \item[(a)] For every $j \notin \Gamma'_{\mathbf{V}}$, decode the vector $\mathbf{V} \mid_{C'_j}$ using $\sen$ to obtain $\mathbf{x}_j \in \{0,1\}$.

        \item[(b)] For every $j \in \Gamma'_{\mathbf{V}}$, choose an arbitrary value $\mathbf{x}_j$ from $\{0,1\}$.

        \item[(c)] Using $\mathbf{x} = (\mathbf{x}_1, \ldots, \mathbf{x}_n)$ obtained from (a) and (b), construct a new vector $\mathbf{V}'$ for which $\mathbf{V}'_{\pi}=\fen(\mathbf{z},\mathbf{x})$ for an arbitrary $\mathbf{z} \in [n]^m$, where $\pi$ is the permutation obtained from \findpi in Step (i) of \algadpt.
    \end{enumerate}    

    \item[(iii)] For every vector $\mathbf{V} \in \mathcal{V} \setminus \mathcal{V}^{\mbox{inv}}$, set $\mathbf{V}' = \mathbf{V}$.

    \item[(iv)] Finally define $D_1(\mathbf{W})=\sum_{\mathbf{V}: \mathbf{V'} =\mathbf{W}} D(\mathbf{V})$ for every $\mathbf{W} \in \mathcal{V}$.

\end{enumerate}

Let $\mathcal{V}'$ be the set of vectors in $\{0,1\}^N$ that are in the support of $D_1$ but not in $\mathcal{U}$, that is, $\mathcal{V}' = \{\mathbf{X} : \mathbf{X} \in \mathrm{Supp}(D_1) \setminus \mathcal{U}\}$. From the construction of $D_1$, the following observation follows.
\begin{obs}\label{obs:v-dash}
$D_1(\mathcal{U})=D(\mathcal{U})\leq 5 \alpha$ and $D_1(\mathcal{V}')=D(\mathcal{V})=1-D(\mathcal{U})\geq 1-5\alpha.$
\end{obs}

Now we prove that the distributions $D$ and $D_1$ are not far in Earth Mover Distance.

\begin{lem}\label{lem:d1d2emd}
The Earth Mover Distance between $D$ and $D_1$ is at most $\eps/10$.
\end{lem}

\begin{proof}
Recall that the EMD between $D$ and $D_1$ is the solution to the following LP:
\begin{eqnarray*}
    &&\mbox{Minimize}\quad{\sum_{\mathbf{X},\mathbf{Y} \in \{0,1\}^N} f_{\mathbf{X}\mathbf{Y}} d_H(\mathbf{X},\mathbf{Y})}\\
    &&\mbox{Subject to}\quad\sum_{\mathbf{Y} \in \{0,1\}^N} f_{\mathbf{X}\mathbf{Y}} = D(\mathbf{X})~\forall \ \mathbf{X} \in \{0,1\}^N \; \mbox{and} \;
    \sum_{\mathbf{X} \in \{0,1\}^N} f_{\mathbf{X}\mathbf{Y}} = D_1(\mathbf{Y})~\forall \ \mathbf{Y} \in \{0,1\}^N.
\end{eqnarray*}

Consider the flow $f^*$ such that $f^*_{\mathbf{XX}}=D(\mathbf{X})$ for every $\mathbf{X} \in \mathcal{U} \cup (\mathcal{V} \setminus \mathcal{V}^{\mbox{inv}})$, $f^*_{\mathbf{VV'}}=D_1(\mathbf{V})$ for every $\mathbf{V} \in \mathcal{V}^{\mbox{inv}}$, and $f^*_{\mathbf{X}\mathbf{Y}}=0$ for all other vectors. Then we have the following:
\begin{eqnarray*}
d_{EM}(D,D_1) &\leq& \sum_{{\bf X},{\bf Y} \in \{0,1\}^N} f^*_{{\bf X Y}}d_H({\bf X},{\bf Y}) \\
&\leq& \sum_{\mathbf{X} \in \{0,1\}^N \setminus \mathcal{V}^{\mbox{inv}}} f^*_{{\bf X X}}d_H({\bf X},{\bf X}) + \sum_{{\bf V} \in \mathcal{V}^{\mbox{inv}}} f^*_{{\bf V} {\bf V}'}d_H({\bf V},{\bf V'})\\
&\leq& 0 + \sum_{\mathbf{V}\in \mathcal{V}^{\mbox{inv}}}D(\mathbf{V})d_H({\bf V},{\bf V'}).
\end{eqnarray*}

To bound the second term of the last expression, note that $$\sum\limits_{\mathbf{V}\in \mathcal{V}^{\mbox{inv}}}D(\mathbf{V})d_H({\bf V},{\bf V'}) \leq D(\mathcal{V}^{\mbox{inv}}) \leq \eps/10.$$ This follows from Observation~\ref{obs:v-enc}. Thus, we conclude that $d_{EM}(D,D_1) \leq \eps/10$, completing the proof of the lemma.
\end{proof}

Now we have the following observation regarding the rejection probabilities of \algadpt for the distributions $D$ and $D_1$. This will imply that, as we are executing all steps of {\sc Supp-Est}$(\cY',\eps/3)$, the steps of our algorithm are oblivious to both $D$ and $D_1$. That is, we can assume that the input to the algorithm \algadpt is the distribution $D_1$ instead of $D$.

\begin{obs}
The probability that the tester \algadpt outputs {\sc Reject} in Step (iv) where the input distribution is $D$ is at least as large as the probability that \algadpt outputs {\sc Reject} in Step (iv) when the input distribution is $D_1$.
\end{obs}

\begin{proof}
Note that in the distribution $D$, there can be some vectors in $\mbox{Supp}(D)$ that are not valid encodings with respect to the function $\fen$. Thus during its execution, the tester \algadpt can {\sc Reject} $D$ by Condition (ii) and Condition (iv) (b). However, by the construction of $D_1$ from $D$, we have replaced the invalid encoding vectors with valid encoding vectors. Thus, the only difference it makes here is that \algadpt may eventually accept a sample from $D_1$ when encountering such a place where a sample from $D$ would have been immediately rejected by Condition (ii) or Condition (iv) (b). Other than this difference, the distributions $D$ and $D_1$ are identical. So, the probability that \algadpt will {\sc Reject} $D$ is at least as large as the probability that it {\sc Rejects} $D_1$.
\end{proof}

Now, let us come back to the proof of Lemma~\ref{lem:far-D-star}. Recall that $\mathcal{V}' = \{\mathbf{X} : \mathbf{X} \in \mathrm{Supp}(D_1) \setminus \mathcal{U}\}$. Let us define the distribution $\dhash$ over $\{0,1\}^n$ referred to in Lemma~\ref{lem:far-D-star}.  For $\mathbf{x}\in \{0,1\}^n$, we have the following:
\begin{equation}\label{def:enc-dist}
    \dhash(\mathbf{x})=D_1^{{dec}}(\mathbf{x})=\frac{1}{D_1(\mathcal{V}')}\sum_{\substack{\mathbf{Y}_{\pi} = \fen(\mathbf{z},\mathbf{x}) \\ \mbox{for some} \ \mathbf{z} \in [n]^m}} D_1(\mathbf{Y}) = \frac{1}{D_1(\mathcal{V}')}\sum_{\mathbf{z}\in [n]^m}D_1(\mathrm{FE}(\mathbf{z},\mathbf{x})_{\pi^{-1}}).
\end{equation}

For the sake of contradiction, assume that $\dhash=D_1^{dec}$ is $\eps/3$-close to having support size at most $n$. Let $D_{2}$ be a distribution over $\{0,1\}^{n}$ having support size at most $n$ such that the Earth Mover Distance between $D_2$ and $D_1^{dec}$ is at most $\eps/3$.

Given the distribution $D_{2}$ over $\{0,1\}^{n}$, and the flow $f'_{\mathbf{x}\mathbf{y}}$ from $D_2$ to $D_1^{dec}$ realizing the EMD of at most $\eps/3$ between them, let us consider the distribution $D_2^{enc}$ over $\{0,1\}^N$ as follows:

\begin{enumerate}
    \item[(i)] For any $\mathbf{X} \in \mathcal{V'}$, for which $\mathbf{X}_{\pi} = \fen(\mathbf{z},\mathbf{x})$ for some $\mathbf{z} \in [n]^m$, set:
    \begin{center}$D_2^{enc}(\mathbf{X})=\sum\limits_{\mathbf{y} \in \{0,1\}^n}f'_{\mathbf{x}\mathbf{y}}\frac{D_1(\mathrm{FE}(\mathbf{z},\mathbf{y})_{\pi^{-1}})}{D_1^{dec}(\mathbf{y})}$.
    \end{center}
    
    \item[(ii)] For every $\mathbf{X} \in \mathcal{U}$, set $D_2^{enc}(\mathbf{X}) = D_1(\mathbf{X})$.
\end{enumerate}


The following observation follows from Observation~\ref{obs:v-dash} and the construction of $D_2^{enc}$.

\begin{obs}\label{obs:v-dash-1}
$D_2^{enc}(\mathcal{U})=D_1(\mathcal{U})\leq 5 \alpha$ and $D_2^{enc}(\mathcal{V}')=D_1(\mathcal{V}')=1-D_1(\mathcal{U})\geq 1-5\alpha.$
\end{obs}
  
The following two lemmas bound the distance of $D_2^{enc}$ from $\cP_{\mathrm{Gap}}$ and from $D_1$, where $D_1^{dec}$ is $\eps/3$-close to having support size at most $n$. We will prove these two lemmas later.

\begin{lem}\label{cl:close-P-gap}
$D_2^{enc}$ is $6\alpha$-close to $\cP_{\mathrm{Gap}}$.
\end{lem}

\begin{lem}\label{cl:emd-far}
The Earth Mover Distance between $D_{2}^{enc}$ and $D_1$ is at most $\eps/3$.
\end{lem}


Assuming Lemma~\ref{cl:close-P-gap} and Lemma~\ref{cl:emd-far} hold, now we proceed to prove Lemma~\ref{lem:far-D-star}.

\begin{proof}[Proof of Lemma~\ref{lem:far-D-star}]

From Lemma~\ref{lem:d1d2emd}, we know that $d_{EM}(D,D_1)\leq \eps/10$. So, the above two lemmas imply that $D$ is $(\eps/3+\eps/10+6\alpha)=2 \eps/3$-close to $\cP_{\mathrm{Gap}}$, which contradicts the fact that $D$ is $\eps$-far from $\cP_{\mathrm{Gap}}$. This completes the proof of the lemma.
\end{proof}

\color{black}

Now we will prove Lemma~\ref{cl:close-P-gap} and Lemma~\ref{cl:emd-far}.

\begin{proof}[Proof of Lemma\ref{cl:close-P-gap}]

We define another distribution $D_3$ over $\{0,1\}^N$ from $D_2^{enc}$ such that $D_3$ is in $\cP_{\mathrm{Gap}}$ and $d_{EM}(D_2^{enc},D_3)\leq 6 \alpha$ as follows:

\begin{enumerate}
    \item[(i)] $D_3(\mathbf{U}') = \alpha$.
    
    \item[(ii)]  $D_3(\mathbf{X}) = \frac{\alpha}{b}$ for every $\mathbf{X} \in \cS'$, $D_3(\mathbf{X}) = \frac{\alpha}{\lceil \log kn \rceil}$ for every $\mathbf{X} \in \cT'$.
    
    \item[(iii)] $D_3(\mathbf{X})= {(1 - 3 \alpha)} \cdot \frac{D_2^{enc}(\mathbf{X})}{D_2^{enc}(\mathcal{V}')}$ for every $\mathbf{X} \in \mathcal{V'}$.
\end{enumerate}

Recall that $D_2$ is a distribution over $\{0,1\}^n$ that has support size at most $n$. This implies that the set of vectors in $\mbox{{\sc Supp}}(D_2^{enc})\setminus \mathcal{U}$ is the encoding of at most $n$ vectors in $\{0,1\}^n$. So, from the definition of $\cP_{\mathrm{Gap}}$ and $D_3$, it is clear that $D_3 \in \cP_{\mathrm{Gap}}$.

Now we show that the Earth Mover Distance between the distributions $D_3$ and $D_2^{enc}$ is not large.

\begin{cl}\label{lem:dd1emd}
The Earth Mover Distance between $D_2^{enc}$ and $D_3$ is at most $6\alpha$.
\end{cl}

\begin{proof}
We will bound the Earth Mover Distance between $D_{2}^{enc}$ and $D_3$ in terms of the variation distance between them as follows:
\begin{eqnarray}
    d_{EM}(D_{2}^{enc},D_3) &\leq& \frac{1}{2} \cdot \sum_{\mathbf{X} \in \{0,1\}^N}|D_{2}^{enc}(\mathbf{X}) - D_3(\mathbf{X})| \nonumber \\ &=& \frac{1}{2} \cdot \sum_{\mathbf{X} \in \mathcal{V'}}|D_{2}^{enc}(\mathbf{X}) - D_3(\mathbf{X})| + \frac{1}{2} \cdot \sum_{\mathbf{X} \in \{0,1\}^N \setminus \mathcal{V'}}|D_{2}^{enc}(\mathbf{X}) - D_3(\mathbf{X})| \label{eqn:emddd1}.
\end{eqnarray}


Let us bound the first term as follows:
\begin{eqnarray*}
    \sum_{\mathbf{X} \in \mathcal{V'}}|D_{2}^{enc}(\mathbf{X}) - D_3(\mathbf{X})|
    &=& \sum_{\mathbf{X} \in \mathcal{V'}}| (1- 3 \alpha) \frac{D_{2}^{enc}(\mathbf{X})}{D_{2}^{enc}(\mathcal{V'})} - D_{2}^{enc}(\mathbf{X})| \\ &=& \sum_{\mathbf{X} \in \mathcal{V'}} \frac{D_{2}^{enc}(\mathbf{X})}{D_{2}^{enc}(\mathcal{V'})} |(1- 3 \alpha) - D_{2}^{enc}(\mathcal{V'})| \\ &=& \sum_{\mathbf{X} \in \mathcal{V'}} \frac{D_{2}^{enc}(\mathbf{X})}{D_{2}^{enc}(\mathcal{V'})} |3 \alpha - (1- D_{2}^{enc}(\mathcal{V'}))| \\ &\leq & \sum_{\mathbf{X} \in \mathcal{V'}} 3 \alpha \frac{D_{2}^{enc}(\mathbf{X})}{D_{2}^{enc}(\mathcal{V'})} \leq 3 \alpha. ~~~\left(\because D_{2}^{enc}(\mathcal{V'})\geq 1- 5 \alpha, \mbox{ Observation}~\ref{obs:v-dash-1}\right)
\end{eqnarray*}

From Observation~\ref{obs:v-dash}, $D_{2}^{enc}(\mathcal{U}) \leq 5 \alpha$. From the definition of $D_3$, $D_3(\mathcal{U}) = 3 \alpha$, we have
   $$\sum_{\mathbf{X} \in \{0,1\}^N \setminus \mathcal{V'}}|D_{2}^{enc}(\mathbf{X}) - D_3(\mathbf{X})| \leq 8 \alpha.$$
Following Equation~\ref{eqn:emddd1}, we conclude that $d_{EM}(D_{2}^{enc},D_3) \leq 6 \alpha$, which completes the proof.
\end{proof}

Since $D_3 \in \cP_{\mathrm{Gap}}$, and $d_{EM}(D_2^{enc}, D_3) \leq 6 \alpha$, we conclude that $D_2^{enc}$ is $6 \alpha$-close to $\cP_{\mathrm{Gap}}$.
\end{proof}


\begin{proof}[Proof of Lemma~\ref{cl:emd-far}]


Recall that the EMD between $D_2^{enc}$ and $D_1$ is the solution to the following LP:
\begin{eqnarray*}
    &&\mbox{Minimize}\quad{\sum_{\mathbf{X},\mathbf{Y} \in \{0,1\}^N} f_{\mathbf{X}\mathbf{Y}} d_H(\mathbf{X},\mathbf{Y})}\\
    &&\mbox{Subject to}\quad\sum_{\mathbf{Y} \in \{0,1\}^N} f_{\mathbf{X}\mathbf{Y}} = D_2^{enc}(\mathbf{X})~\forall \mathbf{X} \in \{0,1\}^N \; \mbox{and} \;
    \sum_{\mathbf{X} \in \{0,1\}^N} f_{\mathbf{X}\mathbf{Y}} = D_1(\mathbf{Y})~\forall \mathbf{Y} \in \{0,1\}^N.
\end{eqnarray*}

Let $f'_{\mathbf{x} \mathbf{y}}$ be the flow realizing the EMD between $D_2$ and $D_1^{dec}$. Using $f'$, we now construct a new flow $f^\star$ between $D_2^{enc}$ and $D_1$ as follows:

\begin{enumerate}
    \item[(i)]  For vectors $\mathbf{X}, \mathbf{Y} \in \mathcal{U}$,
    
\begin{enumerate}
    \item[(a)] If $\mathbf{X} \neq \mathbf{Y}$, then set $f^\star_{\mathbf{X} \mathbf{Y}}=0$. 

    \item[(b)] If $\mathbf{X} = \mathbf{Y}$, then set $f^\star_{\mathbf{X} \mathbf{Y}} = D_2^{enc}(\mathbf{X})=D_1(\mathbf{Y})$.
    \end{enumerate}
    
    \item[(ii)] For two vectors $\mathbf{X}, \mathbf{Y} \in \mathcal{V}$, we take the vectors $\mathbf{x}, \mathbf{y} \in \{0,1\}^n$ such that $\mathbf{X}, \mathbf{Y} \in \{0,1\}^N$ are their
valid encodings (by construction, if $\mathbf{X}$ and $\mathbf{Y}$ are in the support of $D_2$ and $D_1^{enc}$ respectively, such vectors $\mathbf{x}, \mathbf{y}$ exist), and vectors $\mathbf{z}_1,\mathbf{z}_2$ such that $\mathbf{X}_{\pi}=\fen(\mathbf{z}_1,\mathbf{x})$ and $\mathbf{Y}_{\pi}=\fen(\mathbf{z}_2, \mathbf{y})$. Now we set the flow as follows:

\begin{enumerate}
    \item[(a)] If $\mathbf{z}_1 \neq \mathbf{z}_2$, then set $f^\star_{\mathbf{X} \mathbf{Y}} = 0$.

    \item[(b)] If $\mathbf{z}_1 = \mathbf{z}_2$, then set $f^\star_{\mathbf{X} \mathbf{Y}} =  f'_{\mathbf{x} \mathbf{y}} \cdot \frac{D_1(\mathbf{Y})}{D_1^{dec}(\mathbf{y})}$.
\end{enumerate}

\item[(iii)] If one of $\mathbf{X}$ and $\mathbf{Y}$ is in $\mathcal{U}$ and the other one is in $\mathcal{V}$, then $f^\star_{\mathbf{X}\mathbf{Y}}=0$.

\end{enumerate}

We first argue that the flow $f^*_{\mathbf{X} \mathbf{Y}}$ constructed as above is indeed a valid flow, that is, we have:
$ \sum\limits_{\mathbf{Y} \in \{0,1\}^N} f^\star_{\mathbf{X} \mathbf{Y}} = D_2^{enc}(\mathbf{X})$ and $ \sum\limits_{\mathbf{X} \in \{0,1\}^N} f^\star_{\mathbf{X} \mathbf{Y}} = D_1(\mathbf{Y})$. 


To prove $\sum\limits_{\mathbf{Y} \in \{0,1\}^N} f_{\mathbf{X} \mathbf{Y}} = D_2^{enc}(\mathbf{X})$, first observe that it holds when $\mathbf{X} \in \mathcal{U}$ from (i) and (iii) in the description of $f^{\star}_{\mathbf{X} \mathbf{Y}}$. Now consider the case where $\mathbf{X} \in \mathcal{V}$. Assume $\mathbf{X}_{\pi}=\fen\left(\mathbf{z},\mathbf{x}\right)$, where $\mathbf{z} \in [n]^m$ and $\mathbf{x} \in \{0,1\}^n$. So, from (ii) in the description of $f^\star_{\mathbf{X}\mathbf{Y}}$, we have
$$\sum\limits_{\mathbf{Y} \in \{0,1\}^N} f^\star_{\mathbf{X} \mathbf{Y}}=\sum\limits_{\mathbf{y} \in \{0,1\}^n} f^\star_{\mathbf{X} \mathrm{FE}(\mathbf{z},\mathbf{y})_{\pi^{-1}}} = \sum_{\mathbf{y}\in\{0,1\}^n}f'_{\mathbf{x}\mathbf{y}}\frac{D_1(\mathrm{FE}(\mathbf{z},\mathbf{y})_{\pi^{-1}})}{D_1^{dec}(\mathbf{y})}=
D_2^{enc}(\mathbf{X}).$$



For $ \sum\limits_{\mathbf{X} \in \{0,1\}^N} f^\star_{\mathbf{X} \mathbf{Y}} = D_1(\mathbf{Y})$, consider $\mathbf{Y} \in \mathcal{V}$ for which $\mathbf{Y}_{\pi} = \mathrm{FE}(\mathbf{z}, \mathbf{y})$ for some $\mathbf{z} \in [n]^m$. Then we have the following:
$$\sum_{\mathbf{X}\in\{0,1\}^N}f_{\mathbf{X}\mathbf{Y}}^{\star}=\sum_{\mathbf{x}\in\{0,1\}^n}f'_{\mathbf{x}\mathbf{y}}\frac{D_1(\mathrm{FE}(\mathbf{z},\mathbf{y})_{\pi^{-1}})}{D_1^{dec}(\mathbf{y})}=D_1^{dec}(\mathbf{y})\frac{D_1(\mathrm{FE}(\mathbf{z},\mathbf{y})_{\pi^{-1}})}{D_1^{dec}(\mathbf{y})}=D_1(\mathbf{Y}).$$

In the above, we have used the fact that $f'_{\mathbf{x} \mathbf{y}}$ is a valid flow from $D_2$ to $D_1^{dec}$.\\

Now, to bound EMD between $D_2^{enc}$ and $D_1$, let us bound the sum $\sum\limits_{\mathbf{X},\mathbf{Y} \in \{0,1\}^N} f^\star_{\mathbf{X}\mathbf{Y}} d_H(\mathbf{X},\mathbf{Y})$.
\begin{eqnarray*}
&&\sum\limits_{\mathbf{X},\mathbf{Y} \in \{0,1\}^N} f^\star_{\mathbf{X}\mathbf{Y}} d_H(\mathbf{X},\mathbf{Y})\\
&=&\sum\limits_{\mathbf{X},\mathbf{Y} \in \mathcal{V}} f^\star_{\mathbf{X}\mathbf{Y}} d_H(\mathbf{X},\mathbf{Y})~~~\left(\mbox{From (i) and (iii) in the description of $f^\star$}\right)\\
&=& \sum\limits_{\mathbf{x},\mathbf{y} \in \{0,1\}^n}\sum\limits_{\mathbf{z}\in [n]^m} f^\star_{\mathbf{\fen(\mathbf{z},\mathbf{x})_{\pi^{-1}}}\fen(\mathbf{z},\mathbf{y})_{\pi^{-1}}}\cdot d_H({\fen(\mathbf{z},\mathbf{x})_{\pi^{-1}}},\fen(\mathbf{z},\mathbf{y})_{\pi^{-1}})\\ &&~~~~~~~~~~~~~~~~~~~~~~~~~~~~~~~~~~~~~~~~~~~~~~~~~~~~~~~~~~~~~~~~~~\left(\mbox{From (ii) in the description of $f^\star$}\right)\\
&\leq&\sum\limits_{\mathbf{x},\mathbf{y} \in \{0,1\}^n}\sum\limits_{\mathbf{z}\in [n]^m} f^\star_{\mathbf{\fen(\mathbf{z},\mathbf{x})_{\pi^{-1}}}\fen(\mathbf{z},\mathbf{y})_{\pi^{-1}}}\cdot d_H(\mathbf{x},\mathbf{y})~~~\left(\mbox{Observation~\ref{obs:FEproperties} (iii)}\right)\\
&=& \sum\limits_{\mathbf{x},\mathbf{y} \in \{0,1\}^n}\sum\limits_{\mathbf{z}\in [n]^m} f'_{\mathbf{x}\mathbf{y}}\frac{D_1(\fen(\mathbf{z},\mathbf{y})_{\pi^{-1}})}{D_1^{dec}(\mathbf{y})}\cdot d_H(\mathbf{x},\mathbf{y})~~~\left(\mbox{From (ii) in the description of $f^\star$}\right)\\
&=& \sum\limits_{\mathbf{x},\mathbf{y} \in \{0,1\}^n}\left(f'_{\mathbf{x}\mathbf{y}}d_H(\mathbf{x},\mathbf{y})\cdot  \sum\limits_{\mathbf{z}\in [n]^m} \frac{D_1(\fen(\mathbf{z},\mathbf{y})_{\pi^{-1}})}{D_1^{dec}(\mathbf{y})}\right)\\
&=& D_1(\mathcal{V}') \sum\limits_{\mathbf{x},\mathbf{y} \in \{0,1\}^n} f'_{\mathbf{x}\mathbf{y}}d_H(\mathbf{x},\mathbf{y})~~~(\mbox{By Equation}~\eqref{def:enc-dist})\\
&\leq& \sum\limits_{\mathbf{x},\mathbf{y} \in \{0,1\}^n} f'_{\mathbf{x}\mathbf{y}}d_H(\mathbf{x},\mathbf{y}) \leq \frac{\eps}{3}.
\end{eqnarray*}

The last inequality follows from the fact that $f'$ realizes the assumed EMD between $D_1$ and $D^{dec}_2$.
\end{proof}
\color{black}
\color{black}

%% file: quadratic_nonadaptive_lb_new_3.tex
\color{black}
\subsection{Near-quadratic lower bound for non-adaptive testers for testing $\cP_{\mathrm{Gap}}$}\label{sec:quadratic_nonadapt_lb}

\begin{lem}[{\bf Lower bound on non-adaptive testers}]\label{theo:nonadaptivelb}
Given sample and query access to an unknown distribution $D$, in order to distinguish whether $D$ satisfies $\cP_{\mathrm{Gap}}$ or is $\eps$-far from satisfying it, any non-adaptive tester must perform $\widetilde{\Omega}(n^{2})$ queries to the samples obtained from $D$, for some $\eps \in (0,1)$.
\end{lem}

To prove the above lemma, we will construct two hard distributions over distributions, $D_{yes}$ which is supported over $\cP_{\mathrm{Gap}}$, and $D_{no}$ which is supported over distributions far from $\cP_{\mathrm{Gap}}$, where to distinguish them, any non-adaptive tester must perform $\widetilde{\Omega}(n^{2})$ queries. Recall from Theorem~\ref{theo:valiantlb} that $D_{yes}^{\mathrm{Supp}}$ and $D_{no}^{\mathrm{Supp}}$ are two distributions defined over distributions over $\{1, \ldots, 2n\}$, where $D_{yes}^{\mathrm{Supp}}$ provides distributions whose support sizes are $n$, and $D_{no}^{\mathrm{Supp}}$ provides distributions that are $\eta$-far from distributions whose support size is $(1 + 2 \eta)n$, for some constant $\eta \in (0, 1/8)$. We will use these two distributions to construct the hard distributions $D_{yes}$ and $D_{no}$ for the property $\cP_{\mbox{Gap}}$.

\paragraph*{The hard distributions $D_{yes}$ and $D_{no}$:}
We describe the distributions $D_{yes}$ and $D_{no}$ over distributions over $\{0,1\}^N$ such that $D_{yes}$ is supported over $\cP_{\mathrm{Gap}}$ and $D_{no}$ is supported over distributions that are $\zeta^2 \cdot \eta/5$-far from $\cP_{\mathrm{Gap}}$. In what follows, we describe a distribution $D$ ($D=D_{yes}$ or $D=D_{no}$) with $D^{\mathrm{Supp}}$ as parameter, where $D^{\mathrm{Supp}}$ is a distribution defined over distributions over $[2n]$. In particular, $D^{\mathrm{Supp}}$ is either $D_{yes}^{\mathrm{Supp}}$ or $D_{no}^{\mathrm{Supp}}$, where $D=D_{yes}$ when $D^{\mathrm{Supp}}=D_{yes}^{\mathrm{Supp}}$, or $D=D_{no}$ when $D^{\mathrm{Supp}}=D_{no}^{\mathrm{Supp}}$. To generate $D$, we first construct a distribution over distributions $D^0$ as follows. We denote by $\widehat{D}$ the distribution over $\{0,1\}^N$ that we draw according to $D^0$.



\begin{description}
\item[(i)] Set $\widehat{D}(U)=\alpha$, where $\mathbf{U}= \mathbf{1}\mathbf{0}^{N-1}$ is the indicator vector for the index $1$.

\item[(ii)] Take a set of vectors $\cS= \{\mathbf{V}_1, \ldots, \mathbf{V}_b\}$ in $\{0,1\}^N$ such that for every $i\in [b]$, the $i$-th vector $\mathbf{V}_i$ is of the form $\mathbf{1}^{i+1}\mathbf{0}^{N-1-i}$. Set $\widehat{D}(\mathbf{V}_i) = \alpha/b$ for every $i \in [b]$.
     
\item[(iii)] Take another set of vectors $\cT= \{\mathbf{W}_0, \ldots, \mathbf{W}_{\lceil \log kn \rceil -1}\}$ (disjoint from $\cS$) in $\{0,1\}^N$ such that for every $\mathbf{W}_i \in \cT$, $\mathbf{W}_i$ is of the form $\mathbf{0}(b(i))(\mathbf{0}^{2^i}\mathbf{1}^{2^i})^{kn/2^{i+1}}$, where $b(i)$ denotes the length $b$ binary representation of $i$~\footnote{If $kn/2^{i+1}$ is not an integer, we trim the rightmost copy of $\mathbf{0}^{2^i}\mathbf{1}^{2^i}$ so that the total length of ``$(\mathbf{0}^{2^i} \mathbf{1}^{2^i})^{kn/2^{i+1}}$'' is exactly $kn$.}. Set $\widehat{D}(\mathbf{W}_i)=\alpha/|\cT|$ for every $\mathbf{W}_i \in \cT$.


\item[(iv)] Take a set of vectors $\cY \subseteq \{0, 1\}^n$ such that $\size{\cY}= 2n$, and for any two vectors $\mathbf{y}_i, \mathbf{y}_j \in \cY$, $i \neq j$, $\delta_H(\mathbf{y}_i, \mathbf{y}_j ) \geq n/3$. Also, draw a distribution $\widetilde{D}$ over $[2n]$ according to $D^{\mathrm{Supp}}$.

\item[(v)] Define $\widehat{D}(\mbox{{\sc FE}}(\mathbf{z}, \mathbf{y}_i)) = (1 - 3 \alpha)\widetilde{D}(i)/n^m$ for every $i \in [2n]$ and $\mathbf{z} \in [n]^m$, where $\mbox{{\sc FE}}:[n]^m \times \{0, 1\}^n \rightarrow \{0, 1\}^N$ is the encoding function from Definition~\ref{defi:encoding}.




    \item[(vi)] For all other remaining vectors that are not assigned probability mass in the above description, set their probabilities to $0$.

\end{description}

We define $D$ as the process of drawing a distribution $\widehat{D}$ according to $D^0$, and permuting it using a uniformly random permutation $\pi: [N] \rightarrow [N]$.

\begin{rem}[{\bf Intuition behind the above hard distributions}]

Unlike our adaptive algorithm to test $\cP_{\mathrm{Gap}}$ (Algorithm~\ref{algo:adaptive} in Subsection~\ref{sec:quadratic_adaptive_ub}), we can not determine the permutation $\pi$ first, and then perform queries depending on the permutation $\pi$. When the permutation $\pi$ is not known, even if we obtain a sample $\mathbf{X}$ and know that it is equal to $\mbox{\sc FE}(\mathbf{z},\mathbf{x})_{\pi^{-1}}$ for some $\mathbf{x}\in \{0,1\}^n$ and $\mathbf{z} \in [n]^m$, we can not even decode a single bit of $\mathbf{x}$, unless we query too many of the indices of $\mathbf{X}$. This follows from the properties of our encodings functions $\mbox{\sc SE}$ and $\mbox{\sc GE}$, used to construct $\mbox{\sc FE}$ (see Lemma~\ref{lem:restrictionuniform}), which ``hides'' $\mathbf{x}$ inside $\mathbf{X}$. Intuitively, this says that we have to query a quasilinear number of the coordinates of the sample. Since the support estimation problem admits a sample complexity lower bound of $\Omega(n/\log n)$, the non-adaptive query complexity of $\widetilde{\Omega}(n^2)$ follows for non-adaptive algorithms. We will formalize this intuition below.
\end{rem}

We will start with the following simple observation.

\begin{obs}
The distribution $D_{yes}$ is supported over $\cP_{\mathrm{Gap}}$.
\end{obs}

\begin{proof}
From the construction of $D_{yes}$, which is constructed by encoding the elements of the support of the distribution $D_{yes}$ drawn from $D_{yes}^{\mathrm{Supp}}$, it is clear that $D_{yes} \in \cP_{\mathrm{Gap}}$.
\end{proof}

Now we show that the distribution $D_{no}$ is supported over distributions that are far from the property $\cP_{\mathrm{Gap}}$.

\begin{lem}[{\bf Farness lemma}]\label{lem:farness}
$D_{no}$ is supported over distributions that are $\zeta^2 \cdot \eta/5$-far from $\cP_{\mathrm{Gap}}$.

\end{lem}

Before directly proceeding to the proof, let us first prove an additional lemma which will be used in the proof of Lemma~\ref{lem:farness}.

\begin{lem}\label{cl:encodevecfar}
    For any two distinct vectors $\mathbf{X}_1$ and $\mathbf{X}_2$ where $\mathbf{X}_{1,\pi}$, $\mathbf{X}_{2,\pi} \in \mathrm{Supp} (\widehat{D}) \setminus (\{\mathbf{U}\} \cup \cS \cup \cT)$ for $\widehat{D} \in \mathrm{Supp}(D_{no})$, and $\pi$ is the permutation for which $\widehat{D}_{\pi} \in D^0_{no}$, we have $\delta_H(\mathbf{X}_1, \mathbf{X}_2) \geq \zeta^2 \cdot N/2$.
\end{lem}

\begin{proof}
We will use the properties of the function $\mbox{{\sc FE}}$ as mentioned in Observation~\ref{obs:FEproperties}. Recall that for a string $\mathbf{z} \in [n]^m$, and a vector $\mathbf{x} = (\mathbf{x}_1, \ldots, \mathbf{x}_n) \in \{0,1\}^n$, we have $\mbox{{\sc FE}}(\mathbf{z},\mathbf{x})= \mathbf{0}(\mathbf{1}^b)\sen(\gen(\mathbf{z})_1,\mathbf{x}_1) \ldots \sen(\gen(\mathbf{z})_n,\mathbf{x}_n)$. Now we have the following two cases:

\begin{enumerate}
    \item[(a)] Suppose that for some vectors $\mathbf{x} \in \{0,1\}^n$, and $\mathbf{z}_1, \mathbf{z}_2 \in [n]^m$ such that $\mathbf{z}_1 \neq \mathbf{z}_2$, we have $\mathbf{X}_{1, \pi} = \mbox{{\sc FE}}(\mathbf{z}_1, \mathbf{x})$ and $\mathbf{X}_{2, \pi} = \mbox{{\sc FE}}(\mathbf{z}_2, \mathbf{x})$. Then following Property (i) of $\mbox{{\sc FE}}$ in Observation~\ref{obs:FEproperties}, we know that $\delta_H(\mathbf{X}_1, \mathbf{X}_2) \geq \zeta^2 \cdot N/2$ (noting that permuting the two vectors by the permutation $\pi$ preserves their pairwise distance). 
    

    \item[(b)] Suppose that for some vectors $\mathbf{z} \in [n]^m$, and $\mathbf{x}_1, \mathbf{x}_2 \in \{0,1\}^n$ such that $\mathbf{x}_1 \neq \mathbf{x}_2$, we have $\mathbf{X}_{1, \pi} = \fen(\mathbf{z}, \mathbf{x}_1)$ and $\mathbf{X}_{2,\pi} = \mbox{{\sc FE}}(\mathbf{z}, \mathbf{x}_2)$. Then following Property (ii) of $\mbox{{\sc FE}}$ in Observation~\ref{obs:FEproperties}, we know that $\delta_H(\mathbf{X}_1, \mathbf{X}_2) \geq \zeta \cdot \delta_H(\mathbf{x}_1, \mathbf{x}_2)$. From the choice of the vectors $\mathbf{y}_1, \ldots , \mathbf{y}_{2n}$, we know that $\delta_H(\mathbf{x}_1, \mathbf{x}_2) \geq n/3$. Thus, we can say that in this case $\delta_H(\mathbf{X}_1, \mathbf{X}_2) \geq \zeta \cdot nk/3 > \zeta^2 \cdot N/2$ (recalling that $\zeta<1/2$).

\end{enumerate}

Combining the above, we conclude that $\delta_H(\mathbf{X}_1, \mathbf{X}_2) \geq \zeta^2 \cdot N/2$, for any two distinct vectors $\mathbf{X}_1, \mathbf{X}_2$ as above.
\end{proof}



\begin{proof}[Proof of Lemma~\ref{lem:farness}]
Suppose that $\widehat{D} \in \mbox{Supp}(D_{no})$, and $\pi$ is the permutation for which $\widehat{D}_{\pi} \in \mbox{Supp}(D^0_{no})$. We will bound $d_{EM}(\widehat{D}, \cP_{\mathrm{Gap}})$. Let us denote the distribution $D_Y \in \cP_{\mathrm{Gap}}$ that is closest to $\widehat{D}$, where $\pi_Y$ is the permutation for which $D_{Y,\pi_Y} \in \cP^0_{\mathrm{Gap}}$. Let us first define a new distribution $\widetilde{D}_Y$ over $\{0,1\}^N$ as follows:
\[ \widetilde{D}_Y (\mathbf{X})=  \left\{
\begin{array}{ll}       
      \frac{1}{(1-3 \alpha)} D_Y(\mathbf{X})& \mathbf{X}_{\pi_Y} \notin (\{\mathbf{U}\} \cup \cS \cup \cT)  \\    
      0 & \mbox{otherwise}
\end{array} 
\right. \]
Similarly, we also define another distribution $\widetilde{D}$ from $\widehat{D}$, using $\pi$ instead of $\pi_Y$.

Now we have the following claim that bounds the distance between $\widetilde{D}_Y$ and $\widetilde{D}$.

\begin{cl}\label{obs:distencode}
    $d_{EM}(\widetilde{D}, \widetilde{D}_Y) \geq \zeta^2 \cdot \eta/4$.
\end{cl}

\begin{proof}
Following the definition of the property $\cP_{\mathrm{Gap}}$, we know that $\mbox{Supp}(\widetilde{D}_Y)$ consists of possible encodings of $n$ distinct vectors from $\{0,1\}^n$, and there are at most $n^m$ valid encodings of every such vector (as per the number of possible vectors $\mathbf{z} \in [n]^m$ that are given as input to $\gen$). This implies that the size of the support of the distribution $\widetilde{D}_Y$ is at most $n^{m+1}$.



Since any distribution in the support of $D_{no}^{\mathrm{Supp}}$ has support size at least $(1 + 2\eta)n$, following a similar argument as above, we infer that the size of the support of $\widetilde{D}$ is at least $(1+2\eta)n^{m+1}$. Moreover, by Lemma~\ref{cl:encodevecfar},  we know that any pair of vectors there has distance at least $\zeta^2/2$ (in relative distance). Also, as any vector in the support of any distribution in the support of $D_{no}^{\mathrm{Supp}}$ has probability mass that is multiple of $1/2n$, we infer that every vector in the support of $\widetilde{D}$ has probability mass at least $n^{-m-1}/2$ (as per Item (v) in the definition of $D^0$).

Summing up, we obtain that there are at least $2\eta \cdot n^{m+1}$ many vectors in $\mbox{Supp}(\widetilde{D})$ that are $\zeta^2/4$-far (in relative distance) from any vector in $\mbox{Supp}(\widetilde{D}_Y)$, all of whose weights are at least $n^{-m-1}/2$~\footnote{By the triangle inequality, if we consider a Hamming ball of radius $\zeta^2/4$ around every vector in $\mbox{Supp}(\widetilde{D}_Y)$, there can be at most one vector from $\mbox{Supp}(\widetilde{D})$ inside the ball.}.  Thus, the Earth Mover Distance of $\widetilde{D}$ from $\widetilde{D}_Y$ is at least $\zeta^2 \cdot \eta/4$.
\end{proof}


Recall that we need to bound the distance between $\widehat{D}$ and $D_Y$. From Claim~\ref{obs:distencode}, we know that $d_{EM}(\widetilde{D}, \widetilde{D}_Y) \geq \zeta^2 \cdot \eta/4$, where the distributions $\widetilde{D}$ and $\widetilde{D}_Y$ are defined over the encoding vectors. From the definition of $\widetilde{D}$ and $\widetilde{D}_Y$ from $\widehat{D}$ and $D_Y$, we conclude that $d_{EM}(D_{Y}, \widehat{D}) = (1-3\alpha)d_{EM}(\widetilde{D},\widetilde{D}_Y) \geq \zeta^2 \cdot \eta/5$.
\end{proof}

Now we prove that the distributions $D_{yes}$ and $D_{no}$ remain indistinguishable to any non-adaptive tester, unless it performs $\widetilde{\Omega}(n^2)$ queries. We start with some definitions that will be required for the proof. Recall that $N= \Oh(n \log n)$.

\begin{defi}[{\bf Large and small query set}]
A set of indices $I \subseteq [N]$ is said to be a \emph{large} if $\size{I} > n/\log^{10} n$. Otherwise, $I$ is said to be a \emph{small}.
\end{defi}

Now we show that for a uniformly random permutation $\sigma$, and any $C_j$ as defined in the property $\cP_{\mathrm{Gap}}$, with high probability the size of the set of indices $\size{I \cap \sigma(C_j)}$ will be small, unless $I$ is a large query set.

\begin{obs}\label{obs:queryintersectCj}
Let $\sigma:[N] \rightarrow [N]$ be a uniformly random permutation, and $C_j$ correspond to a ``bit encoding set'' of size $k$ (as per the definition of $\cP_{\mathrm{Gap}}$) for an arbitrary $j \in [n]$. For a fixed small query set $I \subseteq [N]$, the probability that $|I \cap \sigma(C_j)|$ is at least $\zeta \cdot k$ is at most $1/n^{10}$.
\end{obs}

\begin{proof}
    Let us define a collection of binary random variables $ \langle X_i : i \in I \rangle$ such that the following holds:
    \[ X_i=  \left\{
\begin{array}{ll}       
      1& i \in \sigma(C_j)  \\    
      0& \mbox{otherwise}
\end{array} 
\right. \]

Then as $\sigma$ is a uniformly random permutation, $\pr(X_i=1) = \frac{\size{\sigma(C_j)}}{N} = \Oh(\frac{1}{n})$ for any $i \in [n]$. Now let us define another random variable $X= \sum_{i=1}^{n} X_i$. Noting that $X = |I \cap \sigma(C_j)|$, we obtain $\E[X]=\Oh(1/ \log^{10} n)$. By applying Hoeffding's bound for sampling without replacement (Lemma~\ref{lem:hoeffdingineq_without_replacement}), we can say that $\pr(X \geq \zeta \cdot k) \leq 1/n^{10}$. This completes the proof.
\end{proof}

Now let us define an event $\cE_{I,j}$ as follows:
\[ \cE_{I,j} := \mbox{The query set $I$ satisfies $|I \cap \sigma(C_j)| \leq \zeta \cdot k$.} \]

Now we are ready to prove that unless $\widetilde{\Omega}(n^2)$ queries are performed, no non-adaptive tester can distinguish $D_{yes}$ from $D_{no}$.

\begin{lem}[{\bf Indistinguishibility lemma}]
With probability at least $2/3$, in order to distinguish $D_{yes}$ from $D_{no}$, $\widetilde{\Omega}(n^2)$ queries are necessary for any non-adaptive tester.
\end{lem}

\begin{proof}
From our result on the adaptive $\eps$-tester for $\cP_{\mathrm{Gap}}$, we know that $\widetilde{\Oh}(n)$ queries are sufficient for adaptively testing $\cP_{\mathrm{Gap}}$. Without loss of generality, let us assume that the non-adaptive tester takes at most $n^2$ samples from the unknown distribution $D$ (since we can assume that at least one query is performed in every sample). As per the definition of a non-adaptive tester, assume that the samples taken are $\mathbf{X}_1, \ldots, \mathbf{X}_s$, and their respective query sets are $I_1, \ldots, I_s$ for some integer $s$.

Consider an event $\cE$ as follows:
$$\cE := \mbox{For every $\ell \in [s]$ for which $I_{\ell}$ is small and every $j \in [n]$, the event $\cE_{I_{\ell},j}$ occurs.}$$   

Since the non-adaptive tester takes at most $n^2$ samples, there can be at most $n^2$ samples for which a small set was queried, that is, $s \leq n^2$. Moreover, there are $n$ possible sets $C_j$ present in a sample. Using the union bound, along with Observation~\ref{obs:queryintersectCj}, we can say that the event $\cE$ holds with probability at least $1 - 1/n^7$. Given that the event $\cE$ holds, we will now show that the induced distributions of $D_{yes}$ and $D_{no}$ on small query sets are identical and independent of the samples with large query sets.

\begin{cl}\label{cl:largequeryset}
Assume that the event $\cE$ holds. Then a non-adaptive tester that uses at most $o(n/\log n)$ large query sets, can not distinguish $D_{yes}$ from $D_{no}$ with probability more than $1/4$.
\end{cl}

\begin{proof}
Since the distributions produced by $D_{yes}$ and $D_{no}$ are identical over the respective permutations of $(\{\mathbf{U}\} \cup \cS \cup \cT)$, it is sufficient to prove indistinguishability over the restrictions to the valid encodings of $\mathbf{y}_1, \ldots, \mathbf{y}_{2n}$ (as they appear in the definition of $D^0$). Furthermore, we argue that this claim holds even if for every large query set, the tester is provided with the entire vector that was sampled.

Given that the event $\cE$ holds, regardless of whether the distribution was produced by $D_{yes}$ or $D_{no}$, the restriction of the samples to the small queried sets are completely uniformly distributed, even when conditioned on the samples with large query sets (which are taken independently of them). Thus we may assume that all samples with small query sets are ignored by the tester, since the answers to these queries can be simulated without taking any samples at all.

Finally, we appeal to the construction of the hard distributions $D_{yes}$ and $D_{no}$ from $D_{yes}^{\mathrm{Supp}}$ and $D_{no}^{\mathrm{Supp}}$. By Theorem~\ref{theo:valiantlb}, the distance between these two distributions over the sample sequence is at most $1/4$, unless there were more than $o(n/\log n)$ samples with large sets. This completes the proof of the claim.
\end{proof}

Combining Claim~\ref{cl:largequeryset} with the above bound on the probability of the event $\cE$, we conclude that $\widetilde{\Omega}(n^2)$ queries are necessary for any non-adaptive tester to distinguish $D_{yes}$ from $D_{no}$ with probability at least $2/3$, that is, with a probability difference of at least $1/3$. This concludes the proof of the lemma.
\end{proof}

%% file: appendix.tex
\section{Some probability results} \label{sec:prelim_prob}


\begin{lem}[{\bf Multiplicative Chernoff bound}~\cite{DubhashiP09}]
\label{lem:cher_bound1}
Let $X_1, \ldots, X_n$ be independent random variables such that $X_i \in [0,1]$. For $X=\sum\limits_{i=1}^n X_i$ and $\mu=\E[X]$, the following holds for any $0\leq \delta \leq 1$.
$$ \pr(\size{X-\mu} \geq \delta\mu) \leq 2\exp{\left(-{\mu \delta^2}/{3}\right)}.$$
\end{lem}

\begin{lem}[{\bf Additive Chernoff bound}~\cite{DubhashiP09}]
\label{lem:cher_bound2}
Let $X_1, \ldots, X_n$ be independent random variables such that $X_i \in [0,1]$. For $X=\sum\limits_{i=1}^n X_i$ and $\mu_l \leq \E[X] \leq \mu_h$, the following hold for any $\delta >0$.
\begin{itemize}
\item[(i)] $\pr \left( X \geq \mu_h + \delta \right) \leq \exp{\left(-{2\delta^2}/{n}\right)}$.
\item[(ii)] $\pr \left( X \leq \mu_l - \delta \right) \leq \exp{\left(-{2\delta^2} /{ n}\right)}$.
\end{itemize}

\end{lem}

\begin{lem}[{\bf Hoeffding's Inequality}~\cite{DubhashiP09}] \label{lem:hoeffdingineq}

Let $X_1,\ldots,X_n$ be independent random variables such that $a_i \leq X_i \leq b_i$ and $X=\sum\limits_{i=1}^n X_i$. Then, for all $\delta >0$, 
$$ \pr\left(\size{X-\E[X]} \geq \delta\right) \leq  2\exp\left({-2\delta^2}/ {\sum\limits_{i=1}^{n}(b_i-a_i)^2}\right).$$

\end{lem}


\begin{lem}[{\bf Hoeffding's Inequality for sampling without replacement}~\cite{hoeffding1994probability}] \label{lem:hoeffdingineq_without_replacement}

Let $n$ and $m$ be two integers such that $1 \leq n \leq m$, and  $x_1, \ldots, x_m$ be real numbers,  with $a \leq x_i \leq b$ for every $i \in [m]$. Suppose that $I$ is a set that is drawn uniformly from all subsets of $[m]$ of size $n$, and let $X = \sum\limits_{i \in I} x_i$.
Then, for all $\delta >0$, 
$$ \pr\left(\size{X-\E[X]} \geq \delta\right) \leq  2\exp\left({-2\delta^2}/ {n \cdot (b-a)^2}\right).$$
\end{lem}


\color{black}

\noindent
Now let us consider the following observation which states that if the normalized Hamming distance between two vectors $\mathbf{X}$ and $\mathbf{Y}$ are small, the same also holds with high probability when $\mathbf{X}$ and $\mathbf{Y}$ are projected on a set of random indices $K$. A similar result also holds when the distance is large between the two vectors $\mathbf{X}$ and $\mathbf{Y}$.
\begin{obs}[{\bf Approximating-string-distances}]\label{cl:distchernoff}
For $\mathbf{U}, \mathbf{V} \in \{0,1\}^n$ and assume that $K \subseteq [n]$ is a set of indices chosen uniformly at random without replacement. Then 
the following holds with probability at least $1-e^{-~\Oh(\delta^2 \size{K})}$:
$$\size{{d_H} (\mathbf{U},\mathbf{V})- {d_H} (\mathbf{U} \mid_{K},\mathbf{V}\mid_{K})} \leq \delta. $$
    
\end{obs}

\begin{proof}
Follows from the fact that sampling without replacement is as good as sampling with replacement (Lemma~\ref{lem:hoeffdingineq_without_replacement}).
\end{proof}

\begin{lem}[{\bf Chernoff bound for bounded dependency}~\cite{Janson04}]
\label{lem:depend:high_exact_statement}
Let $X_1,\ldots,X_n$ be random variables such that $a_i \leq X_i \leq b_i$ and $X=\sum\limits_{i=1}^n X_i$. Let $\cD$ be the (directed) \emph{dependency} graph, where $V(\cD)=\{X_1,\ldots,X_n\}$ and $X_i$ is completely independent of all variables $X_j$  for which $(X_i,X_j)$ is not a directed edge. Then for any $\delta >0$, 
$$ \pr(\size{X-\E[X]} \geq \delta) \leq  2e^{-2\delta^2 / \chi^*(\cD)\sum\limits_{i=1}^{n}(b_i-a_i)^2}.$$
where $\chi^*(\cD)$ denotes the \emph{fractional chromatic number} of $\cD$.
\end{lem}

\begin{coro}[{\bf Corollary of Lemma~\ref{lem:depend:high_exact_statement}}]
\label{lem:depend:high_prob}
Let $X_1,\ldots,X_n$ be indicator random variables such that the dependency graph is a disjoint union of $n/k$ many $k$ size cliques. For $X=\sum\limits_{i=1}^n X_i$ and $\mu_l \leq \E[X] \leq \mu_h$, the followings hold for any $\delta >0$:
\begin{itemize}
\item[(i)] $\pr\left(X \geq \mu_h + \delta\right) \leq \exp{\left(\frac{-2\delta^2}{kn}\right)}$,
\item[(ii)] $\pr\left(X \leq \mu_\ell- \delta\right) \leq \exp\left(\frac{-2\delta^2}{kn}\right)$.
\end{itemize} 
\end{coro}

\begin{proof}
Follows from the fact that the dependency graph has chromatic number $k$, and the fractional chromatic number of a graph is at most the chromatic number of any graph.
\end{proof}